%% file: main.tex
\documentclass[a4paper,UKenglish,cleveref,autoref,thm-restate]{lipics-v2021}

\usepackage{booktabs}   %
\usepackage{subcaption} %

\nolinenumbers

\title{
  Generalised Multiparty Session Types with Crash-Stop Failures (Technical
  Report, \today)
}
\titlerunning{%
  Generalised Multiparty Session Types with Crash-Stop Failures (Tech Report)
}
\author{Adam D. Barwell}{Imperial College London}{%
a.barwell@imperial.ac.uk
}{%
https://orcid.org/0000-0003-1236-7160
}{}
\author{Alceste Scalas}{DTU Compute --- Technical University of Denmark}{%
alcsc@dtu.dk
}{%
https://orcid.org/0000-0002-1153-6164
}{}
\author{Nobuko Yoshida}{Imperial College London}{%
yoshida@doc.ic.ac.uk
}{%
https://orcid.org/0000-0002-3925-8557
}{}
\author{Fangyi Zhou}{Imperial College London}{%
fangyi.zhou15@imperial.ac.uk
}{%
https://orcid.org/0000-0002-8973-0821
}{}

\authorrunning{A.D. Barwell, A. Scalas, N. Yoshida, F. Zhou} %

\Copyright{Adam D. Barwell, Alceste Scalas, Nobuko Yoshida, and Fangyi Zhou} %

\ccsdesc[500]{Theory of computation~Distributed computing models}
\ccsdesc[500]{Theory of computation~Process calculi}
\ccsdesc[500]{Software and its engineering~Model checking}

\keywords{Session Types, Concurrency, Failure Handling, Model Checking} %

\category{} %

\funding{
  Work supported by:
  EU Horizon 2020 project 830929; %
  EPSRC grants EP/T006544/1, EP/K011715/1, EP/K034413/1, EP/L00058X/1, EP/N027833/1, EP/N028201/1, EP/T006544/1, EP/T014709/1, EP/V000462/1, and NCSS/EPSRC VeTSS;
  Danmarks Industriens Fond %
  2020-0489.%
}%

\EventEditors{Bartek Klin, S\l{}awomir Lasota, and Anca Muscholl}
\EventNoEds{3}
\EventLongTitle{33rd International Conference on Concurrency Theory (CONCUR 2022)}
\EventShortTitle{CONCUR 2022}
\EventAcronym{CONCUR}
\EventYear{2022}
\EventDate{September 12--16, 2022}
\EventLocation{Warsaw, Poland}
\EventLogo{}
\SeriesVolume{243}
\ArticleNo{3}

\input{packages.tex}
\input{macros.tex}
\settoggle{techreport}{true}
\ifTR{%
  \hideLIPIcs
}

\relatedversion{
CONCUR 2022 paper: \url{https://doi.org/10.4230/LIPIcs.CONCUR.2022.35}
}

\begin{document}

\maketitle              %

\input{abstract.tex}

\input{1_intro.tex}
\input{3_calculus/session-calculus.tex}
\input{4_mpst_crash/mpst_with_crashes.tex}
\input{5_modelchecking/modelchecking.tex}

\input{6_related.tex}
\input{7_conclusion.tex}

\bibliography{main}

\newpage
\appendix
\input{proofs/structural-congruence.tex}
\input{proofs/subtyping.tex}
\input{proofs/all_examples.tex}

\input{proofs/results.tex}
\ifTR{
  \fontsize{10}{11}\selectfont
  \input{proofs/subtyping-properties.tex}
  \input{proofs/type-system-properties.tex}

  \input{proofs/subject-reduction.tex}
  \input{proofs/session-fidelity.tex}
}

\end{document}

%% file: packages.tex
\usepackage{etoolbox}%

\usepackage[noinline,index,%
  status=final %
]{fixme} %
\fxusetheme{colorsig}
\fxsetup{layout=pdfcmargin}
\FXRegisterAuthor{fxAS}{anfxAS}{AS}%
\FXRegisterAuthor{fxAB}{anfxAB}{AB}%
\FXRegisterAuthor{fxNY}{anfxNY}{NY}%
\FXRegisterAuthor{fxFZ}{anfxFZ}{FZ}

\usepackage{nicefrac}%

\usepackage{proof}%

\usepackage[most]{tcolorbox}%
\newtcolorbox{myframe}[1][]{
  enhanced,
  arc=0pt,
  outer arc=0pt,
  colback=white,
  boxrule=0.5pt,
  boxsep=0mm,
  left=1mm,
  right=1mm,
  top=0.5mm,
  bottom=0.5mm,
  #1
}

\lstdefinelanguage{Scribble}{%
  basicstyle=\footnotesize\ttfamily,
  stringstyle=\color{Blue},
  showstringspaces=false,
  keywords={nested,new,calls,and,as,at,by,catches,choice,continue,do,from,global,import,instantiates,interruptible,local,module,or,par,protocol,rec,role,sig,throws,to,type,with,int,aux},
  morestring=[b]",
  morestring=[b]',
  morecomment=[l][\color{greencomments}]{//},
}

\lstdefinelanguage{nuScr}{%
  basicstyle=\footnotesize\ttfamily,
  stringstyle=\color{Blue},
  showstringspaces=false,
  keywords={
    nested,new,calls,and,as,at,by,catches,choice,continue,do,from,global,import,instantiates,interruptible,local,module,or,par,protocol,rec,role,sig,throws,to,type,with,int,aux,
    safe
  },
  morestring=[b]",
  morestring=[b]',
  morecomment=[l][\color{greencomments}]{//},
  morecomment=[s][\color{magenta}]{(*}{*)},
}

\usepackage{amssymb}%
\usepackage{stmaryrd}
\usepackage{thmtools}%
\usepackage{arydshln}

\usepackage{xifthen}%
\usepackage{xspace}%

\usepackage{mathtools}%

\usepackage{balance}%

\usepackage{hyperref}
\hypersetup{hidelinks}
\usepackage[capitalise]{cleveref}

\usepackage{tikz}
\usetikzlibrary{automata, positioning, arrows, calc, fit}
\tikzset{
  >=stealth,
  node distance=2cm,
  every state/.style={thick, fill=gray!10},
  initial text=$ $,
}

\usepackage[export]{adjustbox}

\usepackage[pass]{geometry}%

\Crefname{section}{\S\!}{\S\!}%
\Crefname{subsection}{\S\!}{\S\!}%
\Crefname{subsubsection}{\S\!}{\S\!}%
\Crefname{appendix}{\S\!}{\S\!}
\Crefname{definition}{Def.\@}{Defs.\@}%
\Crefname{figure}{Fig.\@}{Figs.\@}%
\Crefname{example}{Ex.\@}{Exs.\@}%
\Crefname{corollary}{Cor.\@}{Cors.\@}%
\Crefname{theorem}{Thm.\@}{Thms.\@}%
\Crefname{proposition}{Prop.\@}{Props.\@}%

\crefname{section}{\S\!}{\S\!}%
\crefname{subsection}{\S\!}{\S\!}%
\crefname{subsubsection}{\S\!}{\S\!}%
\crefname{appendix}{\S\!}{\S\!}
\crefname{definition}{Def.\@}{Defs.\@}%
\crefname{figure}{Fig.\@}{Figs.\@}%
\crefname{example}{Ex.\@}{Exs.\@}%
\crefname{corollary}{Cor.\@}{Cors.\@}%
\crefname{theorem}{Thm.\@}{Thms.\@}%
\crefname{proposition}{Prop.\@}{Props.\@}%

\usepackage{cite} %

\newcolumntype{L}{>{$}l<{$}}
\newcolumntype{C}{>{$}c<{$}}
\newcolumntype{P}[1]{>{\centering\arraybackslash$}p{#1}<{$}}

%% file: macros.tex
\newif\ifdraft%
  \draftfalse%
  \drafttrue%

\newcommand{\ifempty}[3]{%
  \ifthenelse{\isempty{#1}}{#2}{#3}%
}%

\newtoggle{techreport}%
\newcommand{\ifTR}[1]{%
  \iftoggle{techreport}{#1}{}%
}%

\newtoggle{cruft}%
\newcommand{\dom}[1]{{\color{black}\operatorname{dom}\!\left({#1}\right)}}%
\newcommand{\fv}[1]{\operatorname{fv}\!\left({#1}\right)}%
\newcommand{\fc}[1]{\operatorname{fc}\!\left({#1}\right)}%
\newcommand{\dpv}[1]{\operatorname{dpv}\!\left({#1}\right)}%
\newcommand{\fpv}[1]{\operatorname{fpv}\!\left({#1}\right)}%
\newcommand{\notImpliedBy}{\mathrel{{\kern 1em}{\not{\kern -1em}\impliedby}}}%
\newcommand{\coloncolonequals}{\Coloneqq}%
\newcommand{\bnfdef}{\coloncolonequals}%
\newcommand{\bnfsep}{\mathbin{\;\big|\;}}%

\newcommand{\theTool}[0]{\texttt{mpstk}\xspace}

\def\aka{a.k.a.\@\xspace}%
\def\etc{\emph{etc.}\@\xspace}%
\def\eg{e.g.\@\xspace}%
\def\Eg{E.g.\@\xspace}%
\def\ie{i.e.\@\xspace}%
\def\wrt{w.r.t.\@\xspace}%
\def\formulae{formul\ae\xspace}
\def\Formulae{Formul\ae\xspace}

\newcommand{\exampleName}[1]{\ensuremath{\mathsf{#1}}\xspace}

\definecolor{ruleColor}{rgb}{0.1, 0.3, 0.1}%
\newcommand{\inferrule}[1]{{\color{ruleColor}\text{\upshape\textsc{\scriptsize[#1]}}}}%
\newcommand{\inference}[3][]{\infer[\ifempty{#1}{}{\inferrule{#1}}]{#3}{#2}}%
\newcommand{\cinference}[3][]{\infer=[\ifempty{#1}{}{\inferrule{#1}}]{#3}{#2}}%
\newcommand{\inferenceSingle}[2][]{{#2}\ifempty{#1}{}{\;\;\inferrule{#1}}}%

\newcommand{\setenum}[1]{\mathord{{\color{black}\left\{#1\right\}}}}%
\newcommand{\setcomp}[2]{\mathord{%
  {\color{black}\left\{{#1} \,\middle|\, {#2}\right\}}}}%

\newcommand{\predP}[1][]{\ifempty{#1}{\varphi}{\varphi_{#1}}}%
\newcommand{\predPApp}[2][]{\ifempty{#1}{\predP}{\predP[{#1}]}\!\left({#2}\right)}%

\newcommand{\bind}[2]{\nicefrac{#2}{#1}}%
\newcommand{\substenum}[1]{\mathord{\left\{{#1}\right\}}}%
\newcommand{\subst}[2]{\substenum{\bind{#1}{#2}}}%
\definecolor{hlColor}{rgb}{0.65, 1.0, 0.65}%

\newcommand{\highlightColour}{yellow!40}
\newcommand{\highlightText}[1]{%
  {\setlength{\fboxsep}{0pt}\colorbox{\highlightColour}{#1}}\xspace%
}%

\newcommand{\highlight}[2][\highlightColour]{\mathchoice%
  {\setlength{\fboxsep}{0pt}\colorbox{#1}{$\displaystyle#2$}}%
  {\setlength{\fboxsep}{0pt}\colorbox{#1}{$\textstyle#2$}}%
  {\setlength{\fboxsep}{0pt}\colorbox{#1}{$\scriptstyle#2$}}%
  {\setlength{\fboxsep}{0pt}\colorbox{#1}{$\scriptscriptstyle#2$}}}%

\newcommand{\lbbar}{\{\kern-0.2em|}
\newcommand{\rbbar}{|\kern-0.2em\}}

\definecolor{tyColorCustom}{rgb}{0.0, 0.0, 0.85}%
\newcommand{\tyCol}[1]{{\color{tyColorCustom}{#1}}}%
\newcommand{\tyFont}[1]{{#1}}%
\newcommand{\tyFmt}[1]{\tyCol{\tyFont{#1}}}%
\newcommand{\tyFontC}[1]{\operatorname{#1}}%
\newcommand{\tyFmtC}[1]{\tyCol{\tyFontC{#1}}}%

\newcommand{\tyGround}[1][]{\tyFmt{\ifempty{#1}{B}{B_{#1}}}}%
\newcommand{\tyGroundi}[1][]{\tyFmt{\ifempty{#1}{B'}{B'_{#1}}}}%
\newcommand{\tyBool}{\tyFmtC{bool}}%
\newcommand{\tyUnit}{\tyFmtC{unit}}%
\newcommand{\tyInt}{\tyFmtC{int}}%
\newcommand{\tyReal}{\tyFmtC{real}}%

\newcommand{\tyS}[1][]{\tyCol{\ifempty{#1}{\tyFont{S}}{\tyFont{S}_{#1}}}}%
\newcommand{\tySi}[1][]{\tyCol{\ifempty{#1}{\tyFont{S}'}{\tyFont{S}'_{#1}}}}%
\newcommand{\tyT}[1][]{\tyCol{\ifempty{#1}{\tyFont{T}}{\tyFont{T}_{#1}}}}%
\newcommand{\tyTi}[1][]{\tyCol{\ifempty{#1}{\tyFont{T}'}{\tyFont{T}'_{#1}}}}%
{\centerline{\bf --- Begin Copied From Previous Paper ---} \hrule}%
{\hrule \centerline{\bf --- End Copied From Previous Paper ---}}%

{\centerline{\bf --- Begin Discussion\ifempty{#1}{}{: {#1}} ---}
  \hrule\vspace{1mm}}%
{\hrule\vspace{1mm}\centerline{\bf --- End Discussion ---}}%

\definecolor{roleColor}{rgb}{0.5, 0.0, 0.0}%
\newcommand{\roleCol}[1]{{\color{roleColor}#1}}%
\newcommand{\roleSet}{\roleCol{\mathfrak{R}}}%
\newcommand{\roleFmt}[1]{\ensuremath{{\boldsymbol{\roleCol{\mathtt{#1}}}}}\xspace}%

\newcommand{\roleP}[1][]{%
  \ifempty{#1}{{\color{roleColor}\roleFmt{p}}}{{\color{roleColor}\roleFmt{p}_{#1}}}%
}%
\newcommand{\rolePi}[1][]{%
  \ifempty{#1}{{\color{roleColor}\roleFmt{p}'}}{{\color{roleColor}\roleFmt{p}'_{#1}}}%
}%
\newcommand{\roleQ}[1][]{%
  \ifempty{#1}{{\color{roleColor}\roleFmt{q}}}{{\color{roleColor}\roleFmt{q}_{#1}}}%
}%
\newcommand{\roleQi}[1][]{%
  \ifempty{#1}{{\color{roleColor}\roleFmt{q}'}}{{\color{roleColor}\roleFmt{q}'_{#1}}}%
}%
\newcommand{\roleR}[1][]{%
  \ifempty{#1}{{\color{roleColor}\roleFmt{r}}}{{\color{roleColor}\roleFmt{r}_{\!#1}}}%
}%
\newcommand{\roleS}[1][]{%
  \ifempty{#1}{{\color{roleColor}\roleFmt{s}}}{{\color{roleColor}\roleFmt{s}_{\!#1}}}%
}%

\newcommand{\rolesR}[1][]{%
  \ifempty{#1}{{\color{roleColor}\roleFmt{\mathcal
        R}}}{{\color{roleColor}\roleFmt{\mathcal R}_{#1}}}
}
\newcommand{\rolesRi}[1][]{%
  \ifempty{#1}{{\color{roleColor}\roleFmt{\mathcal
        R'}}}{{\color{roleColor}\roleFmt{\mathcal R'}_{#1}}}
}
\newcommand{\rolesS}[1][]{%
  \ifempty{#1}{{\color{roleColor}\roleFmt{\mathcal
        R}}}{{\color{roleColor}\roleFmt{\mathcal R}_{#1}}}
}
\newcommand{\rolesSEmpty}[0]{\gtFmt{\emptyset}}

\definecolor{gtColor}{rgb}{0.43, 0.21, 0.1}%
\newcommand{\gtFmt}[1]{\ensuremath{{\color{gtColor}#1}}\xspace}%
\newcommand{\gtMsgFmt}[1]{\gtFmt{\labFmt{#1}}}%
\newcommand{\gtLab}[1][]{%
  \ifempty{#1}{\gtMsgFmt{m}}{{\color{gtColor}\gtMsgFmt{m}_{#1}}}%
}%
\newcommand{\labFmt}[2][]{\ensuremath{\ifempty{#1}{\mathtt{#2}}{\mathtt{#2}\textsubscript{#1}}}\xspace}%

\newcommand{\errKFmt}[1]{\stStopLab \stSeq {#1}}

\definecolor{stColor}{rgb}{0, 0, 0.9}%
\newcommand{\stFmt}[1]{\ensuremath{{\color{stColor}#1}}\xspace}%

\newcommand{\stTypeInt}{\stFmt{\operatorname{Int}}}%
\newcommand{\stTypeString}{\stFmt{\operatorname{Str}}}%

\newcommand{\stIn}[4]{\ifempty{#1}{}{\roleFmt{#1}}\stFmt{\&\left\{{#2}\ifempty{#3}{}{({#3})} \ifempty{#4}{}{\stSeq #4}\right\}}}%
\newcommand{\stInNB}[4]{\ifempty{#1}{}{\roleFmt{#1}}\stFmt{\&{#2}\ifempty{#3}{}{({#3})} \ifempty{#4}{}{\stSeq #4}}}%
\newcommand{\stOut}[3]{\ifempty{#1}{}{\roleFmt{#1}}\stFmt{\oplus{#2}\ifempty{#3}{}{({#3})}}}%

\newcommand{\stChoice}[2]{\stLabFmt{#1}\ifempty{#2}{}{\stFmt{({#2})}}}%

\newcommand{\stSeq}{\mathbin{\!\stFmt{.}\!}}%
\newcommand{\stIntC}{\mathbin{\stFmt{\oplus}}}%
\newcommand{\stIntSum}[3]{\roleFmt{#1}\stFmt{\oplus\!\left\{#3\right\}_{#2}}}%
\newcommand{\stExtC}{\mathbin{\stFmt{\&}}}%
\newcommand{\stExtSum}[3]{\roleFmt{#1}\stFmt{\&\!\left\{#3\right\}_{#2}}}%
\newcommand{\stExtSumErr}[4]{\roleFmt{#1}\stFmt{\&\!\left\{#3, \errKFmt{#4}\right\}\ifempty{#2}{}{_{#2}}}}%
\newcommand{\stRec}[2]{\stFmt{\mu{#1}.{#2}}}%
\newcommand{\stEnd}{\stFmt{\mathsf{end}}}%
\newcommand{\stStop}{\stFmt{\mathsf{stop}}}%

\newcommand{\stLabFmt}[1]{\stFmt{\labFmt{#1}}}%
\newcommand{\stLab}[1][]{%
  \ifempty{#1}{\stLabFmt{m}}{\stLabFmt{m}_{{\color{stColor}#1}}}
}%
\newcommand{\stLabi}[1][]{%
  \ifempty{#1}{\stLabFmt{m'}}{\stLabFmt{m'}_{{\color{stColor}#1}}}
}%

\newcommand{\stStopLab}[0]{\stLabFmt{\mathsf{crash}}}
\newcommand{\stCrashLab}[0]{\stLabFmt{\mathsf{crash}}}
\newcommand{\stLabOK}[0]{\stLabFmt{ok}}
\newcommand{\stLabKO}[0]{\stLabFmt{ko}}

\newcommand{\stS}[1][]{\stFmt{\ifempty{#1}{S}{S_{#1}}}}%
\newcommand{\stSi}[1][]{\stFmt{\ifempty{#1}{S'}{S'_{#1}}}}%

\newcommand{\stT}[1][]{\stFmt{\ifempty{#1}{T}{T_{#1}}}}%
\newcommand{\stTi}[1][]{\stFmt{\ifempty{#1}{T'}{T'_{#1}}}}%

\newcommand{\stU}[1][]{\stFmt{\ifempty{#1}{U}{U_{#1}}}}%
\newcommand{\stLabP}[1][]{\stFmt{\stLab[#1](\stS[#1])}}%
\newcommand{\stLabPi}[1][]{\stFmt{\stLabi[#1](\stSi[#1])}}%

\newcommand{\stRecVarBase}{\stFmt{\mathbf{t}}}%
\newcommand{\stRecVar}[1][]{\stFmt{\ifempty{#1}{\stRecVarBase}{\stRecVarBase_{#1}}}}%
\newcommand{\stRecVari}[1][]{\stFmt{\ifempty{#1}{\stRecVar'}{\stRecVar'_{#1}}}}%
\newcommand{\tyGroundSub}{\mathrel{{<}{:}}}
\newcommand{\stSub}{\mathrel{\stFmt{\leqslant}}}%
\newcommand{\stNotSub}{\mathrel{\stFmt{\not\leqslant}}}%

\newcommand{\ltsSendRecv}[4]{\mpChanRole{#1}{#2}\mpFmt{[\roleFmt{#3}]}\labFmt{#4}}

\newcommand{\ltsCrDe}[3]{\mpChanRole{#1}{#2}\stFmt{\mathord{\odot}}\roleFmt{#3}}
\newcommand{\ltsCrash}[2]{\mpChanRole{#1}{#2}\stFmt{\lightning}}
\newcommand{\ltsCrashSmall}[2]{\mpChanRole{#1}{#2}\stFmt{\lightning}}

\definecolor{mpColor}{rgb}{0, 0, 0}%
\newcommand{\mpFmt}[1]{{\color{mpColor}#1}}%

\newcommand{\mpLab}[1][]{%
  \mpFmt{\ifempty{#1}{\labFmt{m}}{{\labFmt{m}}_{\mathnormal #1}}}%
}%
\newcommand{\mpLabi}[1][]{%
  \mpFmt{\ifempty{#1}{\labFmt{m}'}{\labFmt{m}'_{\mathnormal #1}}}%
}%
\newcommand{\mpLabCrash}[0]{\mpFmt{\labFmt{\mathsf{crash}}}}
\newcommand{\mpLabFmt}[1]{\mpFmt{\labFmt{#1}}}%

\newcommand{\mpChanRole}[2]{\mpFmt{{#1}[{#2}]}}%

\newcommand{\mpNil}{\mpFmt{\mathbf{0}}}%
\newcommand{\mpSeq}{\mathbin{\mpFmt{\!.\!}}}%
\newcommand{\mpChoice}[3]{%
  \mpFmt{%
    \mpLabFmt{#1}\ifempty{#2}{}{({#2})}\ifempty{#3}{}{\vphantom{x}\mpSeq {#3}}%
  }%
}%
\newcommand{\mpChoiceNoBind}[3]{%
  \mpFmt{%
    \mpLabFmt{#1}\ifempty{#2}{}{\langle{#2}\rangle}\ifempty{#3}{}{\vphantom{x}\mpSeq {#3}}%
  }%
}%
\newcommand{\mpBranch}[7]{%
  \mpFmt{%
    {#1}[\roleFmt{#2}] \mathbin{\!\ifempty{\sum}{\sum}{\&}\!}%
    \ifempty{#3}{%
      \{\mpChoice{#4}{#5}{#6}\ifempty{#7}{}{,\;\mpChoice{\mpLabCrash}{}{#7}}\}_{#3}
    }{%
      \{\mpChoice{#4}{#5}{#6}\ifempty{#7}{}{,\;\mpChoice{\mpLabCrash}{}{#7}}\}_{#3}
    }%
  }%
}%
\newcommand{\mpBranchSingle}[5]{%
  \mpFmt{%
    {#1}[\roleFmt{#2}] \mathbin{\!\ifempty{\sum}{\sum}{\&}\!}%
    \mpChoice{#3}{#4}{#5}
  }%
}%
\newcommand{\mpSel}[5]{%
  \mpFmt{%
    {#1}[\roleFmt{#2}] \mathbin{\!\oplus\!}%
    \mpChoiceNoBind{#3}{#4}{#5}%
  }%
}%
\newcommand{\mpPar}{\mathbin{\mpFmt{\mid}}}%
\newcommand{\mpBigPar}[2]{\mathbin{\mpFmt{\Pi_{#1}}{#2}}}%
\newcommand{\mpRes}[2]{\mpFmt{\left(\mathbf{\nu}{#1}\right){#2}}}%
\newcommand{\mpJustDef}[3]{%
  \mpFmt{{#1}(#2) = {#3}}%
}%
\newcommand{\mpDef}[4]{%
  \mpFmt{\mathsf{def}\;\mpJustDef{#1}{#2}{#3}\;\mathsf{in}\;{#4}}%
}%
\newcommand{\mpDefAbbrev}[2]{%
  \mpFmt{\mathsf{def}\;{#1}\;\mathsf{in}\;{#2}}%
}%
\newcommand{\mpCall}[2]{\mpFmt{{#1}\!\left\langle{#2}\right\rangle}}%
\newcommand{\mpCallSmall}[2]{\mpFmt{{#1}\langle{#2}\rangle}}%
\newcommand{\mpErr}{\mpFmt{\boldsymbol{\mathsf{err}}}}%
\newcommand{\mpStop}[2]{\mpChanRole{#1}{#2}\mpFmt{\ensuremath{\lightning}}}
\newcommand{\mpStopDef}[0]{\mpStop{\mpS}{\roleP}}

\newcommand{\mpCtx}[1][]{\mpFmt{\ifempty{#1}{\mathbb{C}}{\mathbb{C}_{#1}}}}%
\newcommand{\mpCtxi}[1][]{\mpFmt{\ifempty{#1}{\mathbb{C}'}{\mathbb{C}'_{#1}}}}%
\newcommand{\mpCtxHole}{[\,]}%
\newcommand{\mpCtxApp}[2]{{#1}\!\left[{#2}\right]}%

\newcommand{\mpV}[1][]{\mpFmt{\ifempty{#1}{v}{v_{#1}}}}

\newcommand{\mpW}[1][]{\mpFmt{\ifempty{#1}{w}{w_{#1}}}}

\newcommand{\mpC}[1][]{\mpFmt{\ifempty{#1}{c}{c_{#1}}}}%
\newcommand{\mpD}[1][]{\mpFmt{\ifempty{#1}{d}{d_{#1}}}}%
\newcommand{\mpS}[1][]{\mpFmt{\ifempty{#1}{s}{s_{#1}}}}%
\newcommand{\mpSi}[1][]{\mpFmt{\ifempty{#1}{s'}{s'_{#1}}}}%
\newcommand{\mpX}[1][]{\mpFmt{\ifempty{#1}{X}{X_{#1}}}}%
\newcommand{\mpY}[1][]{\mpFmt{\ifempty{#1}{Y}{Y_{#1}}}}%
\newcommand{\mpP}[1][]{\mpFmt{\ifempty{#1}{P}{P_{#1}}}}%
\newcommand{\mpPi}[1][]{\mpFmt{\ifempty{#1}{P'}{P'_{#1}}}}%
\newcommand{\mpPii}[1][]{\mpFmt{\ifempty{#1}{P''}{P''_{#1}}}}%
\newcommand{\mpQ}[1][]{\mpFmt{\ifempty{#1}{Q}{Q_{#1}}}}%
\newcommand{\mpQi}[1][]{\mpFmt{\ifempty{#1}{Q'}{Q'_{#1}}}}%
\newcommand{\mpR}[1][]{\mpFmt{\ifempty{#1}{R}{R_{#1}}}}%
\newcommand{\mpDefD}[1][]{\mpFmt{\ifempty{#1}{D}{D_{#1}}}}%
\newcommand{\mpDefDi}[1][]{\mpFmt{\ifempty{#1}{D'}{D'_{#1}}}}%
\newcommand{\mpMove}{\to}%
\newcommand{\mpMoveLab}[1]{\xrightarrow{#1}}
\newcommand{\mpMoveMaybeCrash}[1][]{%
  \ifempty{#1}{%
    \mathrel{\mpMove}%
  }{%
    \mathrel{\mpMove_{\!\lightning \setminus {#1}}}%
  }
}%
\newcommand{\mpMoveMaybeCrashChecked}[1][]{%
  \ifempty{#1}{%
    \mathrel{\mpMoveLab{\checkmark}}%
  }{%
    \text{\fxASerror{FIX MACRO USE!}}
  }
}%
\newcommand{\mpMoveStar}{\mathrel{\mpMove{}^{\!\!\!*}}}%
\newcommand{\mpMovePlus}{\mathrel{\mpMove{}^{\!\!\!+}}}%
\newcommand{\mpMoveMaybeCrashCheckedStar}[1][]{%
  \ifempty{#1}{%
    \mathrel{{\mpMoveLab{\checkmark}}{}^{\!*}}%
  }{%
    \text{\fxASerror{FIX MACRO USE!}}
  }
}%
\newcommand{\mpMoveMaybeCrashCheckedPlus}[1][]{%
  \ifempty{#1}{%
    \mathrel{{\mpMoveLab{\checkmark}}{}^{\!+}}%
  }{%
    \text{\fxASerror{FIX MACRO USE!}}
  }
}%
\newcommand{\mpMoveP}[1]{\mpFmt{#1}\!{\mpMove}}%
\newcommand{\mpNotMoveP}[1]{\mpFmt{#1}\!\not{\!\!\mpMove}}%
\newcommand{\mpMoveCrash}{\mpMove}
\newcommand{\iruleMPRedComm}{R-$\oplus\&$}%
\newcommand{\iruleMPRedCommE}{R-$\lightning\mpLab$}%
\newcommand{\iruleMPRedCommEGround}{R-$\lightning\mpLab\tyGround$}%
\newcommand{\iruleMPRedCommD}{R-$\odot$}%
\newcommand{\iruleMPRedPar}{R-$\mpPar$}%
\newcommand{\iruleMPRedRes}{R-$\mpFmt{\mathbf{\nu}}$}%
\newcommand{\iruleMPRedCall}{R-$\mpX$}%
\newcommand{\iruleMPRedCongr}{R-$\equiv$}%
\newcommand{\iruleMPCtx}[1][]{R-Ctx\ifempty{#1}{}{#1}}
\newcommand{\iruleMPRedCtx}{\iruleMPCtx[]}%
\newcommand{\iruleMPRedCtxCrash}{\iruleMPCtx[$\lightning$\!]}%
\newcommand{\iruleMPErrLabel}{R-Err}%
\newcommand{\iruleMPCrash}[1][]{R-$\lightning\ifempty{#1}{}{#1}$}
\newcommand{\iruleMPCrashS}{\iruleMPCrash[\oplus]}
\newcommand{\iruleMPCrashR}{\iruleMPCrash[\&]}

\newcommand{\iruleSafeComm}{S-${\stIntC}{\stExtC}$}%
\newcommand{\iruleSafeCrash}{S-${\stFmt{\lightning}}{\stExtC}$}%
\newcommand{\iruleSafeMove}{S-$\stEnvMoveMaybeCrash$}%

\newcommand{\iruleStSubGround}{Sub-$\tyGround$}
\newcommand{\iruleStSubEnd}{Sub-$\stEnd$}
\newcommand{\iruleStSubStop}{Sub-$\stStop$}
\newcommand{\iruleStSubRecL}{Sub-$\stFmt{\mu}$L}
\newcommand{\iruleStSubRecR}{Sub-$\stFmt{\mu}$R}
\newcommand{\iruleStSubOut}{Sub-$\stFmt{\oplus}$}
\newcommand{\iruleStSubIn}{Sub-$\stFmt{\&}$}

\newcommand{\iruleMPEnd}{T-$\stEnvEndPred$}%
\newcommand{\iruleMPX}{T-$\mpFmt{X}$}%
\newcommand{\iruleMPGround}{T-$\tyGround$}
\newcommand{\iruleMPSub}{T-Sub}%
\newcommand{\iruleMPNil}{T-$\mpNil$}%
\newcommand{\iruleMPDef}{T-$\mpFmt{\mathsf{def}}$}%
\newcommand{\iruleMPCall}{T-Call}%
\newcommand{\iruleMPPar}{T-$\mpPar$}%
\newcommand{\iruleMPResProp}{T-$\mpFmt{\mathbf{\nu}}$}%
\newcommand{\iruleMPBranch}{T-$\mpFmt{\&}$}%
\newcommand{\iruleMPSel}{T-$\mpFmt{\oplus}$}%

\newcommand{\stStopSym}{\stFmt{\ensuremath{\lightning}}}
\newcommand{\iruleMPStop}{T-$\lightning$}%

\newcommand{\iruleTCtxOut}{$\stEnv$-$\stFmt{\oplus}$}%
\newcommand{\iruleTCtxIn}{$\stEnv$-$\stFmt{\&}$}%
\newcommand{\iruleTCtxCom}{$\stEnv$-$\stFmt{\oplus\&}$}%
\newcommand{\iruleTCtxRec}{$\stEnv$-$\mu$}%
\newcommand{\iruleTCtxCong}{$\stEnv$-$\stEnvComp$}%
\newcommand{\iruleTCtxCongBasic}{$\stEnv$-$\stEnvComp$$\tyGround$}%
\newcommand{\iruleTCtxCrash}{$\stEnv$-$\lightning$}
\newcommand{\iruleTCtxCrashDetect}{$\stEnv$-$\odot$}
\newcommand{\iruleTCtxSendToCrashed}{$\stEnv$-$\lightning\gtLab$}
\newcommand{\iruleTCtxCrashed}{$\stEnv$-$\stStop$}

\newcommand{\iruleCongParComm}[0]{C-Par}
\newcommand{\iruleCongParAssoc}[0]{C-Assoc}
\newcommand{\iruleCongParId}[0]{C-ParId}
\newcommand{\iruleCongResElim}[0]{C-ResElim}
\newcommand{\iruleCongResVar}[0]{C-ResVar}
\newcommand{\iruleCongResLift}[0]{C-ResLift}
\newcommand{\iruleCongStopElim}[0]{C-CrashElim}
\newcommand{\iruleCongDefElim}[0]{C-DefElim}
\newcommand{\iruleCongDefLift}[0]{C-DefLift}
\newcommand{\iruleCongDefParLift}[0]{C-DefParLift}
\newcommand{\iruleCongDefOrd}[0]{C-DefOrd}

\newcommand{\iruleFmlaSafe}{$\mu$-safe}
\newcommand{\iruleFmlaDF}[0]{$\mu$-df}
\newcommand{\iruleFmlaTerm}[0]{$\mu$-term}
\newcommand{\iruleFmlaNTerm}[0]{$\mu$-nterm}
\newcommand{\iruleFmlaLive}[0]{$\mu$-live}

\newcommand{\stEnv}[1][]{\stFmt{\ifempty{#1}{\Gamma}{\Gamma_{\!#1}}}}%
\newcommand{\stEnvi}[1][]{\stFmt{\ifempty{#1}{\Gamma'}{\Gamma'_{\!#1}}}}%
\newcommand{\stEnvii}[1][]{\stFmt{\ifempty{#1}{\Gamma''}{\Gamma''_{\!#1}}}}%
\newcommand{\stEnviii}[1][]{\stFmt{\ifempty{#1}{\Gamma'''}{\Gamma'''_{\!#1}}}}%
\newcommand{\stEnvEmpty}{\stFmt{\emptyset}}%
\newcommand{\stEnvMap}[2]{\stFmt{\mpFmt{#1}\mathbin{\!:\!}{#2}}}%
\newcommand{\stEnvComp}{\mathpunct{\stFmt{,}}}%
\newcommand{\stEnvApp}[2]{\stFmt{#1\!\left(\mpFmt{#2}\right)}}%

\newcommand{\stEnvMove}{\mathrel{\stFmt{\to}}}%
\newcommand{\stEnvMoveMaybeCrash}[1][]{%
  \ifempty{#1}{%
    \mathrel{\stFmt{\to_{\!\lightning}}}%
  }{%
    \mathrel{\stFmt{\to_{\!\lightning \setminus {#1}}}}%
  }
}%
\newcommand{\stEnvAnnotOutSym}{\stFmt{\oplus}}%
\newcommand{\stEnvAnnotInSym}{\stFmt{\&}}%
\newcommand{\stEnvAnnotGenericSym}[1][]{\stFmt{\ifempty{#1}{\alpha}{\alpha_{#1}}}}%
\newcommand{\stEnvMoveAnnot}[1]{\mathrel{\stFmt{\xrightarrow{#1}}}}
\newcommand{\stEnvMoveGenAnnot}{\stEnvMoveAnnot{\stEnvAnnotGenericSym}}%
\newcommand{\stEnvMoveInAnnot}[3]{%
  \stEnvMoveAnnot{\stEnvInAnnot{#1}{#2}{#3}}%
}%
\newcommand{\stEnvMoveOutAnnot}[3]{%
  \stEnvMoveAnnot{\stEnvOutAnnot{#1}{#2}{#3}}%
}%
\newcommand{\stEnvMoveCommAnnot}[4]{%
  \stEnvMoveAnnot{\ltsSendRecv{#1}{#2}{#3}{#4}}%
}%
\newcommand{\stEnvMoveCrDeAnnot}[3]{%
  \stEnvMoveAnnot{\ltsCrDe{#1}{#2}{#3}}%
}%

\newcommand{\stEnvInAnnot}[3]{\mpChanRole{\mpS}{#1}:{#2}{\stEnvAnnotInSym}{#3}}%
\newcommand{\stEnvOutAnnot}[3]{\mpChanRole{\mpS}{#1}:{#2}{\stEnvAnnotOutSym}{#3}}%
\newcommand{\stEnvInAnnotSmall}[3]{\mpChanRole{\mpS}{#1}{:}{#2}{\stEnvAnnotInSym}{#3}}%
\newcommand{\stEnvOutAnnotSmall}[3]{\mpChanRole{\mpS}{#1}{:}{#2}{\stEnvAnnotOutSym}{#3}}%
\newcommand{\stEnvCommAnnotSmall}[3]{\ltsSendRecv{\mpS}{#1}{#2}{#3}}%
\newcommand{\stEnvMoveP}[1]{{#1}\!\!\stEnvMove}%
\newcommand{\stEnvMoveMaybeCrashP}[2][]{{#2}\!\!\stEnvMoveMaybeCrash[#1]}%
\newcommand{\stEnvNotMoveP}[1]{{#1}\!\!\not\stEnvMove}%
\newcommand{\stEnvNotMoveMaybeCrashP}[2][]{{#2}\!\!\not\stEnvMoveMaybeCrash[#1]}%
\newcommand{\stEnvMoveStar}{\mathrel{\stFmt{\stEnvMove{}^{\!\!\!*}}}}%
\newcommand{\stEnvMoveMaybeCrashStar}[1][]{%
  \ifempty{#1}{%
    \mathrel{\stFmt{\to^*_{\!\lightning}}}%
  }{%
    \mathrel{\stFmt{\to^*_{\!\lightning \setminus {#1}}}}%
  }
}%
\newcommand{\stEnvMoveAnnotP}[2]{{#1}\!\!\stEnvMoveAnnot{#2}}%
\newcommand{\stEnvMoveGenAnnotP}[1]{{#1}\!\!\stEnvMoveGenAnnot}%

\newcommand{\stEnvEndPred}{\operatorname{end}}%
\newcommand{\stEnvEndP}[1]{\stEnvEndPred(\stFmt{#1})}%

\newcommand{\stEnvSafePred}{\operatorname{safe}}%
\newcommand{\stEnvSafeP}[1]{\stEnvSafePred(\stFmt{#1})}%
\newcommand{\stEnvSafeSessRolesSP}[3]{\stEnvSafePred({#1};{#2},\stFmt{#3})}%
\newcommand{\mpEnv}[1][]{\stFmt{\ifempty{#1}{\Theta}{\Theta_{#1}}}}%
\newcommand{\mpEnvEmpty}{\stFmt{\emptyset}}%
\newcommand{\mpEnvMap}[2]{\stFmt{\mpFmt{#1}{:}\stFmt{#2}}}%
\newcommand{\mpEnvComp}{\mathpunct{\stFmt{,}}}%
\newcommand{\mpEnvApp}[2]{\stFmt{#1}\!\left(\mpFmt{#2}\right)}%

\newcommand{\stEnvEntails}[3]{%
  \stFmt{#1} \vdash \stFmt{\mpFmt{#2} \mathbin{\!:\!} {#3}}%
}%
\newcommand{\mpEnvEntails}[3]{%
  \stFmt{#1} \vdash \stFmt{\mpFmt{#2} \mathbin{\!:\!} \stFmt{#3}}%
}%

\newcommand{\stJudge}[3]{%
  \stFmt{\ifempty{#1}{#2}{{#1} \cdot {#2}}%
  \mathrel{\mpFmt{\vdash}} \mpFmt{#3}}%
}%
\newcommand{\actCrashed}[2]{\mpChanRole{#1}{#2}\stLabFmt{\mathsf{stop}}}%

\newcommand{\muCol}[1]{{\color{red}#1}}%
\newcommand{\muFmt}[1]{\muCol{\mathsf{#1}}}%

\newcommand{\muJudge}[2]{{#1} \mathrel{\muCol{\models}} \muCol{#2}}%

\newcommand{\muVar}[1][]{\muCol{\ifempty{#1}{\muFmt{Z}}{\muFmt{Z}_{#1}}}}%
\newcommand{\muVari}[1][]{\muCol{\ifempty{#1}{\muFmt{Z}'}{\muFmt{Z}'_{#1}}}}%

\newcommand{\muData}[1][]{\muCol{\ifempty{#1}{\muFmt{d}}{\muFmt{d}_{#1}}}}%

\newcommand{\muAct}[1][]{\muCol{\ifempty{#1}{\alpha}{\alpha_{#1}}}}%

\newcommand{\muForall}[2]{\muCol{\forall{#1}\mathbin{\!.\!}{#2}}}%
\newcommand{\muExists}[2]{\muCol{\exists{#1}\mathbin{\!.\!}{#2}}}%

\newcommand{\muPred}[1][]{\muCol{\ifempty{#1}{\phi}{\phi_{#1}}}}%
\newcommand{\muPredi}[1][]{\muCol{\ifempty{#1}{\phi'}{\phi'_{#1}}}}%
\newcommand{\muAnd}{\mathbin{\muCol{\land}}}%
\newcommand{\muBox}[2]{\muCol{[{#1}]{#2}}}%
\newcommand{\muDiamond}[2]{\muCol{\langle{#1}\rangle{#2}}}%
\newcommand{\muTrue}{\muCol{\top}}%
\newcommand{\muGFP}[2]{\muCol{\nu{#1}\mathbin{\!.\!}{#2}}}%

\newcommand{\muLFP}[2]{\muCol{\mu{#1}\mathbin{\!.\!}{#2}}}%
\newcommand{\muOr}{\mathbin{\muCol{\lor}}}%
\newcommand{\muFalse}{\muCol{\bot}}%
\newcommand{\muImplies}{\mathbin{\muCol{\Rightarrow}}}%
\newcommand{\muWordEmpty}[1][]{\muCol{\epsilon}}%

%% file: abstract.tex
\begin{abstract}
Session types enable the specification and verification of communicating systems.
However, their theory often assumes that processes never fail.
To address this limitation, we present a generalised multiparty session type
(MPST) theory with \emph{crash-stop failures}, where processes can crash arbitrarily.

Our new theory validates more protocols and processes \wrt previous work.
We apply minimal syntactic changes to standard session $\pi$-calculus and types:
we model crashes and their handling semantically,
with a generalised MPST typing system parametric on a
behavioural safety property. %
We cover the spectrum between fully reliable and fully unreliable sessions,
via \emph{optional reliability assumptions},
and prove type safety and protocol conformance in the presence of crash-stop
failures.%

Introducing crash-stop failures has non-trivial consequences:
writing correct processes that handle %
all crash scenarios
can be difficult.
Yet, our generalised MPST theory
allows us to tame this complexity, via model checking,
to validate whether a multiparty session
satisfies desired behavioural properties, \eg deadlock-freedom or
liveness, even in presence of crashes.
We implement our approach %
using the mCRL2 model checker, and
evaluate it with
examples extended from the literature.
\end{abstract}

%% file: 1_intro.tex
\section{Introduction}
Multiparty session types (MPST)~\cite{HYC16} provide a typing discipline for
message-passing processes.
The theory ensures well-typed processes
enjoy desirable properties,
\aka \emph{the Session Theorems}:
type safety (processes communicate without errors),
protocol conformance
(\aka~\emph{session fidelity}, processes behave according to their types),
deadlock-freedom (processes do not get stuck), and liveness (input/output
actions eventually succeed).
Researchers devote significant effort into integrating session
types in programming languages and
tools~\cite{BehTypesTheoryTools2017}.

A common assumption in session type theory is that everything is
reliable and there are no failures, which
is often unrealistic in real-world %
systems.
So, we pose a question:
how can we better
model %
systems \emph{with failures}, and make session types %
less idealistic?

In this paper, we take steps towards bridging the gap between
theory
and practice with a new \emph{generalised} multiparty session type
theory that models failures with \emph{crash-stop}
semantics~\cite[\S 2.2]{DBLP:books/daglib/0025983}:
processes may crash, and crashed processes stop interacting with the world.
This model is standard in distributed systems, and is used in
related work on session types with error-handling
capabilities~\cite{ESOP18CrashHandling,OOPSLA21FaultTolerantMPST}.
However, unlike previous work,
we allow
\emph{any} process to crash arbitrarily,
and
support optional assumptions on non-crashing processes.

In our new theory, we add crashing and crash handling
semantics to processes and session types.
With minimal changes to the standard surface syntax,
we
model a variety of subtle, complex behaviours arising from unreliable communicating
processes.
An active process $\mpP$ may crash arbitrarily, and a process $\mpQ$
interacting with $\mpP$ might need to be prepared to handle possible crashes.
Messages sent from $\mpQ$ to a crashed $\mpP$ are lost
-- but if $\mpQ$ tries to receive from $\mpP$, then $\mpQ$ can detect that
$\mpP$ has crashed, and take a crash handling branch.
Meanwhile, another process $\mpR$ may (or may not) have detected $\mpP$'s crash,
and may be handling it -- and in either case, any interaction between $\mpQ$
and $\mpR$ should remain correct.

Our MPST theory is generalised in two aspects:
\emph{(1)} we introduce \emph{optional reliability assumptions}, so we can
model a mixture of reliable and unreliable communicating peers; and
\emph{(2)} our type system is parametric on a type-level behavioural property $\predP$
which can be instantiated as safety,
deadlock freedom,
liveness, \etc (in the style of \cite{POPL19LessIsMore}),
while accounting for potential crashes.
We prove %
\emph{session fidelity},
showing how type-level properties transfer to
well-typed processes;
we also prove that our new theory satisfies other Session Theorems of MPST, %
while (unlike previous work) being resilient to arbitrary crash-stop failures.
With
\emph{optional reliability assumptions},
one may declare that some peers will never crash for the
duration of the protocol.
Such optional assumptions allow for simplifying protocols and
programs: if a peer is assumed reliable, the other peers can interact with it
without needing to handle its crashes.
By making such assumptions explicit and customisable, our theory
supports a \emph{spectrum} of scenarios ranging from only having sessions
with reliable peers (thus subsuming classic MPST
works~\cite{POPL19LessIsMore,HYC16}), to having no reliable
peers at all.

As in the real world, a system with crash-stop failures can have subtle complex
behaviours;
hence, writing protocols and processes where all possible crash scenarios are
correctly handled can be hard.
This highlights a further benefit of our generalised theory:
we formalise our  %
behavioural properties as modal $\mu$-calculus \formulae,
and verify them with a model checker.
To show the feasibility of our approach, we present an accompanying
tool, utilising the mCRL2 model checker~\cite{TACAS19mCRL2},
for verifying session properties under
optional reliability assumptions.

\input{2_overview.tex}
%
\subparagraph{Contributions and Structure.}
In~\!\!\descriptionlabel{\cref{sec:session-calculus}}
  we introduce a
  multiparty session $\pi$-calculus (with minimal changes to the standard syntax)
  giving crash and crash handling semantics modelling \emph{crash-stop} failures.
In~\!\!\descriptionlabel{\cref{sec:gtype}}
  we present \emph{multiparty session types with crashes}:
  they describe how communication channels should be used
  to send/receive messages, and handle crashes.
  We formalise the semantics
  of
  collections of local types under \emph{optional}
  reliability assumptions;
  we introduce a type system, and
  prove %
  the Session Theorems:
  type safety, protocol conformance, and process properties (deadlock-freedom, termination, liveness, \etc)
  in
  \cref{lem:subject-reduction,lem:session-fidelity,lem:stenv-proc-properties}.
In~\!\!\descriptionlabel{\cref{sec:model-checking}}
  we show how model checking can be
  incorporated to verify %
  our behavioural properties, %
  by %
  expressing them %
  as modal $\mu$-calculus \formulae.
We discuss related work and conclude
in~\!\!\descriptionlabel{\cref{sec:related}}.
\iftoggle{techreport}{
  The appendices%
}{
  Appendices~\!\!\descriptionlabel{\cref{sec:appendix:examples}} and~\!\!\descriptionlabel{\cref{sec:appdx:examples}}
}
include additional examples\ifTR{, definitions, proofs of main theorems,}
and
more details about the tool implementing our
theory using the mCRL2 model checker.%
\iftoggle{techreport}{}{
  Further appendices with proofs are available in a separate technical report \cite{techreport}.%
}

%% file: 2_overview.tex
\subparagraph{Overview.}
\label{sec:overview}

Session typing systems assign \emph{session types} (\aka local
types) to communication channels, used by processes
to send and receive messages.
In essence, a session type
describes a \emph{protocol}:
how a \emph{role} is expected to interact with other roles
in a multiparty session.
The type system checks %
whether a process implements desired protocols. %

As an example, consider a simple Domain Name System (DNS) scenario:
a client $\roleP$ queries a server $\roleQ$
for an IP address of a host name.
With classic session types (without crashes),
we use the type
\(
\stT[\roleP] = \stOut{\roleQ}{\stLabFmt{req}}{} \stSeq
\stInNB{\roleQ}{\stLabFmt{res}}{}{}
\)
to represent the client $\roleP$'s behaviour:
first sending ($\stFmt{\oplus}$) a $\stLabFmt{req}$uest message
to server $\roleQ$, and then receiving ($\stFmt{\&}$) a $\stLabFmt{res}$ponse
from $\roleQ$.
The server implements a dual type %
\(
\stT[\roleQ] = \stInNB{\roleP}{\stLabFmt{req}}{}{
  \stOut{\roleP}{\stLabFmt{res}}{} \stSeq %
}
\),
who receives a $\stLabFmt{req}$uest from client $\roleP$,
and then sends a $\stLabFmt{res}$ponse to $\roleP$.
We can write a process %
$\mpQ =
\mpBranchSingle{\mpChanRole{\mpS}{\roleQ}}{\roleP}{req}{}{\mpSel{\mpChanRole{\mpS}{\roleQ}}{\roleP}{res}{}{\mpNil}}$
for the server.
Using $\stT[\roleQ]$, we type-check the channel (\aka session endpoint)
$\mpChanRole{\mpS}{\roleQ}$, where $\mpQ$ plays the role $\roleQ$ on session
$\mpS$.
Here, $\mpQ$ type-checks -- it uses channel $\mpChanRole{\mpS}{\roleQ}$ correctly, according to type $\stT[\roleQ]$.
In this work, we augment the classic session types theory by introducing process
failures with \emph{crash-stop} semantics~\cite[\S 2.2]{DBLP:books/daglib/0025983}.
We adopt the following failure model:
\emph{(1)}
    processes have \emph{crash-stop} failures, \ie~they may crash and do
    not recover;
\emph{(2)}
    communication channels deliver messages in order, without losses (unless
    the recipient has crashed);
\emph{(3)}
    each process has a failure detector~%
    \cite{JACM96FailureDetector},
    so a process trying to receive from a crashed peer accurately detects the
    crash. %
    The combination of \emph{(1)}, %
\emph{(2)}, and
\emph{(3)}
is called the \emph{crash-fail}
model in \cite[\S 2.6.2]{DBLP:books/daglib/0025983}.

We now revise our DNS example in the presence of failures.
Let us assume that the server $\roleQ$ may crash, whereas the client $\roleP$
remains reliable.
The client $\roleP$ %
may now send its $\stLabFmt{req}$uest to a new \emph{failover server} $\roleR$
(assumed reliable for simplicity).
We represent this scenario %
by a type for the new failover server $\stTi[\roleR]$, and a new branch in
$\stTi[\roleP]$ for handling $\roleQ$'s crash:

\smallskip
\centerline{\(
\stTi[\roleP] \;=\; \stOut{\roleQ}{\stLabFmt{req}}{} \stSeq
  \stExtSum{\roleQ}{}{
    \begin{array}{@{}l@{}}
    \stChoice{\stLabFmt{res}}{}%
    \\
    \stChoice{\stCrashLab}{} \stSeq
    \stOut{\roleR}{\stLabFmt{req}}{} \stSeq
    \stInNB{\roleR}{\stLabFmt{res}}{}{}%
    \end{array}
  }
\qquad
\begin{array}{l}
\stTi[\roleQ] \;=\; \stInNB{\roleP}{\stLabFmt{req}}{}{
  \stOut{\roleP}{\stLabFmt{res}}{}%
}
\\
\stTi[\roleR] \;=\; \stInNB{\roleQ}{\stCrashLab}{}{} \stSeq
  \stInNB{\roleP}{\stLabFmt{req}}{}{} \stSeq
  \stOut{\roleP}{\stLabFmt{res}}{}{}%
\end{array}
\)}
\smallskip
\noindent
Here, $\stTi[\roleP]$ states that client $\roleP$ first sends
a message to
the unreliable server $\roleQ$; %
then, $\roleP$ expects a $\stLabFmt{res}$ponse
from $\roleQ$.
If $\roleQ$ crashes, %
the client $\roleP$ detects the crash and
handles it (via the new $\stCrashLab$ handling branch)
by $\stLabFmt{req}$uesting from the failover
server $\roleR$.
Meanwhile, $\roleR$ \emph{also} detects
whether $\roleQ$ has crashed.
If so, $\roleR$ activates its $\stCrashLab$ handling branch
and handles $\roleP$'s $\stLabFmt{req}$uest.

\begin{figure}
  \centering
  \begin{minipage}{0.4\textwidth}
    \centering
    \includegraphics[height=2.5cm]{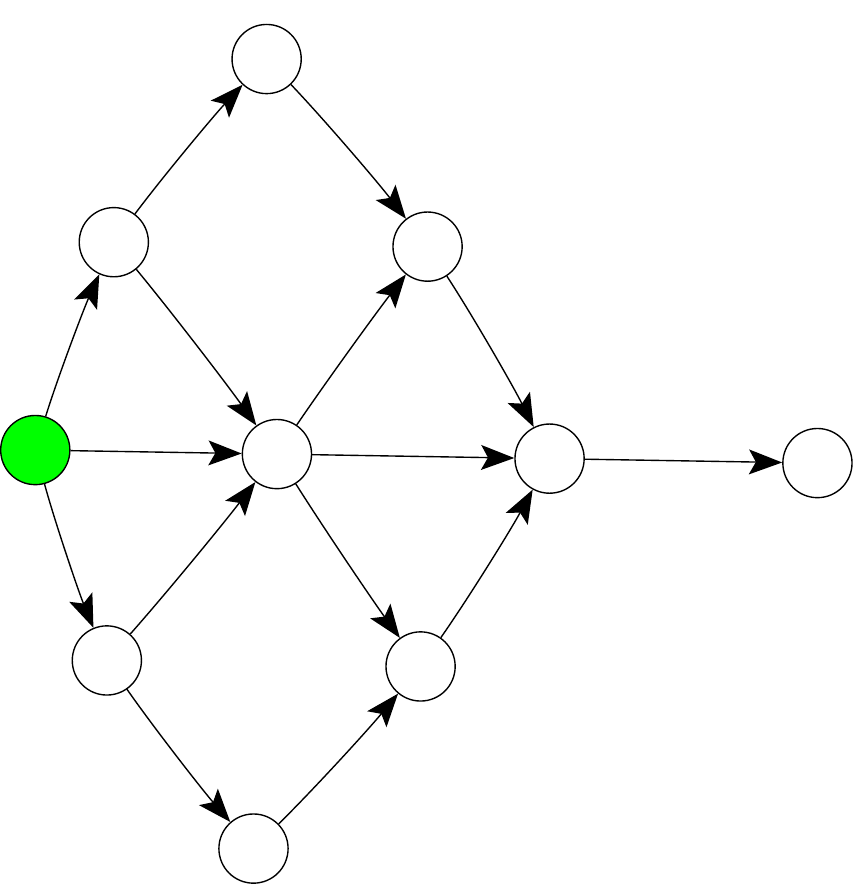}
  \end{minipage}
  \quad
  \begin{minipage}{0.55\textwidth}
    \centering
    \includegraphics[height=2.5cm]{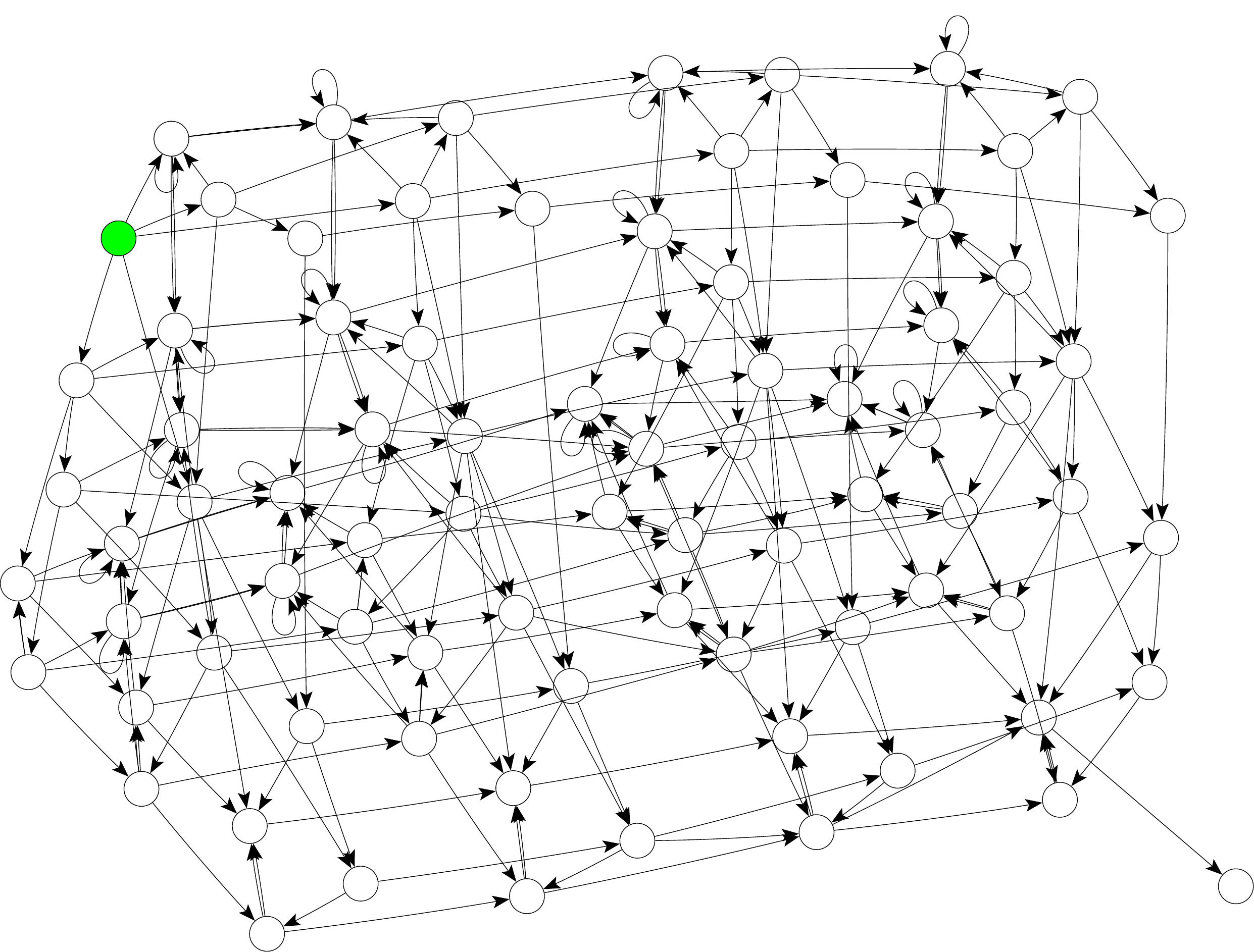}
  \end{minipage}
  \caption{Transition systems (based on \Cref{def:mpst-env-reduction}, with labels omitted) generated from the DNS examples.
  Left: without crashes/handling. Right: with crashes (for $\roleQ$) and crash handling.
  }
  \label{fig:overview:lts-compare}
\end{figure}

In our model,
crash detection and handling is done on the receiving side,
\eg $\stTi[\roleP]$ detects whether $\roleQ$ has $\stCrashLab$ed when
waiting for a $\stLabFmt{res}$ponse,
while $\stTi[\roleR]$ monitors whether $\roleQ$ has
crashed.
Handling crashes when receiving messages from a reliable role is unnecessary,
\eg the server $\roleQ$ does not need crash handling when it receives from the
(reliable) client $\roleP$;
similarly, the (reliable) roles $\roleP$ and $\roleR$ interact without
crash handling.
This failure model is reflected in the semantics of both processes and
session types in our work.
Unlike classic MPST works,
we allow processes to crash \emph{arbitrarily} while attempting inputs or
outputs (\cref{sec:session-calculus}).
When a process crashes, the channel endpoints held by the process also
crash, and are assigned the new type $\stStop$ (\Cref{sec:gtype}).
\Eg when the server process $\mpQ$ crashes, the endpoint
$\mpChanRole{\mpS}{\roleQ}$ held by $\mpQ$ becomes a crashed endpoint
$\mpStop{\mpS}{\roleQ}$;
accordingly, the server type $\stTi[\roleQ]$ advances to $\stStop$ to reflect
the crash.

To ensure that communicating processes are type-safe even in the presence of
crashes,
we require their session types %
to satisfy a \emph{safety property} accounting for possible crashes
(\Cref{def:mpst-env-safe}), %
which can be refined, \eg as deadlock-freedom or liveness
(\Cref{def:typing-ctx-properties}).
We prove subject reduction,
session fidelity,
and various process properties
(deadlock-freedom, liveness, \etc)
even in the presence of crashes and optional reliability assumptions
(Thms.~\ref{lem:subject-reduction}, \ref{lem:session-fidelity}, \ref{lem:stenv-proc-properties}).
Despite minimal changes to the surface syntax of session types
and processes,
the semantics surrounding crashes introduce subtle behaviours and increase complexity.
Taking the DNS examples above, we compare the sizes of their (labelled) transition
systems in \cref{fig:overview:lts-compare} (based on
\Cref{def:mpst-env-reduction}):
the original system (left, two roles $\roleP$ and $\roleQ$, no crashes) has
10 states and 15 transitions;
and
the revised system (right, $\roleQ$ may crash, with a new role $\roleR$)
has 101 states and 427 transitions.
We discuss another, more complex example in \Cref{sec:model-checking}.
Checking whether a given combination of session types with possible crashes
is safe, deadlock-free, or live, can be challenging due to
non-trivial behaviours and increased model size
arising from crashes and crash handling.
To tackle this, we show how to automatically
verify such type-level properties %
by representing them as modal $\mu$-calculus \formulae via
the mCRL2 model checker~\cite{TACAS19mCRL2}. %
%
%
%
%
%
%
%
%
%
%
%
%
%
%
%
%

%
%
%
%
%
%
%
%
%
%
%
%
%
%
%
%
%
%

%
%
%
%
%
%
%
%
%
%
%
%
%
%
%
%
%
%
%
%
%
%
%
%
%
%
%
%
%
%
%
%
%
%
%
%
%
%
%
%
%
%
%
%
%
%
%
%
%
%
%
%
%
%
%
%
%
%
%
%
%
%
%
%
%
%
%
%
%
%
%
%
%
%
%
%
%
%
%
%
%
%
%
%
%
%
%
%
%
%
%
%
%
%
%
%
%
%
%
%
%
%
%
%
%
%
%
%
%
%
%
%
%
%
%
%
%
%
%
%
%
%
%
%
%
%
%
%
%
%
%
%
%
%
%
%
%
%
%
%
%
%

%
%
%
%
%
%
%
%
%
%
%
%
%
%
%
%
%

%% file: 3_calculus/session-calculus.tex
\section{Multiparty Session Calculus with Crash-Stop Semantics}
\label{sec:session-calculus}

In this section, we formalise the syntax %
and operational semantics %
of our multiparty session $\pi$-calculus, where a process can fail arbitrarily,
and crashes can be detected and handled by receiving processes.
For clarity of presentation, we formalise a synchronous semantics.
%
%
%
%

\subparagraph{Syntax of Processes.}%
\label{sec:session-calculus:syntax}
Our multiparty session $\pi$-calculus %
models processes that interact
via multiparty channels, and may arbitrarily crash.
For simplicity of presentation, our calculus is streamlined to focus on communication;
standard extensions, \eg with expressions and ``if\ldots then\ldots else'' statements,
are routine and orthogonal to our formulation. %

\begin{definition}[Syntax of Multiparty Session $\pi$-Calculus]%
\label{def:mpst-syntax-terms}
Let $\roleP, \roleQ, \ldots$ denote \emph{roles}
belonging to a set $\roleSet$;
let $\mpS, \mpSi, \ldots$ denote \emph{sessions};
let $\mpFmt{x}, \mpFmt{y}, \ldots$ denote \emph{variables};
let $\mpLab, \mpLabi, \ldots$ denote \emph{message labels}; %
let $\mpX, \mpY, \ldots$ denote \emph{process variables}.
The \emph{multiparty session $\pi$-calculus} syntax is:%

\input{3_calculus/syntax.tex}
\noindent%
We write $\mpBigPar{i \in I}{\mpP[i]}$ for the parallel composition of processes $\mpP[i]$.
Restriction, branching, and process definitions and declarations act as
binders, as expected;
$\fc{\mpP}$ is the set of \emph{free channels with roles} in $\mpP$ (including
$\mpChanRole{\mpS}{\roleP}$ in $\mpStopDef$), and $\fv{\mpP}$ is
the set of \emph{free variables} in $\mpP$.
Noticeable changes \wrt standard session calculi are \highlightText{highlighted}.
\end{definition}

Our calculus
(\cref{def:mpst-syntax-terms}) %
includes basic values $\mpV$ (\eg unit $()$, integers, strings),
channels with roles (\aka session endpoints) $\mpChanRole{\mpS}{\roleP}$,
session scope restriction $\mpRes{\mpS}{\mpP}$,
inaction $\mpNil$, parallel
composition $P \mpPar Q$, process
definition $\mpDefAbbrev{\mpDefD}{\mpP}$, process call
$\mpCallSmall{\mpX}{\widetilde{\mpD}}$, and error $\mpErr$.
\emph{Selection} (\aka internal choice)
$\mpSel{\mpC}{\roleQ}{\mpLab}{\mpD}{\mpP}$
sends a message $\mpLab$ with payload
$\mpD$ to role $\roleQ$ via endpoint $\mpC$, where $\mpC$ may be a variable or channel with role,
while $\mpD$ may also be a basic value.
\emph{Branching} (\aka external choice)
$\mpBranch{\mpC}{\roleQ}{i \in I}{\mpLab[i]}{x_i}{\mpP[i]}{}$
expects to receive a message $\mpLab[i]$
(for some $i \in I$) %
from role $\roleQ$ via endpoint $\mpC$,
and then continues
as $\mpP[i]$.
Importantly, a process implements crash detection by ``receiving'' the
special message label $\mpLabCrash$ in an external choice; such special message
\emph{cannot} be sent by any process (side condition
$\mpLab \neq \mpLabCrash$ in selection). For example,
$\mpBranch{\mpChanRole{\mpS}{\roleP}}{\roleQ}{}{\mpLab}{\mpFmt{x}}{\mpP}{\mpPi}$
is a process that uses the session endpoint $\mpChanRole{\mpS}{\roleP}$ to
receive message $\mpLab$ from $\roleQ$, but if $\roleQ$ has crashed, then the
process continues as $\mpPi$.
Finally, our calculus includes \emph{crashed session endpoints} $\mpStopDef$,
denoting that the endpoint for role $\roleP$ in session $\mpS$ has crashed.

\input{3_calculus/mpst-semantics.tex}
%

\subparagraph{Operational Semantics.}%
\label{sec:session-calculus:semantics}
We give the operational semantics of our session $\pi$-calculus in
\cref{def:mpst-pi-semantics}, using
a standard \emph{structural congruence}
extended with a new \emph{crash elimination rule} %
which garbage-collects sessions where all endpoints are crashed:
(full congruence rules in \cref{sec:structural-congruence})
\smallskip
\centerline{\(
    \inferenceSingle[\iruleCongStopElim]{
      \mpRes{\mpS}{(\mpStop{\mpS}{\roleP[1]} \mpPar \cdots \mpPar \mpStop{\mpS}{\roleP[n]})}
      \;\equiv\;
      \mpNil
    }
\)} %

\begin{definition}%
  \label{def:mpst-proc-context}%
  \label{def:mpst-pi-reduction-ctx}%
  \label{def:mpst-pi-semantics}%
  \label{def:mpst-pi-error}%
  A \emph{reduction context} $\mpCtx$ is defined as:
  \;\hfill
  \(%
  \mpCtx \,\coloncolonequals\,%
  \mpCtx \mpPar \mpP%
  \bnfsep%
  \mpRes{\mpS}{\mpCtx}%
  \bnfsep%
  \mpDefAbbrev{\mpDefD}{\mpCtx}%
  \bnfsep%
  \mpCtxHole%
  \)%

  \noindent%
  \emph{The reduction $\!\mpMove\!$} %
  is defined in \cref{fig:mpst-pi-semantics}; %
  we write $\mpMovePlus$\slash\,$\mpMoveStar$ for its transitive\,\slash\,reflexive-transitive closure.
  We write $\mpNotMoveP{\mpP}$ iff $\not\exists\mpPi$
  such that $\mpP \!\mpMove\! \mpPi$ is derivable \emph{without}
  rules \inferrule{\iruleMPCrashS} and \inferrule{\iruleMPCrashR}
  (\ie $\mpP$ is stuck, unless a crash occurs).
  We say \emph{$\mpP$ has an error} %
  iff $\exists \mpCtx$ with %
  \,$\mpP \!=\! \mpCtxApp{\mpCtx}{\mpErr}$.
\end{definition}

\noindent

Part of our operational semantics rules in \cref{fig:mpst-pi-semantics} are standard.
Rule \inferrule{\iruleMPRedComm} describes a communication on session $\mpS$
between receiver $\roleP$ and sender
$\roleQ$, if the sent message $\mpLab[k]$ can be handled by the
receiver ($k\!\in\!I$);
otherwise, a message label mismatch causes an $\mpErr$or via rule
\inferrule{\iruleMPErrLabel}.
Rule \inferrule{\iruleMPRedCall} expands process
definitions when called.
Rules \inferrule{\iruleMPRedCtx} and \inferrule{\iruleMPRedCongr}
allow processes to reduce under reduction contexts
and modulo structural congruence.%

The remaining rules in \Cref{fig:mpst-pi-semantics} (\highlightText{highlighted})
are novel: they model crashes, and crash handling.
Rules \inferrule{\iruleMPCrashS} and
\inferrule{\iruleMPCrashR} state that a
process $\mpP$ may crash while attempting any selection or branching operation,
respectively;
when $\mpP$ crashes, it reduces to a parallel composition where all the channel
endpoints held by $\mpP$ are crashed.
The \emph{lost message rules} \inferrule{\iruleMPRedCommEGround} and
\inferrule{\iruleMPRedCommE}
state that if a process %
sends a message to a crashed endpoint, then the message is lost; if the message payload is a session endpoint
$\mpChanRole{\mpSi}{\roleR}$, then it becomes crashed.
Finally, the \emph{crash handling rule} \inferrule{\iruleMPRedCommD} states
that if a process attempts to receive a message from a crashed endpoint,
then the process detects the crash and follows its crash handling branch $\mpPi$.
We now show an example of rule \inferrule{\iruleMPCrashS}; %
more examples can be found in \cref{sec:appendix:examples}.

\begin{example}
Processes
\(
  \mpP = \mpSel{\mpChanRole{\mpS}{\roleP}}{\roleQ}{\mpLabi}{
    \mpChanRole{\mpS}{\roleR}}{
      \mpBranchSingle{\mpChanRole{\mpS}{\roleP}}{\roleR}{\mpLab}{x}{}
    }
\)
and
\(
  \mpQ = \mpBranchSingle{\mpChanRole{\mpS}{\roleQ}}{\roleP}{\mpLabi}{x}{
    \mpSel{x}{\roleP}{\mpLab}{42}{}}
\)
communicate on a session $\mpS$; $\mpP$ uses $\mpChanRole{\mpS}{\roleP}$ to send
$\mpChanRole{\mpS}{\roleR}$ to role $\roleQ$; $\mpQ$ uses
$\mpChanRole{\mpS}{\roleQ}$ to receive it, then sends a message to role $\roleP$
via $\mpChanRole{\mpS}{\roleR}$.
Suppose that $\mpP$
crashes before sending: this gives rise to the reduction (by rule \inferrule{\iruleMPCrashS})\;
\(
  \mpRes{\mpS}{(\mpP \mpPar \mpQ)}
  \mpMoveCrash
  \mpRes{\mpS}{(
    \mpStop{\mpS}{\roleP}
    \mpPar
    \mpStop{\mpS}{\roleR}
    \mpPar
    \mpQ
  )}
\).
\;Observe that $\mpChanRole{\mpS}{\roleP}$ and $\mpChanRole{\mpS}{\roleR}$, which were held by $\mpP$, are now crashed.
\end{example}

%% file: 3_calculus/syntax.tex
\smallskip%
\centerline{$%
\begin{array}{r@{\hskip 2mm}c@{\hskip 2mm}l@{\hskip 2mm}l}
  \textstyle%
  \mpC
  &\coloncolonequals&%
  \mpFmt{x} \bnfsep \mpChanRole{\mpS}{\roleP}%
  & \mbox{\footnotesize(variable or channel for session $\mpS$ with role $\roleP$)}
  \\
  \mpD
  &\coloncolonequals&
  \mpV \bnfsep \mpC
  & \mbox{\footnotesize(basic value, variable, or channel with role)}
  \\
  \mpW
  &\coloncolonequals&
  \mpV \bnfsep \mpChanRole{\mpS}{\roleP}
  & \mbox{\footnotesize(basic value or channel with role)}
  \\[0mm]
  \mpP, \mpQ
  &\coloncolonequals&%
  \mpNil
    \;\bnfsep\; \mpRes{\mpS}{\mpP}%
    \;\bnfsep\; \mpP \mpPar \mpQ
  &
  \mbox{\footnotesize(inaction, restriction, parallel composition)}
  \\
  &&
  \mpSel{\mpC}{\roleQ}{\mpLab}{\mpD}{\mpP}
  \quad \highlight{\mbox{\footnotesize(where $\mpLab \neq \mpLabCrash$)}}
  &
  \mbox{\footnotesize(selection towards role $\roleQ$)}
  \\
  &&
  \mpBranch{\mpC}{\roleQ}{i \in I}{\mpLab[i]}{x_i}{\mpP[i]}{}%
  &
  \mbox{\footnotesize(branching from role $\roleQ$ with an index set $I \neq \emptyset$)}
  \\
  &&
  \mpDefAbbrev{\mpDefD}{\mpP}%
  \bnfsep%
  \mpCallSmall{\mpX}{\widetilde{\mpD}}
  &
  \mbox{\footnotesize(process definition, process call)}
  \\[0mm]
  &&
  \mpErr%
  \bnfsep%
  \highlight{\mpStopDef}%
  &
  \mbox{\footnotesize(error, crashed channel endpoint)}
  \\[0mm]
  \mpDefD%
  &\coloncolonequals&%
  \mpJustDef{\mpX}{\widetilde{x}}{P}
  &
  \mbox{\footnotesize(declaration of process variable $\mpX$)}
  \\[1mm]
\end{array}
$}

%% file: 3_calculus/mpst-semantics.tex
\begin{figure}[t]
  \centerline{\(%
  \begin{array}{@{}c@{}}
    \begin{array}{@{}rl@{}}
      \inferrule{\iruleMPRedComm}& %
      \mpBranch{\mpChanRole{\mpS}{\roleP}}{\roleQ}{i \in I}{%
        \mpLab[i]}{x_i}{\mpP[i]}{}%
      \,\mpPar\,%
      \mpSel{\mpChanRole{\mpS}{\roleQ}}{\roleP}{\mpLab[k]}{%
        \mpW
      }{\mpQ}%
      \;\;\mpMove\;\;%
      \mpP[k]\subst{\mpFmt{x_k}}{\mpW}
      \,\mpPar\,%
      \mpQ%
      \qquad
      \text{if\, $k \!\in\! I$}%
      \\[1mm]%

      \inferrule{\iruleMPErrLabel}&%
      \mpBranch{\mpChanRole{\mpS}{\roleP}}{\roleQ}{i \in I}{%
        \mpLab[i]}{x_i}{\mpP[i]}{}%
      \,\mpPar\,%
      \mpSel{\mpChanRole{\mpS}{\roleQ}}{\roleP}{\mpLab}{%
        \mpW
      }{\mpQ}%
      \;\mpMove\;%
      \mpErr%
      \qquad%
      \text{if\, $\forall i \!\in\! I: \mpLab[i] \!\neq\! \mpLab$}%
    \\[1mm]

      \inferrule{\iruleMPRedCall}& %
      \mpDef{\mpX}{x_1,\ldots,x_n}{\mpP}{(%
        \mpCall{\mpX}{%
          \mpW[1], \ldots, \mpW[n]
        }%
        \,\mpPar\,%
        \mpQ%
        )%
      }%
      \\%
      &\hspace{10mm}%
      \;\mpMove\;%
      \mpDef{\mpX}{x_1,\ldots,x_n}{\mpP}{(%
        \mpP
        \subst{\mpFmt{x_1}}{\mpW[1]}%
        \cdots
        \subst{\mpFmt{x_n}}{\mpW[n]}%
        \,\mpPar\,%
          \mpQ%
          )%
      }%
      \\[1mm]%
      \inferrule{\iruleMPRedCtx}&%
      \mpP \mpMove \mpPi%
      \;\;\text{implies}\;\;%
      \mpCtxApp{\mpCtx}{\mpP} \mpMove \mpCtxApp{\mpCtx}{\mpPi}%
      \\[1mm]%

 \inferrule{\iruleMPRedCongr} &
 \mpPi \equiv \mpP
 \;\;\text{and}\;\;
 \mpP \mpMove \mpQ
 \;\;\text{and}\;\;
 \mpQ \equiv \mpQi
 \;\;\text{implies}\;\;
 \mpPi \mpMove \mpQi

 \\[0mm]%

    \inferrule{\iruleMPCrashS} &
    \highlight{%
    \mpP \;=\; \mpSel{\mpChanRole{\mpS}{\roleP}}{\roleQ}{\mpLab}{
      \mpW
    }{\mpPi}
    \;\mpMoveCrash\;
    \mpBigPar{j \in J}{\mpStop{\mpS[j]}{\roleP[j]}}
    }%

    \hfill
    \highlight{%
      \text{where $\{\mpChanRole{\mpS[j]}{\roleP[j]}\}_{j \in J} = \fc{\mpP}$} %
    }%
    \\[0mm]
  \inferrule{\iruleMPCrashR} &
    \highlight{%
    \mpP \;=\; \mpBranch{\mpChanRole{\mpS}{\roleP}}{\roleQ}{i \in I}{\mpLab[i]}{x_i}{\mpP[i]}{}
    \;\mpMoveCrash\;
    \mpBigPar{j \in J}{\mpStop{\mpS[j]}{\roleP[j]}}
    }%

    \hfill
    \highlight{%
      \text{where $\{\mpChanRole{\mpS[j]}{\roleP[j]}\}_{j \in J} = \fc{\mpP}$}
    }%
  \\[0mm]

    \inferrule{\iruleMPRedCommEGround} &
    \highlight{%
    \mpStop{\mpS}{\roleP}
      \,\mpPar\,
      \mpSel{\mpChanRole{\mpS}{\roleQ}}{\roleP}{\mpLab}{
        \mpV
      }{\mpQi}
      \;\;\mpMove\;\;
      \mpStop{\mpS}{\roleP}
      \,\mpPar\,
      \mpQi
    }%
  \\[0mm]

  \inferrule{\iruleMPRedCommE} &
    \highlight{%
    \mpStop{\mpS}{\roleP}
      \,\mpPar\,
      \mpSel{\mpChanRole{\mpS}{\roleQ}}{\roleP}{\mpLab}{
        \mpChanRole{\mpSi}{\roleR}
      }{\mpQi}
      \;\;\mpMove\;\;
      \mpStop{\mpS}{\roleP}
      \,\mpPar\,
      \mpStop{\mpSi}{\roleR}
      \,\mpPar\,
      \mpQi
    }%
  \\[0mm]

  \inferrule{\iruleMPRedCommD} &
    \highlight{%
    \mpBranch{\mpChanRole{\mpS}{\roleP}}{\roleQ}{i \in I}{\mpLab[i]}{x_i}{\mpP[i]}{\mpPi}
    \,\mpPar\,
    \mpStop{\mpS}{\roleQ}
    \;\;\mpMove\;\;
    \mpPi
    \,\mpPar\,
    \mpStop{\mpS}{\roleQ}
    }%
    \end{array}
    \end{array}
    \)}%
  \caption{%
    Semantics of our session $\pi$-calculus.
    Rule \inferrule{\iruleMPRedCongr} uses the congruence $\equiv$ defined in \Cref{sec:structural-congruence}.
  }%
  \label{fig:mpst-pi-semantics}%
  %
\end{figure}

%% file: 4_mpst_crash/mpst_with_crashes.tex
\section{Multiparty Session Types with Crashes}\label{sec:gtype}

In this section, we present a generalised type system for our multiparty session
$\pi$-calculus (introduced in \cref{def:mpst-syntax-terms}). As in standard MPST, we assign session types to channel endpoints;
we show the syntax of our types in \cref{sec:gtype:syntax}, where
our key additions are \emph{crash handling branches}, and a new type $\stStop$ for crashed endpoints.
In \cref{sec:gtype:lts-context}, we
give a labelled transition system (LTS) semantics to typing contexts,
to represent the behaviour of a collection of types.

Unlike classic MPST, our type system is generalised in the style of
\cite{POPL19LessIsMore}, hence it has \emph{no} global types; rather, it uses a
\emph{safety property}
formalising the \emph{minimum} requirement for a typing context to ensure
\emph{subject reduction} (and thus, type safety).
In this paper, such a safety property is defined in \cref{sec:type-system-safety}:
unlike previous work, the property accounts for potential crashes,
and supports explicit (and optional) reliability assumptions.
We show typing rules
in \cref{sec:type-system:tyrules},
and the main properties of the typing system: subject reduction
(\cref{lem:subject-reduction}) and session fidelity
(\cref{lem:session-fidelity}) in \cref{sec:type-system:subj-red}.
Finally, we demonstrate how we can infer runtime process properties from typing contexts in \cref{sec:dedlock-freedom}.

\input{4_mpst_crash/1_types.tex}

\input{4_mpst_crash/2_context.tex}

\input{4_mpst_crash/3_safety.tex}

\input{4_mpst_crash/4_typesystem.tex}

\input{4_mpst_crash/5_subjred_sessfid.tex}
\input{4_mpst_crash/6_properties.tex}

%% file: 4_mpst_crash/1_types.tex
\subsection{Types}%
\label{sec:gtype:syntax}

A \emph{session type} describes
how a process is expected to use a communication channel to send/receive
messages to/from other roles involved in a multiparty session.
We formalise the syntax of session types in \cref{fig:syntax-mpst}, where we add the
$\stStop$ type to their standard syntax~\cite{POPL19LessIsMore}.

\input{4_mpst_crash/syntax.tex}

The internal choice (selection) type
$\stIntSum{\roleP}{i \in I}{\stChoice{\stLab[i]}{\stS[i]} \stSeq \stT[i]}$
denotes \emph{sending} a message $\stLab[i]$ (by picking some $i \in I$)
with a payload of type $\stS[i]$ to role $\roleP$, and then continue the protocol as $\stT[i]$.
Dually, the external choice (branching) type
$\stExtSum{\roleP}{i \in I}{\stChoice{\stLab[i]}{\stS[i]} \stSeq \stT[i]}$
denotes \emph{receiving} a message $\stLab[i]$ (for any $i \in I$)
with a payload of type $\stS[i]$ from role $\roleP$, and then continue as $\stT[i]$.
The type $\stEnd$ indicates that a session endpoint should not be used for further communications.

\subparagraph{Crashes and Crash Detection.}
The key novelty of \Cref{fig:syntax-local-type} is the new type $\stStop$
describing a crashed session endpoint.
Similarly to \Cref{def:mpst-syntax-terms}, we also introduce a distinguished
message label $\stCrashLab$ for crash handling in external choices. For example,
recall the types in \Cref{sec:overview}:
\begin{itemize}%
  \item the type $\stExtSum{\roleQ}{}{\stChoice{\stLabFmt{res}}{} \stSeq \stT,\,
    \stChoice{\stCrashLab}{} \stSeq \stTi}$ means that we expect a
    $\stChoice{\stLabFmt{res}}{}$ponse message from role $\roleQ$, but if we detect that
    $\roleQ$ has crashed, then the protocol continues along the handling branch
    $\stTi$;
  \item the type $\stInNB{\roleQ}{\stCrashLab}{}{} \stSeq \stT$ denotes a
    ``pure'' crash recovery behaviour: we are not communicating with $\roleQ$,
    but the recovery protocol $\stT$ is activated whenever we detect that $\roleQ$ has crashed.
\end{itemize}
Since $\mpLabCrash$ messages \emph{cannot} be crafted by any role in a session (see \Cref{def:mpst-syntax-terms}),
we postulate that the $\stCrashLab$ message label cannot appear in internal choice types.

\subparagraph{Session Subtyping.} We use a subtyping relation $\stSub$ that is mostly standard:
a subtype can have wider internal choices and narrower external choices
\wrt a supertype.
To correctly support crash handling, we apply two changes:
\emph{(1)}
    we add the relation $\stStop \stSub \stStop$, and
\emph{(2)}
  we treat external choices with a singleton $\stCrashLab$ branch in a special way:
  they represent a ``pure'' crash recovery protocol (as outlined above),
  hence we do not allow the supertype to have more input branches.
  This way, a ``pure'' crash recovery type can only be implemented by
  a ``pure'' crash recovery process (with a singleton $\mpLabCrash$ detection branch);
  such processes are treated specially by the properties
  in \Cref{sec:dedlock-freedom}.
For the complete definition of $\stSub$, see \Cref{sec:subtyping}.

%% file: 4_mpst_crash/syntax.tex
\begin{definition}[Types]%
    \label{fig:syntax-local-type}%
    \label{fig:syntax-mpst}%
    \label{def:local-types}
    \label{def:ground-types}
    Our types include both basic types and \emph{session types}:

    \smallskip
    \centerline{\(
    \begin{array}{r@{\quad}c@{\quad}l@{\quad}l}
      \tyGround & \bnfdef & \tyInt \bnfsep \tyBool \bnfsep \tyReal \bnfsep \tyUnit \bnfsep \ldots
        & \text{\footnotesize (basic types)} \\
      \stS & \bnfdef & \tyGround \bnfsep \stT
        & \text{\footnotesize (basic type or session type)} \\
      \stT & \bnfdef & \stExtSum{\roleP}{i \in I}{\stChoice{\stLab[i]}{\stS[i]} \stSeq \stT[i]}
        \!\!\quad \bnfsep \!\!\quad \stIntSum{\roleP}{i \in I}{\stChoice{\stLab[i]}{\stS[i]} \stSeq \stT[i]}
          & \text{\footnotesize (external or internal choice, with $I \neq
          \emptyset$)} \\
        & \bnfsep & \stRec{\stRecVar}{\stT} \quad \bnfsep \quad \stRecVar
        \quad \bnfsep \quad \stEnd
        &
        \text{\footnotesize (recursion, type variable, or termination)} \\
      \stU & \bnfdef & \tyT \bnfsep \highlight{\stStop}
        & \text{\footnotesize (session type or $\highlight{$\text{crash type}$}$)}
    \end{array}
  \)}%
  \smallskip

  \noindent%
  In internal and external choices, the index set $I$ must be non-empty,
  and labels $\stLab[i]$ must be pair-wise distinct.
  Types are always \emph{closed} (\ie each recursion variable $\stRecVar$ is bound under a $\stRec{\stRecVar}{\ldots}$)
  and recursion variables are \emph{guarded}, \ie they can only appear
  under an internal/external choice (\eg $\stRec{\stRecVar}{\stRec{\stRecVari}{\stRecVar}}$ is not a valid type).
  For brevity, we may omit the payload type $\tyUnit$ and the trailing $\stEnd$: \eg
  $\stOut{\roleP}{\stLab[1]}{} \stSeq \stInNB{\roleR}{\stLab[2]}{}{}$
  is shorthand for $\stOut{\roleP}{\stLab[1]}{\tyUnit} \stSeq \stInNB{\roleR}{\stLab[2]}{\tyUnit}{} \stSeq \stEnd$.
\end{definition}

%% file: 4_mpst_crash/2_context.tex
\subsection{Typing Contexts and their Semantics}%
\label{sec:gtype:lts-context}
Before introducing the typing rules for our calculus (in
\cref{sec:type-system:tyrules}), we first formalise typing contexts
(\cref{def:mpst-env}) and their semantics (\cref{def:mpst-env-reduction}).

\input{4_mpst_crash/context_def.tex}

\input{4_mpst_crash/context_reduction.tex}

Unlike typical session typing systems,
our \Cref{def:mpst-env} allows a session endpoint $\mpChanRole{\mpS}{\roleP}$
to have either a session type $\stT$, or the crash type $\stStop$.
We equip our typing contexts with a labelled transition system (LTS) semantics
(in \Cref{def:mpst-env-reduction})
using the labels in \Cref{def:mpst-env-reduction-label}.

\input{4_mpst_crash/labels_def.tex}

\begin{definition}[Typing Context Semantics]%
  \label{def:mpst-env-reduction}%
  The \emph{typing context transition $\stEnvMoveGenAnnot$} %
  is defined in \cref{fig:gtype:tc-red-rules}.
  We write $\stEnvMoveGenAnnotP{\stEnv}$ %
  iff\, $\stEnv \!\stEnvMoveGenAnnot\! \stEnvi$ for some $\stEnvi$. %
  We define the two \emph{reductions} $\stEnvMove$ %
  and $\stEnvMoveMaybeCrash[\mpS; \rolesS]$ (where $\mpS$ is a session, and $\rolesS$ is a set of roles)
  as follows:
  \begin{itemize}%
    \item $\stEnv \!\stEnvMove\! \stEnvi$  \;holds iff\;
    $\stEnv \stEnvMoveCommAnnot{\mpS}{\roleP}{\roleQ}{\stLab} \stEnvi$ %
    \;or\;
    $\stEnv \stEnvMoveAnnot{\ltsCrDe{\mpS}{\roleQ}{\roleP}} \stEnvi$
    (for some $\mpS, \roleP, \roleQ, \stLab$).
    This means that $\stEnv$ can advance via message transmission or
    crash detection, but it \emph{cannot} advance by crashing one of its entries.
    We write\; $\stEnvMoveP{\stEnv}$ %
    \;iff\; $\stEnv \!\stEnvMove\! \stEnvi$ for some $\stEnvi$, %
    \;and\; $\stEnvNotMoveP{\stEnv}$ \;for its negation %
    (\ie there is no $\stEnvi$ such that %
    $\stEnv \!\stEnvMove\! \stEnvi$), and $\stEnvMoveStar$ %
    \;for the reflexive and transitive closure of $\stEnvMove$;

  \item $\stEnv \!\stEnvMoveMaybeCrash[\mpS; \rolesS]\! \stEnvi$
    \,holds iff\,
    $\stEnv \stEnvMoveGenAnnot \stEnvi$
    with $\stEnvAnnotGenericSym \!\in\! \setcomp{
      \begin{array}{@{}l@{}}
        \ltsSendRecv{\mpS}{\roleQ}{\roleR}{\stLab},\, %
        \ltsCrDe{\mpS}{\roleQ}{\roleR},\, \ltsCrash{\mpS}{\roleP}
      \end{array}}{\roleP,\roleQ,\roleR \!\in\! \roleSet,\, \roleP \!\not\in\! \rolesS}$.
    This means that $\stEnv$ can advance via message transmission or crash detection
    on session $\mpS$, involving any roles $\roleQ$ and $\roleR$.
    (Recall that $\roleSet$ is the set of all roles.)
    \emph{Moreover}, $\stEnv$ can advance by crashing one of its entries $\mpChanRole{\mpS}{\roleP}$ -- unless $\roleP \!\in\! \rolesS$, which means that $\roleP$ is assumed to be \emph{reliable}.
    \;We write\; $\stEnvMoveMaybeCrashP[\mpS; \rolesS]{\stEnv}$ %
    \;iff\; $\stEnv \!\stEnvMoveMaybeCrash[\mpS; \rolesS]\! \stEnvi$ for some $\stEnvi$, %
    \;and\; $\stEnvNotMoveMaybeCrashP[\mpS; \rolesS]{\stEnv}$ \;for its negation, %
    and\; $\stEnvMoveMaybeCrashStar[\mpS; \rolesS]$ %
    \;as the reflexive and transitive closure of $\stEnvMoveMaybeCrash[\mpS;
    \rolesS]$.
    We write $\stEnv \!\stEnvMoveMaybeCrash\! \stEnvi$
    iff $\stEnv \!\stEnvMoveMaybeCrash[\mpS; \gtFmt{\emptyset}]\! \stEnvi$
    for some $\mpS$ (\ie $\stEnv$ may advance by crashing any role on any session).
  \end{itemize}
\end{definition}

\Cref{def:mpst-env-reduction} subsumes the standard typing context
reductions%
~\cite[Def.\@ 2.8]{POPL19LessIsMore}.
Rule $\inferrule{\iruleTCtxOut}$ (resp.\ $\inferrule{\iruleTCtxIn}$)
says that an entry can perform an output (resp.\ input) transition.
Rule $\inferrule{\iruleTCtxCom}$ synchronises matching input/output
transitions, provided that the %
payloads are compatible by subtyping;
as a result, the context advances via a message transmission label
$\stEnvCommAnnotSmall{\roleP}{\roleQ}{\stLab}$.
Other standard rules are
$\inferrule{\iruleTCtxRec}$ for recursion, and
$\inferrule{\iruleTCtxCong}$ and $\inferrule{\iruleTCtxCongBasic}$ for reductions in a larger context.

The key innovations are the (\highlightText{highlighted}) rules modelling crashes and crash
detection.
By rule~\inferrule{\iruleTCtxCrash}, an entry can crash
and become $\stStop$ at any time
(unless it is already $\stEnd$ed or $\stStop$ped);
then, by rule~\inferrule{\iruleTCtxCrashed}, it keeps signalling that it is crashed,
with label $\actCrashed{\mpS}{\roleP}$.

Rule \inferrule{\iruleTCtxCrashDetect} models crash detection and handling:
if $\mpChanRole{\mpS}{\roleP}$ signals that it has crashed and stopped,
another entry $\mpChanRole{\mpS}{\roleQ}$ can then take its $\stCrashLab$ handling branch
(part of an external choice from $\roleP$). %
This corresponds to the process reduction rule \inferrule{\iruleMPRedCommD} for crash detection.

Finally, rule \inferrule{\iruleTCtxSendToCrashed} models the case where
the entry $\mpChanRole{\mpS}{\roleP}$ is sending a message $\stChoice{\stLab}{\stS}$ to a crashed
$\mpChanRole{\mpS}{\roleQ}$: this yields a transmission label $\stEnvCommAnnotSmall{\roleP}{\roleQ}{\stLab}$,
and $\roleP$ %
continues -- although the sent message is not actually received by crashed $\roleQ$.
This corresponds to the process reduction rule \inferrule{\iruleMPRedCommE} where a
process sends a message to a crashed endpoint, and cannot detect its crash.

%% file: 4_mpst_crash/context_def.tex
\begin{definition}[Typing Contexts]%
  \label{def:mpst-env}%
  \label{def:mpst-env-closed}%
  \label{def:mpst-env-comp}%
  \label{def:mpst-env-subtype}%
  $\mpEnv$ denotes a partial mapping
  from process variables to $n$-tuples of types,
  and $\stEnv$ denotes a partial mapping
  from channels to types.
  Their syntax is:

  \smallskip
  \centerline{\(%
  \mpEnv
  \;\;\coloncolonequals\;\;
  \mpEnvEmpty
  \bnfsep
  \mpEnv \mpEnvComp\, \mpEnvMap{\mpX}{\stS[1],\ldots,\stS[n]}
  \qquad\qquad
  \stEnv
  \;\;\coloncolonequals\;\;
  \stEnvEmpty
  \bnfsep
  \stEnv \stEnvComp \stEnvMap{x}{\stS}
  \bnfsep
  \stEnv \stEnvComp \stEnvMap{\mpChanRole{\mpS}{\roleP}}{\stU}
  \)}%
  \smallskip

  \noindent
  The \,\emph{context composition} $\stEnv[1] \stEnvComp \stEnv[2]$\,
  is defined iff $\dom{\stEnv[1]} \cap \dom{\stEnv[2]} = \emptyset$.
  \noindent
  We write\;
  $\mpS \!\not\in\! \stEnv$
  \;iff\;
  $\forall \roleP: \mpChanRole{\mpS}{\roleP} \!\not\in\! \dom{\stEnv}$
  (\ie session $\mpS$ does not occur in $\stEnv$).
  \noindent
  We write\;
  $\stEnv \!\stSub\! \stEnvi$
  \;iff
  $\dom{\stEnv} \!=\! \dom{\stEnvi}$
  and
  $\forall \mpC \!\in\! \dom{\stEnv}:
  \stEnvApp{\stEnv}{\mpC} \!\stSub\! \stEnvApp{\stEnvi}{\mpC}$.
\end{definition}

%% file: 4_mpst_crash/context_reduction.tex
\begin{figure}[t]
  \noindent
  \scalebox{0.9}{
  \begin{minipage}{1.1\linewidth}
  \centerline{\(
  \begin{array}{@{}c@{}}
    \inference[\iruleTCtxOut]{%
      k \in I%
    }{%
      \stEnvMap{%
        \mpChanRole{\mpS}{\roleP}%
      }{%
        \stIntSum{\roleQ}{i \in I}{\stChoice{\stLab[i]}{\stS[i]} \stSeq \stT[i]}%
      }%
      \,\stEnvMoveOutAnnot{\roleP}{\roleQ}{\stChoice{\stLab[k]}{\stS[k]}}\,%
      \stEnvMap{%
        \mpChanRole{\mpS}{\roleP}%
      }{%
        \stT[k]%
      }%
    }%
    \qquad%
    \inference[\iruleTCtxIn]{%
      k \in I%
    }{%
      \stEnvMap{%
        \mpChanRole{\mpS}{\roleP}%
      }{%
        \stExtSum{\roleQ}{i \in I}{\stChoice{\stLab[i]}{\stS[i]} \stSeq \stT[i]}%
      }%
      \,\stEnvMoveInAnnot{\roleP}{\roleQ}{\stChoice{\stLab[k]}{\stS[k]}}\,%
      \stEnvMap{%
        \mpChanRole{\mpS}{\roleP}%
      }{%
        \stT[k]%
      }%
    }%
    \\[2mm]%
    \inference[\iruleTCtxCom]{%
      \stEnv[1]%
      \stEnvMoveOutAnnot{\roleP}{\roleQ}{\stChoice{\stLab}{\stS}}%
      \stEnvi[1]%
      &%
      \stEnv[2]%
      \stEnvMoveInAnnot{\roleQ}{\roleP}{\stChoice{\stLab}{\stSi}}%
      \stEnvi[2]%
      &%
      \stS \!\stSub\! \stSi%
    }{%
      \stEnv[1] \stEnvComp \stEnv[2]%
      \,\stEnvMoveCommAnnot{\mpS}{\roleP}{\roleQ}{\stLab}\,%
      \stEnvi[1] \stEnvComp \stEnvi[2]%
    }%
    \quad
    \inference[\iruleTCtxRec]{%
      \stEnvMap{%
        \mpChanRole{\mpS}{\roleP}%
      }{%
        \stT\subst{\stRecVar}{\stRec{\stRecVar}{\stT}}%
      }%
      \stEnvMoveGenAnnot \stEnvi%
    }{%
      \stEnvMap{%
        \mpChanRole{\mpS}{\roleP}%
      }{%
        \stRec{\stRecVar}{\stT}%
      }%
      \stEnvMoveGenAnnot \stEnvi%
    }%
    \quad
    \inference[\iruleTCtxCong]{%
      \stEnv \stEnvMoveGenAnnot \stEnvi%
    }{%
      \stEnv \!\stEnvComp \stEnvMap{\mpC}{\stU}
      \stEnvMoveGenAnnot%
      \stEnvi \!\stEnvComp \stEnvMap{\mpC}{\stU}
    }%
    \\[1mm]
    \highlight{%
    \inference[\iruleTCtxCrash]{%
      \stT \stNotSub \stEnd %
    }{
      \stEnvMap{%
        \mpChanRole{\mpS}{\roleP}%
      }{\stT}
      \stEnvMoveAnnot{\ltsCrash{\mpS}{\roleP}}
      \stEnvMap{%
        \mpChanRole{\mpS}{\roleP}%
      }{\stStop}
    }
    }%
    \qquad
    \highlight{%
    \inference[\iruleTCtxCrashed]{
      \vphantom{X}
    }{
      \stEnvMap{\mpChanRole{\mpS}{\roleP}}{\stStop}\,\stEnvMoveAnnot{\actCrashed{\mpS}{\roleP}}\,\stEnvMap{\mpChanRole{\mpS}{\roleP}}{\stStop}
    }
    }%
    \qquad
    \inference[\iruleTCtxCongBasic]{%
      \stEnv \stEnvMoveGenAnnot \stEnvi%
    }{%
      \stEnv \!\stEnvComp \stEnvMap{\mpFmt{x}}{\tyGround}
      \stEnvMoveGenAnnot%
      \stEnvi \!\stEnvComp \stEnvMap{\mpFmt{x}}{\tyGround}
    }%
    \\[1mm]
    \highlight{%
    \inference[\iruleTCtxCrashDetect]{%
      \stEnv[1]%
      \stEnvMoveInAnnot{\roleQ}{\roleP}{\stCrashLab}%
      \stEnvi[1]%
      &
      \stEnv[2]%
      \stEnvMoveAnnot{\actCrashed{\mpS}{\roleP}}
      \stEnvi[2]%
    }{%
      \stEnv[1] \stEnvComp \stEnv[2]%
      \,\stEnvMoveAnnot{\ltsCrDe{\mpS}{\roleQ}{\roleP}}\,%
      \stEnvi[1] \stEnvComp \stEnvi[2]%
    }%
    }%
    \qquad
    \highlight{%
    \inference[\iruleTCtxSendToCrashed]{%
      \stEnv[1]%
      \stEnvMoveOutAnnot{\roleP}{\roleQ}{\stChoice{\stLab}{\stS}}%
      \stEnvi[1]%
      &
      \stEnv[2]%
      \stEnvMoveAnnot{\actCrashed{\mpS}{\roleQ}}
      \stEnvi[2]%
    }{%
      \stEnv[1] \stEnvComp \stEnv[2]%
      \,\stEnvMoveCommAnnot{\mpS}{\roleP}{\roleQ}{\stLab}\,%
      \stEnvi[1] \stEnvComp \stEnvi[2]%
    }%
    }%
  \end{array}
  \)}%
  \end{minipage}
  }%
  %
  \caption{Typing context semantics.}
  \label{fig:gtype:tc-red-rules}
\end{figure}

%% file: 4_mpst_crash/labels_def.tex
\begin{definition}[Transition Labels]
  \label{def:mpst-env-reduction-label}%
  \label{def:mpst-label-subject}%
  Let $\stEnvAnnotGenericSym$ %
  denote a transition label having the form:

  \smallskip%
  \centerline{\(%
  \begin{array}{rcll}
  \stEnvAnnotGenericSym &\bnfdef&
    \stEnvInAnnotSmall{\roleP}{\roleQ}{\stChoice{\stLab}{\stS}}&
    \text{(in session $\mpS$, $\roleP$ receives message $\stChoice{\stLab}{\stS}$ from $\roleQ$; we omit $\stS$ if $\stS = \tyUnit$)}
    \\
    &\bnfsep&\stEnvOutAnnotSmall{\roleP}{\roleQ}{\stChoice{\stLab}{\stS}}&
    \text{(in session $\mpS$, $\roleP$ sends message $\stChoice{\stLab}{\stS}$ to $\roleQ$; we omit $\stS$ if $\stS = \tyUnit$)}
    \\
    &\bnfsep&\stEnvCommAnnotSmall{\roleP}{\roleQ}{\stLab}&
    \text{(in session $\mpS$, message $\stLab$ is transmitted from $\roleP$ to $\roleQ$)}
    \\
    &\bnfsep&\ltsCrashSmall{\mpS}{\roleP}&
    \text{(in session $\mpS$, $\roleP$ crashes)}
    \\
    &\bnfsep&\ltsCrDe{\mpS}{\roleP}{\roleQ}&
    \text{(in session $\mpS$, $\roleP$ has detected that $\roleQ$ has crashed)}
    \\
    &\bnfsep&\actCrashed{\mpS}{\roleP}&
    \text{(in session $\mpS$, $\roleP$ has stopped due to a crash)}
  \end{array}
\)}%
\end{definition}

%% file: 4_mpst_crash/3_safety.tex
\subsection{Typing Context Safety}%
\label{sec:type-system-safety}

To ensure type safety
(\cref{cor:type-safety}),
\ie well-typed processes do not result in $\mpErr$ors,
we define a safety
property $\predPApp{\cdot}$
(\cref{def:mpst-env-safe})
as a predicate on typing contexts $\stEnv$.
The safety property $\predP$ is the key feature of generalised MPST
systems~\cite[Def.~4.1]{POPL19LessIsMore};
in this work, we extend its definition in two crucial ways:
    \emph{(1)} we support crashes and crash detection, and
    \emph{(2)} we make the property parametric upon a (possibly empty) set of
    \emph{reliable} roles $\rolesS$, thus introducing \emph{optional reliability
      assumptions} about roles in a session that never fail.

\begin{definition}[Typing Context Safety]\label{def:mpst-env-safe}%
  Given a set of reliable roles $\rolesS$ and a session $\mpS$, we say that
  $\predP$ is an \emph{$(\mpS;\rolesS)$-safety property} of typing contexts %
  iff, whenever $\predPApp{\stEnv}$, we have:

  \noindent%
  \begin{tabular}{@{\;\;}r@{\hskip 2mm}l}
    \inferrule{\iruleSafeComm}%
    &%
    $\stEnvMoveAnnotP{\stEnv}{\stEnvOutAnnot{\roleP}{\roleQ}{\stChoice{\stLab}{\stS}}}$
    \,and\,
    $\stEnvMoveAnnotP{\stEnv}{\stEnvInAnnot{\roleQ}{\roleP}{\stChoice{\stLabi}{\stSi}}}$
    \;\;implies\;\; %
    $\stEnvMoveAnnotP{\stEnv}{\stEnvCommAnnotSmall{\roleP}{\roleQ}{\stLab}}$;
    \\%
    \inferrule{\iruleSafeCrash}%
    &%
    $\stEnvMoveAnnotP{\stEnv}{\actCrashed{\mpS}{\roleP}}$
    \,and\,
    $\stEnvMoveAnnotP{\stEnv}{\stEnvInAnnot{\roleQ}{\roleP}{\stChoice{\stLab}{\stS}}}$
    \;\;implies\;\; %
    $\stEnvMoveAnnotP{\stEnv}{\ltsCrDe{\mpS}{\roleQ}{\roleP}}$;
    \\
    \inferrule{\iruleSafeMove}%
    &%
    $\stEnv \stEnvMoveMaybeCrash[\mpS; \rolesS] \stEnvi$
    \;\;implies\;\; %
    $\predPApp{\stEnvi}$.
  \end{tabular}

  \smallskip

  \noindent%
  We say \emph{$\stEnv$ is $(\mpS;\rolesS)$-safe}, %
  written $\stEnvSafeSessRolesSP{\mpS}{\rolesS}{\stEnv}$, %
  if $\predPApp{\stEnv}$ holds %
  for some $(\mpS;\rolesS)$-safety property $\predP$. %
  We say \emph{$\stEnv$ is safe}, %
  written $\stEnvSafeP{\stEnv}$, %
  if $\predPApp{\stEnv}$ holds %
  for some property $\predP$ which is an $(\mpS;\gtFmt{\emptyset})$-safety property
  for all sessions $\mpS$ occurring in $\dom{\stEnv}$.
\end{definition}

  By \Cref{def:mpst-env-safe}, safety is a \emph{coinductive} property
  \cite{SangiorgiBiSimCoInd}: fix $\mpS$ and $\rolesR$,
  $(\mpS;\rolesS)$-safe is the largest ($\mpS; \rolesR$)-safety property,
  \ie the union of all ($\mpS; \rolesR$)-safety properties; to prove that some $\stEnv$ is $(\mpS;\rolesS)$-safe,
  we must find a property $\predP$ such that $\stEnv \!\in\! \predP$,
  and prove that $\predP$ is an ($\mpS; \rolesR$)-safety property.
  Intuitively, we can construct such $\predP$ (if it exists)
  as the set containing $\stEnv$ and all its reductums
  (via transition $\stEnvMoveMaybeCrashStar[\mpS;\rolesS]$),
  and checking whether all elements of $\predP$ %
  satisfy all clauses of \cref{def:mpst-env-safe}.
By clause~\inferrule{\iruleSafeComm}, whenever two roles $\roleP$ and $\roleQ$
attempt to communicate, the communication must be possible, \ie the receiver
$\roleQ$ must support all output messages of
sender $\roleP$, with compatible payload types (by rule
\inferrule{\iruleTCtxCom} in \Cref{fig:gtype:tc-red-rules}).
For ``pure'' crash recovery types (with a singleton $\stCrashLab$ handling branch)
there would not be corresponding sender, so this clause holds trivially.
Clause~\inferrule{\iruleSafeCrash} states that if a role $\roleQ$ receives
from a crashed role $\roleP$, then $\roleQ$ must have a $\stCrashLab$ %
handling branch.
Clause \inferrule{\iruleSafeMove}
states that any typing context $\stEnvi$ that $\stEnv$ transitions to (on
session $\mpS$) must also be in $\predP$ (hence, $\stEnvi$ must also be $(\mpS; \rolesS)$-safe);
notice that, by using transition $\stEnvMoveMaybeCrash[\mpS;\rolesS]$, we
ignore crashes $\ltsCrash{\mpS}{\roleP}$ of
any reliable role $\roleP \!\in\! \rolesR$.%

\begin{example}
Consider the simple DNS scenario from \cref{sec:overview},
its types $\stTi[\roleP]$, $\stTi[\roleQ]$ and $\stTi[\roleR]$, and
the typing context\;
\(
  \stEnv = \stEnvMap{
      \mpChanRole{\mpS}{\roleP}
    }{
      \stTi[\roleP]
    }
    \stEnvComp
    \stEnvMap{
      \mpChanRole{\mpS}{\roleQ}
    }{
      \stTi[\roleQ]
    }
    \stEnvComp
    \stEnvMap{
      \mpChanRole{\mpS}{\roleR}
    }{
      \stTi[\roleR]
    }
\).\;
We know, and can verify, that $\stEnv$ is $(\mpS;\setenum{\roleP,\roleR})$-safe by
checking its reductions. For example, for the case where $\roleQ$ crashes immediately, we
have:
\(\small
  \stEnv \stEnvMoveMaybeCrash[\mpS; \setenum{\roleP,\roleR}] \stEnvMap{%
      \mpChanRole{\mpS}{\roleP}%
    }{%
      \stTi[\roleP]
    }
    \;\stEnvComp\;
    \stEnvMap{
      \mpChanRole{\mpS}{\roleQ}
    }{
      \stStop
    }
    \;\stEnvComp\;
    \stEnvMap{
      \mpChanRole{\mpS}{\roleR}
    }{
      \stTi[\roleR]
    }
\)
\(\small
    \stEnvMoveMaybeCrashStar[\mpS; \setenum{\roleP,\roleR}]
    \stEnvMap{%
      \mpChanRole{\mpS}{\roleP}%
    }{%
      \stEnd
    }
    \stEnvComp\,
    \stEnvMap{
      \mpChanRole{\mpS}{\roleQ}
    }{
      \stStop
    }
    \stEnvComp\,
    \stEnvMap{
      \mpChanRole{\mpS}{\roleR}
    }{
      \stEnd
    }
\)
and each reductum satisfies all clauses of \Cref{def:mpst-env-safe}.
Full reductions are available in \cref{sec:appendix:examples},
\cref{eg:running:type-safe}.
\end{example}

%% file: 4_mpst_crash/4_typesystem.tex
\input{4_mpst_crash/mpst-typing-rules.tex}

\subsection{Typing Rules}%
\label{sec:type-system:tyrules}

Our type system uses two kinds of typing contexts (introduced in
\cref{def:mpst-env}):
$\mpEnv$ to assign an $n$-tuple of types to each process variable $\mpX$
(one type per argument),
and $\stEnv$
to map variables to payload types (basic types or session types),
and channels with roles
to session types or the $\stStop$ type.
Together, they are used in judgements of the form:

\smallskip\centerline{%
\(\stJudge{\mpEnv}{\stEnv}{\mpP}\)
\;\;(with $\mpEnv$ omitted when empty)
}\smallskip

\noindent%
which reads,
``given the process types in $\mpEnv$,
$\mpP$ uses its variables and channels \emph{linearly}
according to $\stEnv$.''
This \emph{typing judgement} %
is %
defined
by the rules in \Cref{fig:mpst-rules}, where, for convenience, we type-annotate
channels bound by process definitions and restrictions.

The main innovations in \Cref{fig:mpst-rules} are rules
\inferrule{\iruleMPResProp} and \inferrule{\iruleMPStop} (\highlightText{highlighted}).
Rule \inferrule{\iruleMPResProp} utilises a safety property $\predP$
(\cref{def:mpst-env-safe}) to validate session
restrictions, taking into account crashes and crash handling,
and any reliable role assumption in the (possibly empty) set $\rolesR$.
The rule can be instantiated by choosing a set $\rolesR$ and safety property $\predP$
(\eg among the stronger properties presented in \Cref{def:typing-ctx-properties} later on).
Rule \inferrule{\iruleMPStop}
types crashed session endpoints as $\stStop$.

The rest of the rules in \Cref{fig:mpst-rules} are mostly standard.
\inferrule{\iruleMPX} looks up process variables.
\inferrule{\iruleMPGround} types a value $\mpV$ if it belongs to a basic type $\tyGround$.
\inferrule{\iruleMPSub} holds for a singleton typing context
$\stEnvMap{\mpC}{\stS}$, and applies subtyping when assigning a type
$\stSi$ to a variable or channel $\mpC$.
\inferrule{\iruleMPEnd} defines a predicate $\stEnvEndP{\cdot}$ on typing
contexts, indicating all endpoints are terminated -- it is used in
\inferrule{\iruleMPNil} for typing an inactive process $\mpNil$, and in
$\inferrule{\iruleMPStop}$ for crashed endpoints.
\inferrule{\iruleMPSel} and \inferrule{\iruleMPBranch} assign
selection and branching types to channels used by selection and branching processes. %
Minor changes \wrt standard session types are the clauses ``$\stS \stNotSub \stEnd$''
in rules \inferrule{\iruleMPSel} and \inferrule{\iruleMPCall}: they forbid sending or passing
$\stEnd$-typed channels, while allowing sending/passing channels and data of any other type.\footnote{%
  This restriction is needed for \Cref{lem:subject-reduction,lem:stenv-proc-properties}.
  It does not limit the expressiveness of our typed calculus,
  since sending an $\stEnd$-typed channel (not usable for communication)
  amounts to sending a basic value.%
}
Rules $\inferrule{\iruleMPDef}$ and $\inferrule{\iruleMPCall}$ handle
recursive processes declarations and calls.
$\inferrule{\iruleMPPar}$ \emph{linearly} splits the typing context into two,
one for typing each sub-process.

%% file: 4_mpst_crash/mpst-typing-rules.tex
\begin{figure}[t]
  \small
  \centerline{\(
  \begin{array}{c}
    \inference[\iruleMPX]{
      \mpEnvApp{\mpEnv}{X} = \stFmt{\stS[1],\ldots,\stS[n]}
    }{
      \mpEnvEntails{\mpEnv}{X}{\stS[1],\ldots,\stS[n]}
    }
    \qquad
    \inference[\iruleMPGround]{
      \mpV \in \tyGround
    }{
      \stEnvEntails{\stEnvEmpty}{\mpV}{\tyGround}
    }
    \qquad
    \inference[\iruleMPEnd]{
      \forall i \in 1..n
      &
      \text{$\stS[i]$ is basic \;or\; }
      \stEnvEntails{\stEnvMap{\mpC[i]}{\stS[i]}}{
        \mpC[i]
      }{
        \stEnd
      }
    }{
      \stEnvEndP{
        \stEnvMap{\mpC[1]}{\stS[1]}
        \stEnvComp \ldots \stEnvComp
        \stEnvMap{\mpC[n]}{\stS[n]}
      }
    }
    \\[1mm]
    \inference[\iruleMPCall]{
        \mpEnvEntails{\mpEnv}{X}{
          \stS[1],\ldots,\stS[n]
        }
        &
        \stEnvEndP{\stEnv[0]}
        &
        \forall i \in 1..n
        &
        \stEnvEntails{\stEnv[i]}{\mpD[i]}{\stS[i]}
        &
        \highlight{\stS[i] \stNotSub \stEnd}
    }{
      \stJudge{\mpEnv}{
        \stEnv[0] \stEnvComp
        \stEnv[1] \stEnvComp \ldots \stEnvComp \stEnv[n]
      }{
        \mpCall{\mpX}{\mpD[1],\ldots,\mpD[n]}
      }
    }
    \\[2mm]
    \inference[\iruleMPNil]{
      \stEnvEndP{\stEnv}
    }{
      \stJudge{\mpEnv}{\stEnv}{\mpNil}
    }
    \qquad
    \inference[\iruleMPDef]{
        \stJudge{
          \mpEnv \mpEnvComp
          \mpEnvMap{\mpX}{\stS[1],\ldots,\stS[n]}
        }{
          \stEnvMap{x_1}{\stS[1]}
          \stEnvComp \ldots \stEnvComp
          \stEnvMap{x_n}{\stS[n]}
        }{
          \mpP
        }
        \qquad
        \stJudge{
          \mpEnv \mpEnvComp
          \mpEnvMap{\mpX}{\stS[1],\ldots,\stS[n]}
        }{
          \stEnv
        }{
          \mpQ
        }
    }{
      \stJudge{\mpEnv}{
        \stEnv
      }{
        \mpDef{\mpX}{
          \stEnvMap{x_1}{\stS[1]},
          \ldots,
          \stEnvMap{x_n}{\stS[n]}
        }{\mpP}{\mpQ}
      }
    }
    \\[2mm]
    \inference[\iruleMPBranch]{
        \stEnvEntails{\stEnv[1]}{\mpC}{
          \stExtSum{\roleQ}{i \in I}{\stChoice{\stLab[i]}{\tyS[i]} \stSeq \stT[i]}%
        }
        &
        \forall i \!\in\! I
        &
        \stJudge{\mpEnv}{
          \stEnv \stEnvComp
          \stEnvMap{y_i}{\tyS[i]} \stEnvComp
          \stEnvMap{\mpC}{\stT[i]}
        }{
          \mpP[i]
        }
    }{
      \stJudge{\mpEnv}{
        \stEnv \stEnvComp \stEnv[1]
      }{
        \mpBranch{\mpC}{\roleQ}{i \in I}{\mpLab[i]}{y_i}{\mpP[i]}{}
      }
    }
    \qquad
    \inference[\iruleMPPar]{
      \stJudge{\mpEnv}{
        \stEnv[1]
      }{
        \mpP[1]
      }
      \qquad
      \stJudge{\mpEnv}{
        \stEnv[2]
      }{
        \mpP[2]
      }%
    }{
      \stJudge{\mpEnv}{
        \stEnv[1] \stEnvComp \stEnv[2]
      }{
        \mpP[1] \mpPar \mpP[2]
      }
    }
    \\[2mm]
    \inference[\iruleMPSel]{
      \stEnvEntails{\stEnv[1]}{\mpC}{
        \stIntSum{\roleQ}{}{\stChoice{\stLab}{\stS} \stSeq \stT}
      }
      &
      \stEnvEntails{\stEnv[2]}{\mpD}{\tyS}
      &
      \highlight{\stS \stNotSub \stEnd}
      &
      \stJudge{\mpEnv}{
        \stEnv \stEnvComp \stEnvMap{\mpC}{\stT}
      }{
        \mpP
      }
    }{
      \stJudge{\mpEnv}{
        \stEnv \stEnvComp \stEnv[1] \stEnvComp \stEnv[2]
      }{
        \mpSel{\mpC}{\roleQ}{\mpLab}{\mpD}{\mpP}
      }
    }
    \qquad
    \inference[\iruleMPSub]{
      \stS \stSub \stSi
    }{
      \stEnvEntails{\stEnvMap{\mpC}{\stS}}{\mpC}{\stSi}
    }
    \\[2mm]
    \highlight{%
    \inference[\iruleMPStop]{
      \stEnvEndP{\stEnv}
    }{
      \stJudge{\mpEnv}{\stEnv,\stEnvMap{\mpChanRole{\mpS}{\roleP}}{\stStop}}{
        \mpStop{\mpS}{\roleP}
      }
    }
    }%
    \qquad
    \highlight{%
    \inference[\iruleMPResProp]{
      \stEnvi = \setenum{
        \stEnvMap{\mpChanRole{\mpS}{\roleP}}{\stT[\roleP]}
      }_{\roleP \in I}
      \quad
      \predPApp{\stEnvi}
      \quad
      \mpS \!\not\in\! \stEnv
      \quad%
      \stJudge{\mpEnv}{
        \stEnv \stEnvComp \stEnvi
      }{
        \mpP
      }
    }{
      \stJudge{\mpEnv}{
        \stEnv
      }{
        \mpRes{\stEnvMap{\mpS}{\stEnvi}}\mpP
      }
    }
    }%
  \end{array}
  \)}%
  %
  %
  %
  %
  \caption{
    Typing rules for processes;
    $\predP$ in \inferrule{\iruleMPResProp} is an ($\mpS$;$\rolesR$)-safety property, for some $\rolesR$.
  }
  \label{fig:mpst-rules}
  %
\end{figure}

%% file: 4_mpst_crash/5_subjred_sessfid.tex
\subsection{Subject Reduction and Session Fidelity}%
\label{sec:type-system:subj-red}

We present our key results on typed processes:
\emph{subject reduction} and \emph{session fidelity}
(\Cref{lem:subject-reduction,lem:session-fidelity}).
A main feature of our theory is that our results explicitly account for the
\emph{spectrum} of optional reliability assumptions used during typing.
\begin{itemize}%
  \item On one end of the spectrum, our results hold without any reliability
    assumption: any process and session endpoint may crash at any time. This is
    obtained if, for each $\stEnv$ used during typing, we assume
    $\stEnvSafeP{\stEnv}$ (\Cref{def:mpst-env-safe}), with no reliable roles.
  \item At the other end of the spectrum, we recover the classic MPST results
    by assuming that all roles in all sessions are reliable -- \ie if for each
    $\stEnv$ used during typing, and for all $\mpS \!\in\! \stEnv$, we assume
    $\stEnvSafeSessRolesSP{\mpS}{\rolesR[\mpS]}{\stEnv}$ with $\rolesR[\mpS] =
    \setcomp{\roleP\,}{\,\mpChanRole{\mpS}{\roleP} \!\in\! \dom{\stEnv}}$.
\end{itemize}

\smallskip
\noindent
\emph{Subject reduction} (\Cref{lem:subject-reduction} below) states that if
a well-typed process $\mpP$ reduces to $\mpPi$, then the reduction is simulated
by its typing context $\stEnv$, provided that the reliability assumptions
embedded in $\stEnv$ hold when $\mpP$ reduces.
In other words, if a channel endpoint $\mpChanRole{\mpS}{\roleP}$ occurring in
$\mpP$ is assumed reliable in $\stEnv$, then $\mpP$ should \emph{not} crash
$\mpChanRole{\mpS}{\roleP}$ while reducing;
any other reduction of $\mpP$ (including those that crash other session
endpoints) are type-safe.
To formalise this idea, we define \emph{reliable process reduction
$\mpMoveMaybeCrash[\mpS;\rolesR]$} as a subset of $\mpP$'s reductions.
We also define \emph{assumption-abiding reduction $\mpMoveMaybeCrashChecked$} to
enforce reliable process reductions across nested sessions.

\begin{definition}[Reliable Process Reductions and Assumption-Abiding
  Reductions]
  \label{def:reliable-proc-reduction}
  \label{def:assumption-abiding-reduction}
  The \emph{reliable process reduction} $\mpMoveMaybeCrash[\mpS;\rolesR]$ is defined
  as follows:
  
  \smallskip\centerline{\(
    \inference{
      \mpP \mpMoveMaybeCrash \mpPi
      \qquad
      \forall \roleP \in \rolesR:\; \nexists \mpR:\; \mpPi \equiv \mpR \mpPar \mpStop{\mpS}{\roleP}
    }{
      \mpP \;\mpMoveMaybeCrash[\mpS;\rolesR]\; \mpPi
    }
  \)}\smallskip

  \noindent
  Assume\, $\stJudge{\mpEnv}{\stEnv}{\mpP}$ \,where,
  for each $\mpS \!\in\! \stEnv$, there is a set of reliable roles $\rolesR[\mpS]$
  such that $\stEnvSafeSessRolesSP{\mpS}{\rolesR[\mpS]}{\stEnv}$.
  We define the \emph{assumption-abiding} reduction
  $\mpMoveMaybeCrashChecked$ such that\,
  $\mpP \!\mpMoveMaybeCrashChecked\! \mpPi$ \,holds when:
  \emph{(1)}
    $\mpP \!\mpMoveMaybeCrash[\mpS;{\rolesR[\mpS]}]\! \mpPi$
    \,for all $\mpS \in \stEnv$; and
  \emph{(2)}
    if $\mpP \equiv \mpRes{\stEnvMap{\mpSi}{\stEnv[\mpSi]}}\mpQ$ \,(for some $\mpSi,\stEnv[\mpSi],\mpQ$)
    \,and\, $\mpPi \equiv \mpRes{\mpSi}{\mpQi}$ \,and\, $\mpQ \mpMoveMaybeCrash \mpQi$,
    \;then $\exists \rolesRi$ such that\, $\stEnvSafeSessRolesSP{\mpSi}{\rolesRi}{\stEnv[\mpSi]}$
    \;and\; $\mpQ \mpMoveMaybeCrash[\mpSi;\rolesRi] \mpQi$.
  We write $\mpMoveMaybeCrashCheckedPlus$\slash\,$\mpMoveMaybeCrashCheckedStar$ for the transitive\,\slash\,reflexive-transitive closure of $\mpMoveMaybeCrashChecked$.
\end{definition}

\noindent%
Hence, when $\mpP \mpMoveMaybeCrash[\mpS;\rolesR] \mpPi$ holds, none of the session endpoints
$\mpChanRole{\mpS}{\roleP}$ (where $\roleP$ is a reliable role in set
$\rolesR$) are crashed in $\mpPi$.
When $\mpP$ is well-typed, the reduction $\mpP \mpMoveMaybeCrashChecked \mpPi$ covers all (and only) the reductions of $\mpP$
that do not violate any reliability assumption used for deriving $\stJudge{\mpEnv}{\stEnv}{\mpP}$;
notice that we use congruence $\equiv$ to quantify over all restricted
sessions in $\mpP$ and ensure their reductions respect all reliability assumptions in their typing,
by \inferrule{\iruleMPResProp} in \Cref{fig:mpst-rules}.

We can now use $\mpMoveMaybeCrashChecked$ to state our subject reduction
result.
\iftoggle{techreport}{Its proof is available in \cref{sec:proofs:subject-reduction}.}{Its proof is available in~\cite{techreport}.}

\begin{restatable}[Subject Reduction]{theorem}{lemSubjectReduction}
  \label{lem:subject-reduction}%
  Assume\, $\stJudge{\mpEnv}{\stEnv}{\mpP}$ \,where
  $\forall \mpS \in \stEnv: \exists
  \rolesR[\mpS]:$ $\stEnvSafeSessRolesSP{\mpS}{\rolesR[\mpS]}{\stEnv}$.
  If\, $\mpP
  \mpMoveMaybeCrashChecked \mpPi$,
  \,then\,
  $\exists \stEnvi$
  such that\,
  $\stEnv \stEnvMoveMaybeCrashStar \stEnvi$,
  \,and\,
  $\forall \mpS \in \stEnvi: \stEnvSafeSessRolesSP{\mpS}{\rolesR[\mpS]}{\stEnvi}$,
  \,and\,
  $\stJudge{\mpEnv}{\stEnvi}{\mpPi}$.%
\end{restatable}

\vspace{-2ex}
\begin{restatable}[Type Safety]{corollary}{lemTypeSafety}
  \label{cor:type-safety}
  Assume $\stJudge{\emptyset}{\emptyset}{\mpP}$.
  If\, $\mpP \mpMoveMaybeCrashCheckedStar \mpPi$,
  \,then %
  $\mpPi$ has no error.%
\end{restatable}

\begin{example}[Subject reduction]
Take the DNS example (\cref{sec:overview}) and
consider the process acting as the (unreliable) role $\roleQ$:
$
\mpP[\roleQ] = \mpBranchSingle{\mpChanRole{\mpS}{\roleQ}}{\roleP}{req}{}{\mpSel{\mpChanRole{\mpS}{\roleQ}}{\roleP}{res}{}{\mpNil}}
$.
Using type $\stTi[\roleQ]$ from the same example,
can type $\mpP[\roleQ]$ with the typing context $\stEnv[\roleQ] =
\stEnvMap{\mpChanRole{\mpS}{\roleQ}}{\stTi[\roleQ]}$.
Following a crash reduction via \inferrule{\iruleMPCrashR}, the
process evolves as $\mpP[\roleQ] \mpMoveCrash \mpPi[\roleQ] = \mpStop{\mpS}{\roleQ}$.
Observe that the typing context $\stEnv[\roleQ]$ can reduce to $\stEnvi[\roleQ] = \stEnvMap{\mpChanRole{\mpS}{\roleP}}{\stStop}$,
via \inferrule{\iruleTCtxCrash};
and by typing rule \inferrule{\iruleMPStop}, we can type
$\mpPi[\roleQ]$ with $\stEnvi[\roleQ]$.
\end{example}

\emph{Session fidelity} states the opposite implication \wrt subject
reduction: if a process $\mpP$ is typed by $\stEnv$, and $\stEnv$ can
reduce along session $\mpS$ (possibly by crashing some endpoint of $\mpS$),
then $\mpP$ can
reproduce at least one of the reductions of $\stEnv$ (but maybe not all such
reductions, because $\stEnv$ over-approximates the behaviour of $\mpP$).
As a consequence, we can infer $\mpP$'s behaviour from $\stEnv$'s behaviour,
as shown in \Cref{lem:stenv-proc-properties}.
This result does \emph{not} hold for all well-typed processes: %
a well-typed process
can loop in a recursion like $\mpDef{\mpX}{...}{X}{\mpX}$, or
deadlock by suitably interleaving its communications across multiple
sessions~\cite{CDYP2015}.
Thus, similarly to
\cite{POPL19LessIsMore} and most session type works,
we prove session fidelity for processes with guarded recursion,
and implementing a single multiparty session as a parallel
composition of one sub-process per role.
Session fidelity is given in
\Cref{lem:session-fidelity} below, by leveraging \Cref{def:unique-role-proc}.

\begin{definition}[from \cite{POPL19LessIsMore}]
  \label{lem:guarded-definitions}
  \label{def:unique-role-proc}
  Assume\; $\stJudge{\mpEnvEmpty}{\stEnv}{\mpP}$.
  \;We say that $\mpP$:
  \begin{enumerate}[(1)]
  \item\label{item:guarded-definitions:stmt}%
    \textbf{has guarded definitions}
    iff
    in each process definition in $\mpP$ of the form
    \linebreak
    $\mpDef{\mpX}{
      \stEnvMap{x_1}{\stS[1]},...,\stEnvMap{x_n}{\stS[n]}}
      {\mpQ}{\mpPi}$,
    \,for all $i \in 1..n$,\,
    if $\stS[i]$ is a session type, then a call
    \linebreak
    $\mpCall{\mpY}{...,x_i,...}$
    can only occur in $\mpQ$
    as a subterm of\;
    $\mpBranch{x_i}{\roleQ}{j \in J}{\mpLab[j]}{y_j}{\mpP[j]}{}$
    \,or\,
    $\mpSel{x_i}{\roleQ}{\mpLab}{\mpD}{\mpPii}{}$
    (\ie after using $x_i$ for input or output);
  \item\label{item:unique-role-proc:stmt}%
    \textbf{only plays role $\roleP$ in $\mpS$, by $\stEnv$}
    \,iff:
    \!\!\!\descriptionlabel{\upshape (i)}
    $\mpP$ has guarded definitions;\; %
    \!\!\!\descriptionlabel{\upshape (ii)}
    $\fv{\mpP} \!=\! \emptyset$;\;
    \!\!\!\descriptionlabel{\upshape (iii)}
      $\stEnv \!=\!
      \stEnv[0] \stEnvComp \stEnvMap{\mpChanRole{\mpS}{\roleP}}{\stS}$
      with
      $\stS \!\stNotSub\! \stEnd$
      and
      $\stEnvEndP{\stEnv[0]}$;\;
    \!\!\!\descriptionlabel{\upshape (iv)}
      for all subterms
      $\mpRes{\stEnvMap{\mpSi}{\stEnvi}}{\mpPi}$
      in $\mpP$,
      $\stEnvEndP{\stEnvi}$.
  \end{enumerate}
  We say ``\emph{$\mpP$ only plays role $\roleP$ in $\mpS$}''
  \,iff\,
  $\exists\stEnv: \stJudge{\mpEnvEmpty}{\stEnv}{\mpP}$,
  and item~\ref{item:unique-role-proc:stmt} holds.
\end{definition}

Item~\ref{item:guarded-definitions:stmt}
of \Cref{lem:guarded-definitions} formalises guarded recursion for processes.
Item~\ref{item:unique-role-proc:stmt} %
identifies a process that plays
exactly \emph{one} role on \emph{one} session;
clearly,
an ensemble of such processes
cannot deadlock
by waiting for each other on multiple sessions.
All our examples %
satisfy \Cref{def:unique-role-proc}(\ref{item:unique-role-proc:stmt}).

We can now formalise our session fidelity result (\Cref{lem:session-fidelity}).
The statement is superficially similar to Thm.\@ 5.4 in
\cite{POPL19LessIsMore},
but it now includes explicit
reliability assumptions for $\stEnv$; %
it also covers more cases,
since our typing contexts and processes can reduce by crashing, handling
crashes, or losing messages sent to crashed session endpoints.
\iftoggle{techreport}{The proof is available in \Cref{sec:proofs:session-fidelity}.}{The proof is available in~\cite{techreport}.}

\begin{restatable}[Session Fidelity]{theorem}{lemSessionFidelity}%
  \label{lem:session-fidelity}%
  Assume\, $\stJudge{\mpEnvEmpty\!}{\!\stEnv}{\!\mpP}$, %
  with\, %
  $\stEnvSafeSessRolesSP{\mpS}{\rolesR}{\stEnv}$, %
  \,$\mpP \equiv \mpBigPar{\roleP \in I}{\mpP[\roleP]}$, %
  \,and\, $\stEnv = \bigcup_{\roleP \in I}\stEnv[\roleP]$ %
  such that for each \,$\mpP[\roleP]$: %
  (1) $\stJudge{\mpEnvEmpty\!}{\stEnv[\roleP]}{\!\mpP[\roleP]}$,
  and
  (2) either $\mpP[\roleP] \equiv \mpNil$, %
  or $\mpP[\roleP]$ only plays $\roleP$ in $\mpS$, by $\stEnv[\roleP]$. %
  Then, %
    $\stEnvMoveMaybeCrashP[\mpS;\rolesS]{\stEnv}$ %
    \,implies\; %
    $\exists \stEnvi,\mpPi$ %
    such that\, %
    $\stEnv \!\stEnvMoveMaybeCrash[\mpS;\rolesS]\! \stEnvi$,\, %
    $\mpP \!\mpMoveMaybeCrashCheckedPlus\! \mpPi$ %
    \,and\, %
    $\stJudge{\mpEnvEmpty\!}{\!\stEnvi}{\mpPi}$, %
    with\; %
    $\stEnvSafeSessRolesSP{\mpS}{\rolesR}{\stEnvi}$, %
    \,$\mpPi \equiv \mpBigPar{\roleP \in I}{\mpPi[\roleP]}$, %
    \,and\, $\stEnvi = \bigcup_{\roleP \in I}\stEnvi[\roleP]$ %
    such that for each $\mpPi[\roleP]$:
    (1) $\stJudge{\mpEnvEmpty\!}{\stEnvi[\roleP]}{\!\mpPi[\roleP]}$,
    and
    (2) either $\mpPi[\roleP] \equiv \mpNil$,
    or $\mpPi[\roleP]$ only plays $\roleP$ in $\mpS$, by $\stEnvi[\roleP]$. %
\end{restatable}

%% file: 4_mpst_crash/6_properties.tex
\subsection{Statically Verifying Run-Time Properties of Processes with Crashes}%
\label{sec:dedlock-freedom}

We conclude this section by showing how to infer run-time process properties from typing contexts,
even in the presence of arbitrary process crashes.
The formulations are based on \cite[Def.\@~5.1 \& Fig.\@~5(1)]{POPL19LessIsMore},
but
\emph{(1)} we cater for optional assumptions on reliable roles;
\emph{(2)} a successfully-terminated process or typing context may include crashed session endpoints
    and failover types/processes (like DNS server $\roleR$ in \Cref{sec:overview}) that only run after detecting a crash; and
\emph{(3)} non-failover reliable roles %
    terminate by reaching $\mpNil$ (in processes) or $\stEnd$ (in types).

\Cref{def:proc-properties} formalises several desirable process properties,
using the assumption-abiding reduction $\mpMoveMaybeCrashChecked$ (\Cref{def:assumption-abiding-reduction})
to embed any assumptions on reliable roles used for typing. The properties are mostly self-explanatory:
\emph{deadlock-freedom} means that if a process cannot reduce,
then it only contains inactive or crashed sub-processes, or recovery processes attempting to detect others' crashes;
\emph{liveness} means that if a process is trying to perform an input or
output, then it eventually succeeds (unless it is only attempting to detect others' crashes).

\begin{definition}[Runtime Process Properties]
  \label{def:proc-properties}%
  \label{def:proc-deadlock-free}
  \label{def:proc-term}%
  \label{def:proc-nterm}%
  \label{def:proc-df}%
  \label{def:proc-liveness}%
  Assume $\stJudge{\emptyset}{\stEnv}{\mpP}$ where,
  $\forall \mpS \!\in\! \stEnv$, there is a set of roles $\rolesR[\mpS]$
  such that $\stEnvSafeSessRolesSP{\mpS}{\rolesR[\mpS]}{\stEnv}$.
  We say $\mpP$ is:
  \begin{enumerate}[(1)]
  \item\label{item:proc-properties:df}\textbf{deadlock-free} %
    iff %
    $\mpP \!\mpMoveMaybeCrashCheckedStar\! \mpNotMoveP{\mpPi}$
    implies
    \smallskip
    \begin{center}
    \(
    \mpPi \equiv \mpNil \mpPar \mpBigPar{i \in
      I}{\mpStop{\mpS[i]}{\roleP[i]}} \mpPar \mpBigPar{j \in
    J}{(\mpDefAbbrev{\mpDefD[j,1]}{\ldots\mpDefAbbrev{\mpDefD[j,n_j]}{\mpBranchSingle{\mpChanRole{\mpS[j]}{\roleP[j]}}{\roleQ[j]}{\mpLabCrash}{}{\mpQi[j]}}})};\)
    \end{center}
  \item\label{item:proc-properties:term}\textbf{terminating} %
    iff %
    it is deadlock-free, %
    and %
    $\exists j$ finite such that %
    $\forall n \!\ge\! j\!:\!$
    $\mpP \!=\! \mpP[0] \!\mpMoveMaybeCrashChecked\! \mpP[1] \!\mpMoveMaybeCrashChecked\!%
    \cdots \!\mpMoveMaybeCrashChecked\! \mpP[n]$ \,implies\,
    $\mpNotMoveP{\mpP[n]}$;
  \item\label{item:proc-properties:nterm}\textbf{never-terminating} %
    iff \,%
    $\mpP \!\mpMoveMaybeCrashCheckedStar\! \mpPi$ %
    \,implies\, %
    $\mpMoveP{\mpPi}$;%
  \item\label{item:proc-properties:live}\textbf{live} %
    iff \,%
    $\mpP \!\mpMoveMaybeCrashCheckedStar\! \mpPi \!\equiv\! \mpCtxApp{\mpCtx}{\mpQ}$ \,
    implies:
    \begin{enumerate}[(i)]
    \item\label{item:process-liveness:sel}%
      if %
      \,$\mpQ = \mpSel{\mpC}{\roleQ}{\mpLab}{\mpW}{\mpQi}$ %
      \,then %
      $\exists \mpCtxi:$ %
      $\mpPi \mpMoveStar \mpCtxApp{\mpCtxi}{\mpQi}$;\;\; %
    \item\label{item:process-liveness:branch}%
      if \,%
      $\mpQ = \mpBranch{\mpC}{\roleQ}{i \in I}{\mpLab[i]}{x_i}{\mpQi[i]}{}$ %
      \,%
      where $\setcomp{\mpLab[i]}{i \!\in\! I} \!\neq\! \setenum{\mpLabCrash}$%
      , %
      \,then\, %
      $\exists \mpCtxi, k \in I, \mpW:$ %
      $\mpPi \mpMoveStar \mpCtxApp{\mpCtxi}{\mpQi[k]\subst{x_k}{\mpW}}$.%
    \end{enumerate}
  \end{enumerate}
\end{definition}

In \Cref{def:typing-ctx-properties} we formalise the type-level properties corresponding to \Cref{def:proc-properties}.
Type-level liveness means that all pending internal/external choices are eventually fired (via a message transmission or crash detection)
-- assuming \emph{fairness} (\Cref{def:stenv-fairness},
based on \emph{strong fairness of components} \cite[Fact 2]{VanGlabbeekLICS2021}) so all enabled message transmissions are eventually performed.

\begin{definition}[{Non-crashing, Fair, Live Paths
  {(adapted from \cite[Def.\@ 4.4]{POPL21AsyncMPSTSubtyping})}}]%
  \label{def:stenv-fairness}
  A \textbf{non\nobreakdash-crashing path} %
  is %
  a
  possibly infinite sequence of typing contexts %
  $(\stEnv[n])_{n \in N}$, %
  where $N = \{0,1,2,\ldots\}$ is a set of consecutive natural numbers, %
  and, $\forall n \!\in\! N$, $\stEnv[n] \stEnvMove \stEnv[n+1]$. %

  We say that a non-crashing path $(\stEnv[n])_{n \in N}$ is \textbf{fair for session $\mpS$} iff, %
  $\forall n \!\in\! N$:
  \begin{enumerate}[(1)]
    \item
    $\stEnvMoveAnnotP{\stEnv[n]}{\ltsSendRecv{\mpS}{\roleP}{\roleQ}{\stLab}}$ %
    implies $\exists k, \stLabi$ %
    such that $N \ni k \ge n$, %
    and $\stEnv[k] \stEnvMoveCommAnnot{\mpS}{\roleP}{\roleQ}{\stLabi}
    \stEnv[k+1]$;
    \item
    $\stEnvMoveAnnotP{\stEnv[n]}{\ltsCrDe{\mpS}{\roleP}{\roleQ}}$ %
    implies $\exists k$ %
    such that $N \ni k \ge n$, %
    and $\stEnv[k] \stEnvMoveCrDeAnnot{\mpS}{\roleP}{\roleQ} \stEnv[k+1]$.
    \footnotemark
  \end{enumerate}
  \footnotetext{This condition is missing in the published version. We thank
  Ping Hou for pointing out this omission.}

  We say that a non-crashing path $(\stEnv[n])_{n \in N}$ is \textbf{live for session $\mpS$} iff, %
  $\forall n \in N$:
  \begin{enumerate}[(1)]
  \item\label{item:liveness:send}%
    $\stEnvMoveAnnotP{\stEnv[n]}{\stEnvOutAnnot{\roleP}{\roleQ}{\stChoice{\stLab}{\stS}}}$ %
    implies $\exists k, \stLabi$ %
    such that $N \ni k \ge n$ %
    and $\stEnv[k] \stEnvMoveCommAnnot{\mpS}{\roleP}{\roleQ}{\stLabi}
    \stEnv[k+1]$;
  \item\label{item:liveness:recv}%
    $\stEnvMoveAnnotP{\stEnv[n]}{\stEnvInAnnot{\roleQ}{\roleP}{\stChoice{\stLab}{\stS}}}$ %
    and $\stLab \!\neq\! \stCrashLab$ %
    implies $\exists k, \stLabi$ %
    such that $N \ni k \ge n$ %
    and $\stEnv[k] \!\stEnvMoveCommAnnot{\mpS}{\roleP}{\roleQ}{\stLabi}\! \stEnv[k+1]$
    or $\stEnv[k] \!\stEnvMoveAnnot{\ltsCrDe{\mpS}{\roleQ}{\roleP}}\!
    \stEnv[k+1]$.
  \end{enumerate}
\end{definition}

\begin{definition}[Typing Context Properties]
  \label{def:typing-ctx-properties}
  Given a session $\mpS$ and a set of reliable roles $\rolesR$,
  we say $\stEnv$ is:
  \begin{enumerate}[(1)]
  \item\label{item:typing-ctx-properties:df}%
    \textbf{$(\mpS;\rolesR)$-deadlock-free} iff\,
    $\stEnv \!\stEnvMoveMaybeCrashStar[\mpS;\rolesR]\! \stEnvNotMoveP{\stEnvi}$
    \,implies\,
    $\forall \mpChanRole{\mpS}{\roleP} \!\in\! \stEnv:
    \stEnvApp{\stEnv}{\mpChanRole{\mpS}{\roleP}} \!\stSub\! \stEnd$ \,or\,
    $\stEnvApp{\stEnv}{\mpChanRole{\mpS}{\roleP}} \!=\! \stStop$ \,or\,
    $\exists \roleQ$: %
    $\stEnvApp{\stEnv}{\mpChanRole{\mpS}{\roleP}} \stSub \stInNB{\roleQ}{\stCrashLab}{}{} \stSeq \stTi$;
  \item\label{item:typing-ctx-properties:term}%
    \textbf{$(\mpS;\rolesR)$-terminating} iff it is deadlock-free,
    and $\exists j$ finite such that $\forall n \!\ge\! j$: $\stEnv \!=\! \stEnv[0]
    \!\stEnvMoveMaybeCrash[\mpS;\rolesR] \stEnv[1]
    \!\stEnvMoveMaybeCrash[\mpS;\rolesR] \cdots
    \!\stEnvMoveMaybeCrash[\mpS;\rolesR] \stEnv[n]$ \,implies\, $\stEnvNotMoveP{\stEnv[n]}$;
  \item\label{item:typing-ctx-properties:nterm}%
    \textbf{$(\mpS;\rolesR)$-never-terminating} iff\,
    $\stEnv \stEnvMoveMaybeCrashStar[\mpS;\rolesR] \stEnvi$ \,implies\, $\stEnvMoveP{\stEnvi}$;
  \item\label{item:typing-ctx-properties:live}%
    \textbf{$(\mpS;\rolesR)$-live} iff\,
    $\stEnv \stEnvMoveMaybeCrashStar[\mpS;\rolesR] \stEnvi$ \,implies
    all non-crashing paths starting with $\stEnvi$ which are fair for session $\mpS$ are also live for $\mpS$.
  \end{enumerate}
\end{definition}

\begin{example}
  Reliability assumptions $\rolesR$ can affect typing context properties, \eg consider:

\smallskip
\centerline{\small\(
\begin{array}{rcl}
\stEnv \;=\;
\stEnvMap{\mpChanRole{\mpS}{\roleP}}{
  \stRec{\stRecVar[\roleP]}{
    \stOut{\roleQ}{\stLabOK}{} \stSeq
    \stRecVar[\roleP]
  }
}
\stEnvComp\;
\stEnvMap{\mpChanRole{\mpS}{\roleQ}}{
  \stRec{\stRecVar[\roleQ]}{
    \stIn{\roleP}{\stLabOK}{}{
      \stRecVar[\roleQ],\,
      \stChoice{\stCrashLab}{} \stSeq
      \stRec{\stRecVari[\roleQ]}{
        \stIn{\roleR}{\stLabOK}{}{
          \stRecVari[\roleQ],\,
          \stChoice{\stCrashLab}{} \stSeq \stEnd
        }
      }
    }
 }
}
\stEnvComp\;
\stEnvMap{\mpChanRole{\mpS}{\roleR}}{
  \stRec{\stRecVar[\roleR]}{
    \stOut{\roleQ}{\stLabOK}{} \stSeq
    \stRecVar[\roleR]
  }
}
\end{array}
\)}
\smallskip
\noindent
If $\rolesS \!=\! \rolesSEmpty$,
$\stEnv$ is safe and deadlock-free but \emph{not} live: if $\roleP$ does \emph{not} crash, $\roleR$'s $\stLabFmt{ok}$ message is never received by $\roleQ$.
If we have $\rolesS \!=\! \gtFmt{\{\roleR\}}$, $\stEnv$ satisfies never-termination. Here, neither liveness nor termination can be satisfied by adding reliability assumptions. More examples in \Cref{sec:appendix:examples}, \cref{eg:allegs:props}.
\end{example}

We conclude by showing how the type-level properties in \Cref{def:typing-ctx-properties}
allow us to infer the corresponding process properties in \Cref{def:proc-properties}.
\iftoggle{techreport}{The proof is available in \cref{sec:proofs:proc-properties}.}{The proof is available in~\cite{techreport}.}

\begin{restatable}[Verification of Process Properties]{theorem}{lemProcessPropertiesVerif}%
  \label{lem:stenv-proc-properties}
  \label{lem:deadlock-freedom}%
  Assume\, $\stJudge{\mpEnvEmpty\!}{\!\stEnv}{\!\mpP}$, %
  where %
  $\stEnv$ is ($\mpS;\rolesR$)-safe, %
  \,$\mpP \equiv \mpBigPar{\roleP \in I}{\mpP[\roleP]}$, %
  \,and\, $\stEnv = \bigcup_{\roleP \in I}\stEnv[\roleP]$ %
  such that for each $\mpP[\roleP]$, %
  we have\, $\stJudge{\mpEnvEmpty\!}{\stEnv[\roleP]}{\!\mpP[\roleP]}$.
  \,Further, assume that each $\mpP[\roleP]$
  is either\, $\mpNil$ (up to $\equiv$), %
  or only plays $\roleP$ in $\mpS$, by $\stEnv[\roleP]$. %
  Then,
  for all $\predP \in \setenum{\text{deadlock-free}, \text{terminating}, \text{never-terminating}, \text{live}}$,
  if $\stEnv$ is $(\mpS;\rolesR)$-$\predP$,
  then $\mpP$ is $\predP$.
\end{restatable}

%% file: 5_modelchecking/modelchecking.tex
\section{Verifying Type-Level Properties via Model Checking}
\label{sec:model-checking}

In our generalised typing system, we prove subject reduction
when a typing context satisfies a safety property (\Cref{def:mpst-env-safe});
we then give examples of more refined typing context properties (\Cref{def:typing-ctx-properties})
and show how they are inherited by typed processes (\Cref{lem:stenv-proc-properties}).
In this section, we highlight a major benefit of our theory:
we show how such typing context behavioural properties
can be verified using model checkers. %
We use our typing contexts and their semantics
(including crashes and crash handling) as models,
and we express our behavioural properties as modal $\mu$-calculus \formulae;
we then use a model checker (mCRL2~\cite{TACAS19mCRL2})
to verify whether a typing context enjoys a desired property.
\vspace{-2ex}
\subparagraph{Contexts as Models.}
We encode our typing contexts as mCRL2 processes, with LTS semantics
that match \Cref{def:mpst-env-reduction}.
To embed our optional reliability assumptions,
the context encoding reflects the transition relation
$\stEnvMoveMaybeCrash[\mpS; \rolesS]$,
so it never crashes any reliable role in $\rolesS$.

\vspace{-2ex}
\subparagraph{Properties as \Formulae.}
A modal $\mu$-calculus formula $\muPred$ accepts or rejects a typing context $\stEnv$
depending on the transition labels $\stEnv$ can fire while reducing.
We write $\muJudge{\stEnv}{\muPred}$ when a typing context $\stEnv$ satisfies $\muPred$.
Actions $\muAct$ range over transition labels in \cref{def:mpst-env-reduction-label};
$\muData$ (for $\muData$ata) ranges over
sessions, roles, message labels, and types.
Our \formulae $\muPred$ follow a standard syntax:

\smallskip
\centerline{\(
\begin{array}{l}
    \muPred \;\bnfdef\;
    \muTrue
    \hspace{-0.3mm}\bnfsep\hspace{-0.3mm} \muFalse
    \hspace{-0.3mm}\bnfsep\hspace{-0.3mm} \muBox{\muAct}{\muPred}
    \hspace{-0.3mm}\bnfsep\hspace{-0.3mm} \muDiamond{\muAct}{\muPred}
    \hspace{-0.3mm}\bnfsep\hspace{-0.3mm} \muPred[1] \muAnd \muPred[2]
    \hspace{-0.3mm}\bnfsep\hspace{-0.3mm} \muPred[1] \muOr \muPred[2]
    \hspace{-0.3mm}\bnfsep\hspace{-0.3mm} \muPred[1] \muImplies \muPred[2]
    \hspace{-0.3mm}\bnfsep\hspace{-0.3mm} \muLFP{\muVar}{\muPred}
    \hspace{-0.3mm}\bnfsep\hspace{-0.3mm} \muGFP{\muVar}{\muPred}
    \hspace{-0.3mm}\bnfsep\hspace{-0.3mm} \muVar
    \hspace{-0.3mm}\bnfsep\hspace{-0.3mm} \muForall{\muData}{\muPred}
    \hspace{-0.3mm}\bnfsep\hspace{-0.3mm} \muExists{\muData}{\muPred}
\end{array}
\)} %
\smallskip

\noindent
Truth ($\muTrue$) accepts any $\stEnv$; falsity ($\muFalse$) accepts no
$\stEnv$.
The box (resp.\@ diamond) modality, $\muBox{\muAct}{\muPred}$ (resp.\@
$\muDiamond{\muAct}{\muPred}$),
requires that $\muPred$ is satisfied in all cases (resp.\@ some cases)
after action $\muAct$ is fired.
The least (resp.\ greatest) fixed point $\muLFP{\muVar}{\muPred}$ (resp.\
$\muGFP{\muVar}{\muPred}$) allows one to iterate $\muPred$ for a finite (resp.\
infinite) number of times, where $\muVar$ denotes a variable for iteration.
Lastly, the forms $\muPred[1] \muImplies \muPred[2]$,
$\muForall{\muData}{\muPred}$, and $\muExists{\muData}{\muPred}$ denote
implication, and universal and existential quantification.

In \cref{fig:mc-formulae} we show the $\mu$-calculus \formulae corresponding to
our properties in \Cref{def:mpst-env-safe,def:typing-ctx-properties}.
Compared to~\cite{POPL19LessIsMore}, such properties are more complex,
since they cater for crashes and crash handling transitions.
Recall \cref{def:mpst-env-safe}, and take a safety property $\predP$:
for $\predPApp{\stEnv}$ to hold,
clause $\inferrule{\iruleSafeMove}$ requires that whenever
$\stEnv$ can transition to some $\stEnvi$ (via $\stEnvMoveMaybeCrash[\mpS; \rolesS]$),
then $\predPApp{\stEnvi}$ also holds.
To represent this clause in modal $\mu$-calculus,
we use fixed points for possibly infinite paths;
in \cref{fig:mc-formulae} we write $\muPred[\to](\muVar)$ for following a fixed point
$\muVar$ via any transmission, crash,\footnotemark~or crash handling actions,
and we define it as
\(
\muPred[\to](\muVar) =
\muForall{\mpS,\roleP,\roleQ,\stLab}{}%
\muBox{\stEnvCommAnnotSmall{\roleP}{\roleQ}{\stLab}}{\muVar}%
      \muAnd
      \muBox{\ltsCrashSmall{\mpS}{\roleP}}{\muVar}
      \muAnd
      \muBox{\ltsCrDe{\mpS}{\roleP}{\roleQ}}{\muVar}
\).
\footnotetext{%
  Since our typing contexts encoded in mCRL2 produce
  $\stEnvMoveMaybeCrash[\mpS; \rolesS]$-transitions
  that never crash reliable roles in $\rolesS$,
  our $\mu$-calculus \formulae can follow \emph{all} crash transitions;
  hence, the \formulae do not depend on $\rolesS$.
}%

\input{5_modelchecking/formulae.tex}

\emph{\textbf{Safety}} (\inferrule{\iruleFmlaSafe})
requires
(in its second implication) that whenever $\stEnv$ can fire an input action, and either an output or $\actCrashed{\mpS}{\roleQ}$ action, then $\stEnv$ can also
fire a message transmission, $\stEnvCommAnnotSmall{\roleP}{\roleQ}{\stLab}$.
The first implication requires that, if $\stEnv$ can fire a
$\actCrashed{\mpS}{\roleP}$ action and an input action
$\stEnvInAnnotSmall{\roleQ}{\roleP}{\stLabPi}$,
then $\stEnv$
must be capable of firing a crash handling action,
$\ltsCrDe{\mpS}{\roleQ}{\roleP}$.

\emph{\textbf{Deadlock-Freedom}} (\inferrule{\iruleFmlaDF})
requires that, if $\stEnv$ is unable to reduce further
without crashing (via $\stEnvMove$),
then $\stEnv$ can hold only $\stEnd$ed or $\stStop$ped endpoints.
The antecedent of $\muImplies$ characterises a context that is unable to reduce
(since $\stEnvMove$ only allows for transmissions
$\stEnvCommAnnotSmall{\roleP}{\roleQ}{\stLab}$
and crash detection
$\ltsCrDe{\mpS}{\roleP}{\roleQ}$);
the consequent forbids the presence of any input
$\stEnvInAnnotSmall{\roleP}{\roleQ}{\stChoice{\stLab}{\stS}}$
or output
$\stEnvOutAnnotSmall{\roleP}{\roleQ}{\stChoice{\stLab}{\stS}}$
transitions. By \Cref{def:mpst-env-reduction},
this means all session endpoints in $\stEnv$ are $\stEnd$ed or
$\stStop$ped.

\emph{\textbf{Terminating}} (\inferrule{\iruleFmlaTerm})
holds when $\stEnv$ can reach a terminal configuration (\ie cannot further reduce
via $\stEnvMove$) within a \emph{finite} number of steps.
Hence, the formula
is similar to deadlock-freedom, except that it uses the \emph{least} fixed
point ($\muLFP{\muVar}{\ldots}$)
to ensure %
finiteness.

\emph{\textbf{Never-Terminating}} (\inferrule{\iruleFmlaNTerm})
requires that
$\stEnv$ can always keep reducing via $\stEnvMove$ transitions.
Therefore, we require some transmission
$\stEnvCommAnnotSmall{\roleP}{\roleQ}{\stLab}$
or crash detection action
$\ltsCrDe{\mpS}{\roleP}{\roleQ}$
to be always fireable, even after some of the non-reliable roles crash.

\emph{\textbf{Liveness}} (\inferrule{\iruleFmlaLive})
requires that any enabled input/output action is triggered
by a corresponding message transmission or crash detection,
within a finite number of steps.
For input actions (sub-formula $\muPred[in]$): if an input
$\stEnvInAnnotSmall{\roleQ}{\roleP}{\stChoice{\stLab}{\stS}}$ is enabled
(left of $\muImplies$),
then, in a finite number of steps ($\muLFP{\muVari}{\ldots}$)
involving \emph{other} roles $\rolePi, \roleQi$,
a transmission
$\stEnvCommAnnotSmall{\roleP}{\roleQ}{\stLabi}$
or a crash detection
$\ltsCrDe{\mpS}{\roleQ}{\roleP}$
can be fired.
For output actions, the sub-formula $\muPred[out]$ is similar.
The $\mu$-calculus formula embeds fairness (\Cref{def:stenv-fairness})
by finding \emph{some} roles $\rolePi, \roleQi$ that,
no matter how they interact (sub-formula $\muPredi[\to]$),
lead to the desired transmission or crash detection.

\vspace{-2ex}
\subparagraph{Tool Implementation and Example.}
To verify the properties in~\cref{fig:mc-formulae}, we implement a prototype tool
that extends \theTool~\cite{SY19Artifact} (based on the mCRL2 model checker~\cite{TACAS19mCRL2}) with support for our crash-stop semantics. The updated tool is available at:

\smallskip\centerline{
  \url{https://github.com/alcestes/mpstk-crash-stop}
}\smallskip%

\noindent%
We now illustrate how this new tool helps in writing correct session protocols
with crash handling, and briefly discuss its performance.

\begin{figure}
  \input{5_modelchecking/twobuyer.tex}
  \caption{Two-Buyers protocol extended with crash-handling.}
  \label{fig:mc:twobuyers}
\end{figure}

In the \emph{two-buyers protocol} from MPST literature \cite{HYC08},
buyers $\roleFmt{b1}$ and $\roleFmt{b2}$ agree on splitting the cost
of buying a book from seller $\roleFmt{s}$. We tackle this protocol with crashes and \emph{no reliability assumptions}:
all roles may crash, and survivors must end the session correctly.
The resulting \emph{crash-tolerant two-buyers protocol}
(\cref{fig:mc:twobuyers}) is much more complex than the one in the literature.
In fact, the possibility of crashes introduces a variety of scenarios
where different roles may be crashed (or not),
hence the protocol needs many $\stCrashLab$ branches.
The protocol exhibits two crash-handling patterns: \emph{i)} exiting gracefully,
and \emph{ii)} recovery behaviour. The former occurs either when $\roleFmt{s}$
crashes or when $\roleFmt{b1}$ crashes prior to the agreed split. The latter
occurs should $\roleFmt{b2}$ crash after the agreed split, whereupon
$\roleFmt{b2}$ concludes the transaction if both $\roleFmt{b2}$ and
$\roleFmt{s}$ do not crash. This behaviour is activated via a recovery type in
$\mpChanRole{\mpS}{\roleFmt{b1}}$, where the labels $\stLabFmt{rp{\text{$n$}}}$
represent the point at which $\roleFmt{b2}$ crashed: $\stLabFmt{rp1}$ represents
$\roleFmt{b2}$ failing prior to confirmation with $\roleFmt{s}$;
$\stLabFmt{rp2}$ corresponds to before the sending of $\stLabFmt{addr}$; and
$\stLabFmt{rp3}$ prior to receiving the $\stLabFmt{date}$.
Overlooking or mishandling some cases is easy; our tool spots such errors,
so the protocol can be tweaked until all desired properties hold.
We used our tool to verify the protocol: it has $1409$ states and
$10248$ transitions; it is safe, %
deadlock-free,
live,
and it is terminating;
it is %
\emph{not} never-terminating. %
All properties verify within $100$ms on a 4.20~GHz
Intel Core i7-7700K CPU with 16 GB RAM. %
More experimental results can be found in \cref{sec:appdx:examples}.

%% file: 5_modelchecking/formulae.tex
\begin{figure}
\begin{center}
\footnotesize
\begin{tabular}{@{}r@{\hspace{3mm}}>{$}l<{$}}
\inferrule{\iruleFmlaSafe} &
\muJudge{\stEnv}{%
  \muGFP{\muVar}{%
    \left(%
    \begin{array}{@{}l@{}}
      \muForall{\mpS,\roleP,\roleQ,\stLab,\stLabi,\stS,\stSi}{}%
      \big(
      \muDiamond{\actCrashed{\mpS}{\roleP}}{\muTrue}%
      \muAnd%
      \muDiamond{\stEnvInAnnotSmall{\roleQ}{\roleP}{\stLabPi}}{\muTrue}%
      \muImplies%
      \muDiamond{\ltsCrDe{\mpS}{\roleQ}{\roleP}}{\muTrue}%
      \big)
      \\%
      \;\muAnd\;
      \big(%
      \muDiamond{\stEnvOutAnnotSmall{\roleP}{\roleQ}{\stLabP}}{\muTrue}%
      \muAnd%
      (
      \muDiamond{\stEnvInAnnotSmall{\roleQ}{\roleP}{\stLabPi}}{\muTrue}%
      \muOr
      \muDiamond{\actCrashed{\mpS}{\roleQ}}{\muTrue}
      )
      \muImplies%
      \muDiamond{\stEnvCommAnnotSmall{\roleP}{\roleQ}{\stLab}}{\muTrue}%
      \big)%
    \muAnd%
    \muPred[\to](\muVar)
    \end{array}
    \right)
  }
}
\\
\midrule
\inferrule{\iruleFmlaDF} &
\muJudge{\stEnv}{%
  \muGFP{\muVar}{%
    \left(%
    \begin{array}{@{}l@{}}
      \left(%
      \begin{array}{@{}l@{}}
        (\muForall{\mpS, \roleP,\roleQ,\stLab}{%
          \muBox{\stEnvCommAnnotSmall{\roleP}{\roleQ}{\stLab}}{\muFalse}%
          \muAnd
          \muBox{\ltsCrDe{\mpS}{\roleP}{\roleQ}}{\muFalse}
        })%
        \;\muImplies%
        \\%
        \quad%
        \muForall{\mpS,\roleP,\roleQ,\stLab,\stS}{%
          \muBox{\stEnvInAnnotSmall{\roleP}{\roleQ}{\stLabP}}{\muFalse}%
          \muAnd%
          \muBox{\stEnvOutAnnotSmall{\roleP}{\roleQ}{\stLabP}}{\muFalse}%
        }%
      \end{array}
      \right)%
    \end{array}
    \;\muAnd\;%
    \muPred[\to](\muVar)
    \right)%
  }%
}%
\\
\midrule
\inferrule{\iruleFmlaTerm} &
\muJudge{\stEnv}{%
  \muLFP{\muVar}{%
    \left(%
    \begin{array}{@{}l@{}}
      \left({%
        \begin{array}{@{}l@{}}
          \left({
            \muForall{\mpS, \roleP, \roleQ, \stLab}{%
              \muBox{\stEnvCommAnnotSmall{\roleP}{\roleQ}{\stLab}}{%
                \muFalse%
              }%
              \muAnd
              \muBox{\ltsCrDe{\mpS}{\roleP}{\roleQ}}{\muFalse}
            }%
          }\right)%
          \;\muImplies%
          \\%
          \quad%
          \muForall{\mpS, \roleP, \roleQ, \stLab, \stS}{%
            \muBox{\stEnvInAnnotSmall{\roleP}{\roleQ}{\stLabP}}{\muFalse}%
            \muAnd%
            \muBox{\stEnvOutAnnotSmall{\roleP}{\roleQ}{\stLabP}}{\muFalse}%
          }%
        \end{array}
      }\right)%
    \end{array}
    \;\muAnd\;%
    \muPred[\to](\muVar)
    \right)%
  }%
}%
\\
\midrule
\inferrule{\iruleFmlaNTerm} &
\muJudge{\stEnv}{%
  \muGFP{\muVar}{%
    \left(%
      \begin{array}{@{}l@{}}
      \left(%
      \muExists{\mpS,\roleP,\roleQ,\stLab}{%
        \muDiamond{\stEnvCommAnnotSmall{\roleP}{\roleQ}{\stLab}}{\muTrue}
        \muOr
        \muDiamond{\ltsCrDe{\mpS}{\roleP}{\roleQ}}{\muTrue}
      }
      \right)%
      \end{array}
      \;\muAnd\;%
      \muPred[\to](\muVar)
    \right)%
  }%
}%
\\
\midrule
\inferrule{\iruleFmlaLive} &
\muJudge{\stEnv}{%
  \muGFP{\muVar}{%
    \left(%
    \begin{array}{@{}l@{}}
      \muForall{\mpS,\roleP,\roleQ}{}%
      \\%
      \quad%
      \muPred[in] =
      \left({%
      \begin{array}{@{}l@{}}
        (\muExists{\stLab,\stS}{%
          \muDiamond{\stEnvInAnnotSmall{\roleQ}{\roleP}{\stLabP}}{\muTrue}%
        })%
        \;\muImplies %
      \\\quad
      \muLFP{\muVari}{%
        \left(
        \begin{array}{@{}l@{}}
        \muExists{\stLab}{%
          \muDiamond{%
            \stEnvCommAnnotSmall{\roleP}{\roleQ}{\stLab}%
          }{\muTrue}%
        }%
        \,\muOr\,%
        \muDiamond{\ltsCrDe{\mpS}{\roleQ}{\roleP}}{\muTrue}
        \\
        \muOr
        \muExists{\rolePi, \roleQi}{%
          \left({%
            \begin{array}{@{}l@{}}%
              \muExists{\stLabi}{%
                \muDiamond{%
                  \stEnvCommAnnotSmall{\rolePi}{\roleQi}{\stLabi}%
                }{\muTrue}%
                \muOr
                \muDiamond{\ltsCrDe{\mpS}{\roleQi}{\rolePi}}{\muTrue}
              }%
              \\%
              \muAnd\; %
              \muPredi[\to](\mpS, \rolePi, \roleQi, \muVari)
            \end{array}
          }\right)%
        }%
        \end{array}
        \right)
      }%
      \end{array}
      }\right)%
      \\%
      \quad\muAnd\;%
      \muPred[out] =
      \muForall{\stLab}{%
        \left(%
        \begin{array}{@{}l@{}}
          (\muExists{\stS}{%
            \muDiamond{\stEnvOutAnnotSmall{\roleP}{\roleQ}{\stLabP}}{\muTrue}%
          })%
          \;\muImplies %
          \\\quad
          \muLFP{\muVari}{%
            \left(
            \begin{array}{@{}l@{}}
            \muDiamond{%
              \stEnvCommAnnotSmall{\roleP}{\roleQ}{\stLab}%
            }{\muTrue}%
            \\
            \muOr
            \muExists{\rolePi, \roleQi}{%
              \left({%
                \begin{array}{@{}l@{}}%
                  \muExists{\stLabi}{%
                    \muDiamond{%
                      \stEnvCommAnnotSmall{\rolePi}{\roleQi}{\stLabi}%
                    }{\muTrue}%
                    \muOr
                    \muDiamond{\ltsCrDe{\mpS}{\roleQi}{\rolePi}}{\muTrue}
                  }%
                  \\%
                  \muAnd\; %
                  \muPredi[\to](\mpS, \rolePi, \roleQi, \muVari)
                \end{array}
              }\right)%
            }%
            \end{array}
            \right)
          }%
        \end{array}
        \right)%
      }%
    \\
    \;\muAnd\;%
    \muPred[\to](\muVar)
    \end{array}
    \right)%
  }%
}%
\end{tabular}
\end{center}
\vspace{-2ex}
\caption{Modal $\mu$-Calculus \Formulae corresponding to Properties in
  \cref{def:mpst-env-safe,def:typing-ctx-properties}, where
\(
\muPred[\to](\muVar) =
\muForall{\mpS,\roleP,\roleQ}{}%
\muPredi[\to](\mpS, \roleP, \roleQ, \muVar)
\), and
\(
\muPredi[\to](\mpS, \roleP, \roleQ, \muVar) =
\muForall{\stLab}{}%
\muBox{\stEnvCommAnnotSmall{\roleP}{\roleQ}{\stLab}}{\muVar}%
      \muAnd
      \muBox{\ltsCrashSmall{\mpS}{\roleP}}{\muVar}
      \muAnd
      \muBox{\ltsCrDe{\mpS}{\roleP}{\roleQ}}{\muVar}
\).
}
\label{fig:mc-formulae}
\end{figure}

%% file: 5_modelchecking/twobuyer.tex
\noindent%
\fbox{%
\begin{minipage}{0.98\textwidth}
\footnotesize
\(
\begin{array}{l}
\stEnvMap{\mpChanRole{\mpS}{\roleFmt{b1}}}{
  \stOut{\roleS}{\stLabFmt{req}}{\stTypeString} \stSeq
  \stIn{\roleS}{\stLabFmt{quote}}{\stTypeInt}{
    \stOut{\roleFmt{b2}}{\stLabFmt{split}}{\stTypeInt} \stSeq
      \stIn{\roleFmt{b2}}{\stCrashLab}{}{ \stT[1] }
    \stEnvComp\;
    \stCrashLab \stSeq
    \stOut{\roleFmt{b2}}{\stT[\stLabKO]}{} %
  }
}
\\
\hspace{10mm}\stT[1] = \stIn{\roleS}{\stLabFmt{rp1}}{}{
  \stIntSum{\roleS}{}{
    \stLabOK \stSeq \stT[2] %
    \stEnvComp\;
    \stT[\stLabKO]
  }
  \stEnvComp\;
  \stLabFmt{rp2} \stSeq \stT[2] %
  \stEnvComp\;
  \stLabFmt{rp3} \stSeq \stIn{\roleS}{\stLabFmt{date}}{\stTypeString}{\stEnd \stEnvComp\; \stT[\lightning]}
  \stEnvComp\;
  \stT[\lightning]
}
\,
\\
\hspace{10mm}\stT[2] = \stOut{\roleS}{\stLabFmt{addr}}{\stTypeString} \stSeq \stIn{\roleS}{\stLabFmt{date}}{\stTypeString}{\stEnd \stEnvComp\; \stT[\lightning]}
\\
\stEnvMap{\mpChanRole{\mpS}{\roleFmt{b2}}}{
  \stIn{\roleS}{\stLabFmt{quote}}{\stTypeInt}{
      \stT[1]
    \stEnvComp\;
    \stT[\stLabKO]
    \stEnvComp\;
    \stCrashLab \stSeq \stT[1]
  }
}
\\
\hspace{10mm}\stT[1] = \stIn{\roleFmt{b1}}{\stLabFmt{split}}{\stTypeInt}{
  \stIntSum{\roleS}{}{
    \stLabOK \stSeq
    \stOut{\roleS}{\stLabFmt{addr}}{\stTypeString} \stSeq
    \stIn{\roleS}{\stLabFmt{date}}{\stTypeString}{
      \stEnd
      \stEnvComp\;
      \stT[\lightning]
    }
    \stEnvComp\;
    \stT[\stLabKO]
  }
  \stEnvComp\;
  \stT[\stLabKO]
  \stEnvComp\;
  \stCrashLab \stSeq
  \stOut{\roleS}{\stT[\stLabKO]}{} %
}

\\
\stEnvMap{\mpChanRole{\mpS}{\roleS}}{
  \stIn{\roleFmt{b1}}{\stLabFmt{req}}{\stTypeString}{
    \stOut{\roleFmt{b1}}{\stLabFmt{quote}}{\stTypeInt} \stSeq
    \stOut{\roleFmt{b2}}{\stLabFmt{quote}}{\stTypeInt} \stSeq
    \stIn{\roleFmt{b2}}{\stLabOK}{}{
      \stT[1]
      \stEnvComp\;
      \stT[\stLabKO]
      \stEnvComp\;
      \stCrashLab \stSeq
      \stOut{\roleFmt{b1}}{\stLabFmt{rp1}}{} \stSeq
      \stT[2]
    }
    \stEnvComp\;
    \stCrashLab \stSeq
    \stOut{\roleFmt{b2}}{\stT[\stLabKO]}{} %
  }
}
\\
\hspace{10mm}\stT[1] = \stIn{\roleFmt{b2}}{\stLabFmt{addr}}{\stTypeString}{
  \stOut{\roleFmt{b2}}{\stLabFmt{date}}{\stTypeString} \stSeq
  \stIn{\roleFmt{b2}}{\stCrashLab}{}{
    \stOut{\roleFmt{b1}}{\stLabFmt{rp3}}{}
    \stSeq
    \stT[4]
  }
  \stEnvComp\;
  \stCrashLab \stSeq
  \stOut{\roleFmt{b1}}{\stLabFmt{rp2}}{}
  \stSeq
  \stT[3]
}
\\
\hspace{10mm}\stT[2] = \stIn{\roleFmt{b1}}{\stLabOK}{}{
  \stT[3]
  \stEnvComp\;
  \stT[\stLabKO]
  \stEnvComp\;
  \stT[\lightning]
}
\;\;
\hspace{5mm}\stT[3] = \stIn{\roleFmt{b1}}{\stLabFmt{addr}}{\stTypeString}{
  \stT[4]
  \stEnvComp\;
  \stT[\lightning]
}
\;\;
\hspace{5mm}\stT[4] = \stOut{\roleFmt{b1}}{\stLabFmt{date}}{\stTypeString} \stSeq \stEnd
\\
\text{where:}\;\;
  \stT[\lightning] = \stCrashLab \stSeq \stEnd
  \qquad%
  \stT[\stLabOK] = \stLabOK \stSeq \stEnd
  \qquad%
  \stT[\stLabKO] = \stLabKO \stSeq \stEnd
\end{array}
\)
\end{minipage}
}%

%% file: 6_related.tex
\section{Related Work, Conclusions, and Future Work}\label{sec:related}

\subparagraph{Previous Work on Failure Handling in Session Types}
\!\!can be generally classified under two main approaches:
\emph{affine} and \emph{coordinator model}.
The former adapts %
session types
to allow session endpoints to cease prematurely (\eg by throwing an exception);
the latter assumes reliable process coordination to handle failures.

\emph{Affine failure handling}
is first proposed %
in~\cite{LMCS18Affine}
for a $\pi$-calculus with binary sessions (\ie
two roles), and
\cite{DBLP:journals/pacmpl/FowlerLMD19} presents a concurrent
$\lambda$-calculus %
with
binary sessions %
and exception handling;
exceptions are also found in
\cite{DBLP:journals/mscs/CapecchiGY16,
DBLP:conf/concur/CarboneHY08}.
These works model failures at the application level, %
via throw/catch constructs.
Our key innovations are:
\emph{(1) }we model arbitrary failures (\eg hardware failures);
\emph{(2)} we specify what to do when a failure is detected \emph{at the type level};
\emph{(3)} we support multiparty sessions; and %
\emph{(4)} we seamlessly support
    handling the crash of a role while handling another role's crash,
    whereas the \emph{do-catch} constructs %
    cannot be nested.

\emph{Coordinator model approaches} include %
\cite{DBLP:conf/forte/AdameitPN17}, which extends MPST
with \emph{optional blocks} %
where \emph{default values} are used when communications fail;
and \cite{DBLP:conf/forte/ChenVBZE16}, which
uses synchronisation points %
to detect and handle failures.
Both need processes to coordinate to handle failures.
\cite{ESOP18CrashHandling} extends
MPST
 with a \emph{try-handle} construct: a reliable coordinator %
detects and broadcasts failures,
and the remaining processes %
proceed with failure handling.
Unlike these works, we do \emph{not} assume reliable processes,
failure broadcasts, or coordination/synchronisation points.

Other papers address failures with different approaches.
The recent work \cite{DBLP:conf/forte/PetersNW22} annotates
global and local types to specify which interactions may fail, and how
(process crash, message loss). Their failure model is different from ours;
and unlike us, they handle failures by continuing the
protocol via \emph{default branches and values}.
Instead, our types include $\stCrashLab$ branches defining recovery behaviours
that are only executed upon crash detection; further, by nesting such
$\stCrashLab$ branches, we can specify different behaviours depending on
which roles have crashed.
\cite{NY2017}
uses an MPST specification to %
build a dependency graph among running processes, %
supervise them, and restart them %
in case of failure.
\cite{OOPSLA21FaultTolerantMPST} utilises
MPST to specify fault-tolerant, event-driven distributed systems,
where processes are monitored and restarted if they fail;
unlike our work, they require certain reliable
roles,
but their model tolerates false crash suspicions.
More on the theory side,
\cite{DBLP:conf/esop/CairesP17} presents
a Curry-Howard interpretation of a language with binary session types
and internal non-determinism, which is used to model failures
(that are propagated to all relevant sessions, similarly to \cite{LMCS18Affine,DBLP:journals/pacmpl/FowlerLMD19}).
Process calculi with \emph{localities} have been proposed to model
distributed systems with failures
\cite{ICALP97Locality,COORDINATION97Locality,CastellaniLocalities};
unlike our work, they do not have a typing system to verify failure handling.

\subparagraph{Generalised Multiparty Session Type Systems}
\!\!(introduced in~\cite{POPL19LessIsMore})
depart from ``classic'' MPST~\cite{HYC16}
by not requiring top-down syntactic notions of protocol correctness
(global types, projection, \etc);
rather, they %
check behavioural predicates (safety, liveness, \etc) %
against (local) session types.
\cite{ECOOP21MPSTActor} adopts the approach
to model actor systems with explicit connections in their types~\cite{HY2017}.
By adopting this general framework,
we support protocols %
not representable as global types in classic MPST (\eg DNS in \Cref{sec:overview},
two-buyers in \Cref{sec:model-checking}, and all examples in
\Cref{sec:appdx:examples}, %
excepting \exampleName{Adder}).

\subparagraph{Model Checking Behavioural Types.}
\cite{POPL02ModelChecking} develops a behavioural type
system for the $\pi$-calculus,
and check LTL formul\ae\ against such types.
In~\cite{KobayashiS10Hybrid}, the type system combines typing and local
analyses, with liveness properties verified via model checking.
A similar approach is introduced in~\cite{POPL19LessIsMore} for MPST.
Regarding applications,
\cite{ICSE18Go, ECOOP20Go}
verify behavioural types extracted from Go source code;
and in~\cite{PLDI19Effpi}, the \texttt{Effpi} Scala library assigns
behavioural types to communicating programs.
These works use a model checker to validate \eg liveness through
type-level properties, but do not support crashes or crash
handling.

%% file: 7_conclusion.tex
\subparagraph{Conclusions and Future Work.}%
\label{sec:conclusion}

We presented a multiparty session typing system for
verifying processes with crash-stop failures.
We model crashes and crash handling in a session $\pi$-calculus and
its typing contexts, and prove type safety, protocol conformance, deadlock
freedom and liveness.
Our system is generalised in two ways:
\emph{(1)}
    it supports \emph{optional} reliability assumptions, ranging from fully
    reliable (as in classic MPST), to fully unreliable (every process may
    crash); and
\emph{(2)}
    it is parametric on a behavioural property
    $\predP$ (validated by model checking)
    which can %
    ensure deadlock-freedom, liveness, \etc
    \emph{even in presence of crashes}.
We also present a prototype implementation of our approach.
As future work, we plan to study more crash models (\eg crash-recover) and types
of failure (\eg link failures). We also plan to study the use of
\emph{asynchronous} global types for specifying protocols with failure handling
--- but unlike \cite{DBLP:conf/forte/PetersNW22}, we plan to
support the type-level specification of dedicated recovery behaviours
that are only executed upon crash detection.

%% file: proofs/structural-congruence.tex
\section{Structural Congruence}
\label{sec:structural-congruence}
\label{fig:mpst-pi-congruence}

The structural congruence relation of our MPST $\pi$-calculus, %
mentioned in \Cref{fig:mpst-pi-semantics}, %
is formalised below.
These rules are standard, and taken from~\cite{POPL19LessIsMore}; the only extension is rule $\inferrule{\iruleCongStopElim}$.
\input{proofs/mpst-congruence.tex}

%% file: proofs/mpst-congruence.tex
Here, %
$\fpv{\mpDefD}$ %
is the set of \emph{free process variables} in $\mpDefD$, %
and
$\dpv{\mpDefD}$ %
is the set of \emph{declared process variables} in $\mpDefD$.%
  \[
    \begin{array}{@{\hskip 0mm}c@{\hskip 0mm}}
      \inferenceSingle[\iruleCongParComm]{
        \mpP \mpPar \mpQ
        \equiv
        \mpQ \mpPar \mpP
      }
      \quad\;\;%
      \inferenceSingle[\iruleCongParAssoc]{
        \mpFmt{(\mpP \mpPar \mpQ) \mpPar \mpR}%
        \equiv%
        \mpP \mpPar \mpFmt{(\mpQ \mpPar \mpR)}%
      }
      \quad\;\;%
      \inferenceSingle[\iruleCongParId]{
        \mpP \mpPar \mpNil%
        \equiv%
        \mpP%
      }
      \\[2mm]
      \inferenceSingle[\iruleCongResElim]{
        \mpRes{\mpS}{\mpNil}%
        \equiv%
        \mpNil%
      }
      \quad\;\;%
      \inferenceSingle[\iruleCongResVar]{
        \mpRes{\mpS}{%
          \mpRes{\mpSi}{%
            \mpP%
          }%
        }%
        \equiv%
        \mpRes{\mpSi}{%
          \mpRes{\mpS}{%
            \mpP%
          }%
        }%
      }
      \\[2mm]
      \inferenceSingle[\iruleCongResLift]{
        \mpRes{\mpS}{(\mpP \mpPar \mpQ)}%
        \equiv%
        \mpP \mpPar \mpRes{\mpS}{\mpQ}%
        \;\;%
        \text{\footnotesize{}if $\mpS \!\not\in\! \fc{\mpP}$}%
      }
      \quad\;\;
      \inferenceSingle[\iruleCongStopElim]{
        \mpRes{\mpS}{(\mpStop{\mpS}{\roleP[0]} \mpPar \cdots \mpPar \mpStop{\mpS}{\roleP[n]})}
        \equiv
        \mpNil
      }
      \\[2mm]
      \inferenceSingle[\iruleCongDefElim]{
        \mpDefAbbrev{\mpDefD}{\mpNil}%
        \equiv%
        \mpNil%
      }
      \qquad%
      \inferenceSingle[\iruleCongDefLift]{
        \mpDefAbbrev{\mpDefD}{\mpRes{\mpS}{\mpP}}%
        \,\equiv\,%
        \mpRes{\mpS}{(
          \mpDefAbbrev{\mpDefD}{\mpP}%
          )}%
        \quad%
        \text{\footnotesize{}if\, $\mpS \!\not\in\! \fc{\mpDefD}$}%
      }
      \\[2mm]%
      \inferenceSingle[\iruleCongDefParLift]{
        \mpDefAbbrev{\mpDefD}{(\mpP \mpPar \mpQ)}%
        \,\equiv\,%
        \mpFmt{(\mpDefAbbrev{\mpDefD}{\mpP})} \mpPar \mpQ%
        \quad%
        \text{\footnotesize{}if\, $\dpv{\mpDefD} \cap \fpv{\mpQ} = \emptyset$}%
      }
      \\[2mm]%
      \inferenceSingle[\iruleCongDefOrd]{
        \mpDefAbbrev{\mpDefD}{%
          (\mpDefAbbrev{\mpDefDi}{\mpP})%
        }%
        \;\equiv\;%
        \mpDefAbbrev{\mpDefDi}{%
          (\mpDefAbbrev{\mpDefD}{\mpP})%
        }%
      }
      \\%
      \text{\footnotesize%
        if\, %
        \(%
        (\dpv{\mpDefD} \cup \fpv{\mpDefD}) \cap \dpv{\mpDefDi}%
        \,=\,%
        (\dpv{\mpDefDi} \cup \fpv{\mpDefDi}) \cap \dpv{\mpDefD}%
        \,=\,%
        \emptyset%
        \)%
      }%
      \\[2mm]
    \end{array}
    \]

%% file: proofs/subtyping.tex
\section{Session Subtyping}
\label{sec:subtyping}

We formalise our \emph{subtyping} relation $\stSub$ %
in \cref{def:subtyping} below.
The relation is mostly standard~\cite[Def.\@ 2.5]{POPL19LessIsMore}, except
for the new rule $\inferrule{\iruleStSubStop}$,
and the new ($\highlight{$\text{highlighted}$}$) side condition ``$|I| = 1 \implies \ldots$'' in rule \inferrule{\iruleStSubIn}:
this condition prevents the supertype from adding input branches to ``pure'' crash recovery
external choices.

\input{4_mpst_crash/subtyping.tex}

Rule \inferrule{\iruleStSubGround} lifts $\stSub$ to basic types.
The rest of the rules say that a subtype describes a more permissive session
protocol \wrt its supertype.
By rule \inferrule{\iruleStSubOut}, the subtype of an internal choice
allows for selecting from a wider set of message labels, and sending more
generic payloads.
By rule \inferrule{\iruleStSubIn}, the subtype of an external choice
can support a smaller set of input message labels, and less generic payloads;
the side condition ``$|I| = 1 \ldots$'' ensures that if the subtype only has a
singleton $\stCrashLab$ branch,
then the same applies to the supertype --- hence, both subtype and supertype
describe a ``pure'' crash recovery behaviour,
and do not expect to receive any other input.\footnote{%
    Notice, however, that rule \inferrule{\iruleStSubIn} allows a supertype to
    have a $\stCrashLab$-handling branch even when the subtype does not have
    one.%
}
By rules \inferrule{\iruleStSubEnd} and \inferrule{\iruleStSubStop},
the types $\stEnd$ and $\stStop$ are only subtypes of themselves.
Finally, rules \inferrule{\iruleStSubRecL} and \inferrule{\iruleStSubRecR}
say that recursive types are related up to their unfolding.
\ifTR{We study the properties of session subtyping in \Cref{sec:proofs:subtyping-properties}.}

%% file: 4_mpst_crash/subtyping.tex
\begin{definition}[Subtyping]\label{def:subtyping}
  Given a standard subtyping $\tyGroundSub$ for basic types
  (\eg including $\tyInt \tyGroundSub \tyReal$),
  the \emph{session subtyping relation $\stSub$} is coinductively defined:

\smallskip\centerline{\(
\begin{array}{c}
\cinference[\iruleStSubGround]{
  \tyGround \tyGroundSub \tyGroundi
}{
    \tyGround \stSub \tyGroundi
}

\qquad

\cinference[\iruleStSubEnd]{}{
  \stEnd \stSub \stEnd
}

\qquad

\cinference[\iruleStSubOut]{
  \forall i \in I
  &
  \tySi[i] \stSub \tyS[i]
  &
  \stT[i] \stSub \stTi[i]
}{
  \stIntSum{\roleP}{i \in I \cup J}{\stChoice{\stLab[i]}{\tyS[i]} \stSeq \stT[i]}
  \stSub
  \stIntSum{\roleP}{i \in I}{\stChoice{\stLab[i]}{\tySi[i]} \stSeq \stTi[i]}
}

\\[2ex]
\cinference[\iruleStSubStop]{}{
  \stStop \stSub \stStop
}
\qquad

\cinference[\iruleStSubIn]{
  \forall i \in I
  &
  \tyS[i] \stSub \tySi[i]
  &
  \stT[i] \stSub \tyTi[i]
  &
  \highlight{|I| = 1 \implies \left(\stLab[i] \neq \stCrashLab \text{ or } J = \emptyset\right)}
}{
  \stExtSum{\roleP}{i \in I}{\stChoice{\stLab[i]}{\tyS[i]} \stSeq \stT[i]}%
  \stSub
  \stExtSum{\roleP}{i \in I\cup{}J}{\stChoice{\stLab[i]}{\tySi[i]} \stSeq \stTi[i]}%
}

\\[2ex]
\cinference[\iruleStSubRecL]{
  \stT{}[\stRec{\stRecVar}{\stT}/\stRecVar] \stSub \stTi
}{
  \stRec{\stRecVar}{\stT} \stSub \stTi
}

\qquad

\cinference[\iruleStSubRecR]{
  \stT \stSub \stTi{}[\stRec{\stRecVar}{\stTi}/\stRecVar]
}{
  \stT \stSub \stRec{\stRecVar}{\stTi}
}
\end{array}
\)}%
\end{definition}

%% file: proofs/all_examples.tex
\section{Additional Examples}\label{sec:appendix:examples}

\input{3_calculus/example.tex}
\input{4_mpst_crash/safety-example.tex}

\begin{example} \label{eg:allegs:props}
We illustrate safety, deadlock-freedom, liveness, termination, and never-termination over typing contexts via a series of small examples.
We first consider the typing context $\stEnv[\stFmt{A}] = \stEnv[\stFmt{A\roleP}] \stEnvComp \stEnv[\stFmt{A\roleQ}] \stEnvComp \stEnv[\stFmt{A\roleR}]$ where:

\smallskip
\centerline{\(
\begin{array}{rcl}
\stEnv[\stFmt{A\roleP}] &=&
\stEnvMap{\mpChanRole{\mpS}{\roleP}}{
  \stRec{\stRecVar[\roleP]}{
    \stIntSum{\roleQ}{}{
      \stChoice{\stLabOK}{} \stSeq
      \stExtSum{\roleQ}{}{
        \stChoice{\stLabOK}{} \stSeq
        \stRecVar[\roleP]
        \stEnvComp\;
        \stChoice{\stLabKO}{} \stSeq
        \stEnd
        \stEnvComp\;
        \stChoice{\stCrashLab}{} \stSeq
        \stEnd
      }
      \stEnvComp\;
      \stChoice{\stLabKO}{} \stSeq \stEnd
    }
  }
}
\\
\stEnv[\stFmt{A\roleQ}] &=&
\stEnvMap{\mpChanRole{\mpS}{\roleQ}}{
  \stRec{\stRecVar[\roleQ]}{
    \stExtSum{\roleP}{}{
      \stChoice{\stLabOK}{} \stSeq
      \stIntSum{\roleP}{}{
        \stChoice{\stLabOK}{} \stSeq
        \stRecVar[\roleQ]
        \stEnvComp\;
        \stChoice{\stLabKO}{} \stSeq
        \stEnd
      }
      \stEnvComp\;
      \stChoice{\stLabKO}{} \stSeq
      \stEnd
      \stEnvComp
      \stChoice{\stCrashLab}{} \stSeq
      \stOut{\roleR}{\stLabOK}{} \stSeq
      \stEnd
    }
  }
}
\\
\stEnv[\stFmt{A\roleR}] &=&
\stEnvMap{\mpChanRole{\mpS}{\roleR}}{
  \stIn{\roleP}{\stCrashLab}{}{
    \stIn{\roleQ}{\stLabOK}{}{
    \stEnd
    \stEnvComp\;
    \stChoice{\stCrashLab}{} \stSeq
    \stEnd
    }
  }
}
\end{array}
\)}\smallskip

\noindent
If we assume that all roles in $\stEnv[\stFmt{A}]$ are unreliable,
$\stEnv[\stFmt{A}]$ is safe since its inputs/outputs are dual.
However, $\stEnv[\stFmt{A}]$ is \emph{neither} deadlock-free \emph{nor} live since it is possible for $\roleP$ to crash immediately before $\roleQ$ sends $\stLabKO$ to $\roleP$. In such cases, $\roleQ$ will \emph{not} detect that $\roleP$ has crashed (since we only detect crashes on receive actions) and terminate \emph{without} sending a message to the backup process $\roleR$. This results in a deadlock because  $\roleR$ \emph{will} detect that $\roleP$ has crashed, and \emph{will} expect a message from $\roleQ$.

We observe that changing the reliability assumptions, without changing the typing context, may influence whether a typing context property holds.
For example, consider the typing context $\stEnv[\stFmt{B}] = \stEnv[\stFmt{B\roleP}] \stEnvComp \stEnv[\stFmt{B\roleQ}] \stEnvComp \stEnv[\stFmt{B\roleR}]$ where:

\smallskip
\centerline{\(
\begin{array}{rcl}
\stEnv[\stFmt{B\roleP}] &=&
\stEnvMap{\mpChanRole{\mpS}{\roleP}}{
  \stRec{\stRecVar[\roleP]}{
    \stOut{\roleQ}{\stLabOK}{} \stSeq
    \stRecVar[\roleP]
  }
}
\\
\stEnv[\stFmt{B\roleQ}] &=&
\stEnvMap{\mpChanRole{\mpS}{\roleQ}}{
  \stRec{\stRecVar[\roleQ]}{
    \stIn{\roleP}{\stLabOK}{}{
      \stRecVar[\roleQ]
      \stEnvComp\;
      \stChoice{\stCrashLab}{} \stSeq
      \stRec{\stRecVari[\roleQ]}{
        \stIn{\roleR}{\stLabOK}{}{
          \stRecVari[\roleQ]
          \stEnvComp\;
          \stChoice{\stCrashLab}{} \stSeq
          \stEnd
        }
      }
    }
 }
}
\\
\stEnv[\stFmt{B\roleR}] &=&
\stEnvMap{\mpChanRole{\mpS}{\roleR}}{
  \stRec{\stRecVar[\roleR]}{
    \stOut{\roleQ}{\stLabOK}{} \stSeq
    \stRecVar[\roleR]
  }
}
\end{array}
\)}\smallskip

\noindent
If we assume that all roles are unreliable,
$\stEnv[\stFmt{B}]$ is safe and deadlock-free but \emph{not} live --- because $\roleP$ may never crash, and in this case, $\roleR$'s outputs are never received by $\roleQ$.
Notably, $\stEnv[\stFmt{B}]$ is \emph{not} never-terminating because
if both $\roleP$ and $\roleR$ crash, then the surviving $\roleQ$ can reach $\stEnd$;
however, if we assume that just $\roleR$ is reliable (\ie $\rolesS = \gtFmt{\{\roleR\}}$), then $\stEnv[\stFmt{B}]$ becomes also never-terminating 
--- because even if both $\roleP$ and $\roleQ$ crash, role $\roleR$ can keep running by sending forever $\stLabOK$ messages that are lost (by rule \inferrule{\iruleTCtxSendToCrashed} in \Cref{fig:gtype:tc-red-rules}).

Notice that, in the case of $\stEnv[\stFmt{B}]$, we are unable to make liveness hold purely via combinations of reliable roles: this is because (unless $\roleP$ crashes) $\roleR$'s output will never be received by $\roleQ$, irrespective of reliability assumptions. The typing context itself must instead be adapted; for example, by only permitting $\roleR$ to send once it has detected that $\roleP$ has crashed.

Instead, in the case of $\stEnv[\stFmt{A}]$, we \emph{can} obtain liveness by adjusting the reliability assumptions: in fact, if we assume $\roleR \in \rolesS$, then $\stEnv[\stFmt{A}]$ is both deadlock-free and live.

Finally, consider the typing context $\stEnv[\stFmt{C}] = \stEnv[\stFmt{C\roleP}] \stEnvComp \stEnv[\stFmt{C\roleQ}] \stEnvComp \stEnv[\stFmt{C\roleR}]$ where:

\smallskip
\centerline{\(
\begin{array}{rcl}
\stEnv[\stFmt{C\roleP}] &=&
\stEnvMap{\mpChanRole{\mpS}{\roleP}}{
  \stOut{\roleQ}{\stLab[1]}{} \stSeq
  \stIn{\roleQ}{\stLab[2]}{}{
    \stEnd
    \stEnvComp\;
    \stChoice{\stCrashLab}{} \stSeq
    \stRec{\stRecVar[\roleP]}{
      \stOut{\roleR}{\stLabOK}{} \stSeq
      \stRecVar[\roleP]
    }
  }
}
\\
\stEnv[\stFmt{C\roleQ}] &=&
\stEnvMap{\mpChanRole{\mpS}{\roleQ}}{
  \stIn{\roleP}{\stLab[1]}{}{
    \stOut{\roleP}{\stLab[2]}{} \stSeq
    \stEnd
  }
}
\\
\stEnv[\stFmt{C\roleR}] &=&
\stEnvMap{\mpChanRole{\mpS}{\roleR}}{
  \stIn{\roleP}{\stCrashLab}{}{
    \stRec{\stRecVar[\roleQ]}{
      \stIn{\roleP}{\stLabOK}{}{
        \stRecVar[\roleQ]
      }
    }
  }
}
\end{array}
\)}\smallskip

\noindent
$\stEnv[\stFmt{C}]$ satisfies safety, deadlock-freedom, and termination when \emph{all} roles are assumed to be reliable.
However, should we instead assume that only $\roleP$ is reliable, then $\stEnv[\stFmt{C}]$ does not satisfy termination.
Since external choices in $\stEnv[\stFmt{C}]$ do not feature a crash-handling branch when receiving from $\roleP$, should no roles be assumed reliable, $\stEnv[\stFmt{C}]$ satisfies only safety.
\end{example}

%% file: 3_calculus/example.tex
\begin{example} \label{eg:carried}
  We show an example of our crashing semantics.
  Processes $\mpP$ and $\mpQ$ below communicate on a session $\mpS$;
  $\mpP$ uses the endpoint $\mpChanRole{\mpS}{\roleP}$ to
  send an endpoint $\mpChanRole{\mpS}{\roleR}$ to role $\roleQ$;
  $\mpQ$ uses the endpoint $\mpChanRole{\mpS}{\roleQ}$ to receive an
  endpoint $\mpFmt{x}$, then sends a message to role $\roleP$ via
  $\mpFmt{x}$.
\[
\begin{array}{c}
\mpP = \mpSel{\mpChanRole{\mpS}{\roleP}}{\roleQ}
      {\mpLabi}{\mpChanRole{\mpS}{\roleR}}{%
        \mpBranchSingle{\mpChanRole{\mpS}{\roleP}}{\roleR}{\mpLab}{x}{}
} \qquad
\mpQ = \mpBranchSingle{\mpChanRole{\mpS}{\roleQ}}{\roleP}{\mpLabi}{x}{
  \mpSel{x}{\roleP}{\mpLab}{42}{}
}{}
\end{array}
\]

\noindent
On a successful reduction (without crashes), we have:
\[
  \begin{array}{rl}
  &\mpRes{\mpS}{(\mpP \mpPar \mpQ)} \\
  =&\mpRes{\mpS}{(
    \mpSel{\mpChanRole{\mpS}{\roleP}}{\roleQ}{\mpLabi}{\mpChanRole{\mpS}{\roleR}}{%
        \mpBranchSingle{\mpChanRole{\mpS}{\roleP}}{\roleR}{\mpLab}{x}{}}
    \mpPar
    \mpBranchSingle{\mpChanRole{\mpS}{\roleQ}}{\roleP}{\mpLabi}{x}{\mpSel{x}{\roleP}{\mpLab}{42}{}}
  )}\\
  \mpMove&
  \mpRes{\mpS}{(
    \mpBranchSingle{\mpChanRole{\mpS}{\roleP}}{\roleR}{\mpLab}{x}{}{}
    \mpPar
    \mpSel{\mpChanRole{\mpS}{\roleR}}{\roleP}{\mpLab}{42}{}%
  )}\\
  \mpMove&
  \mpNil
  \end{array}
\]

\noindent
Now, suppose that $\mpP$
crashes before sending; this gives rise to the reduction:
\[
  \begin{array}{rl}
  &\mpRes{\mpS}{(\mpP \mpPar \mpQ)} \\
  =&\mpRes{\mpS}{(
    \mpSel{\mpChanRole{\mpS}{\roleP}}{\roleQ}{\mpLabi}{\mpChanRole{\mpS}{\roleR}}{%
        \mpBranchSingle{\mpChanRole{\mpS}{\roleP}}{\roleR}{\mpLab}{x}{}}
    \mpPar
    \mpBranchSingle{\mpChanRole{\mpS}{\roleQ}}{\roleP}{\mpLabi}{x}{\mpSel{x}{\roleP}{\mpLab}{42}{}}
  )}\\
  \mpMoveCrash&
  \mpRes{\mpS}{(
    \mpStop{\mpS}{\roleP}
    \mpPar
    \mpStop{\mpS}{\roleR}
    \mpPar
    \mpBranchSingle{\mpChanRole{\mpS}{\roleQ}}{\roleP}{\mpLabi}{x}{\mpSel{x}{\roleP}{\mpLab}{42}{}}
  )}
  \end{array}
\]
\noindent
We can observe that when the sending process $\mpP$ crashes (by
$\inferrule{\iruleMPCrashS}$), all
endpoints in $\mpP$ (\ie both $\mpChanRole{\mpS}{\roleP}$ and
$\mpChanRole{\mpS}{\roleR}$) crash.
If $\mpQ$ has a crash handling branch, it can be triggered via
$\inferrule{\iruleMPRedCommD}$, suppose instead we have
\[
  \mpQi =
    \mpBranch{\mpChanRole{\mpS}{\roleQ}}{\roleP}{}{\mpLabi}{x}{\mpSel{x}{\roleP}{\mpLab}{42}{}}{\mpNil}
\]
\noindent
A crash handling reduction can trigger when $\mpP$ crashes:
\[
  \begin{array}{rl}
  &\mpRes{\mpS}{(\mpP \mpPar \mpQi)} \\
  =&\mpRes{\mpS}{(
    \mpSel{\mpChanRole{\mpS}{\roleP}}{\roleQ}{\mpLabi}{\mpChanRole{\mpS}{\roleR}}{%
        \mpBranchSingle{\mpChanRole{\mpS}{\roleP}}{\roleR}{\mpLab}{x}{}}
    \mpPar
    \mpBranch{\mpChanRole{\mpS}{\roleQ}}{\roleP}{}{\mpLabi}{x}{\mpSel{x}{\roleP}{\mpLab}{42}{}}{\mpNil}
  )}\\
  \mpMoveCrash&
  \mpRes{\mpS}{(
    \mpStop{\mpS}{\roleP}
    \mpPar
    \mpStop{\mpS}{\roleR}
    \mpPar
    \mpBranch{\mpChanRole{\mpS}{\roleQ}}{\roleP}{}{\mpLabi}{x}{\mpSel{x}{\roleP}{\mpLab}{42}{}}{\mpNil}
  )}\\
  \mpMove &
  \mpRes{\mpS}{(
    \mpStop{\mpS}{\roleP}
    \mpPar
    \mpStop{\mpS}{\roleR}
    \mpPar
    \mpNil
  )}
  \end{array}
\]

\end{example}

%% file: 4_mpst_crash/safety-example.tex
\begin{example} \label{eg:running:type-safe}
Recall the types of the DNS example in \cref{sec:overview}:

\smallskip
\centerline{\(
\stTi[\roleP] = \stOut{\roleQ}{\stLabFmt{req}}{} \stSeq
  \stExtSum{\roleQ}{}{
    \begin{array}{@{}l@{}}
    \stChoice{\stLabFmt{res}}{} \stSeq \stEnd \\
    \stChoice{\stCrashLab}{} \stSeq
    \stOut{\roleR}{\stLabFmt{req}}{} \stSeq
    \stInNB{\roleR}{\stLabFmt{res}}{}{} \stSeq
    \stEnd
    \end{array}
  }
\qquad
\begin{array}{l}
\stTi[\roleQ] = \stInNB{\roleP}{\stLabFmt{req}}{}{
  \stOut{\roleP}{\stLabFmt{res}}{} \stSeq \stEnd
}
\\
\stTi[\roleR] = \stInNB{\roleQ}{\stCrashLab}{}{} \stSeq
  \stInNB{\roleP}{\stLabFmt{req}}{}{} \stSeq
  \stOut{\roleP}{\stLabFmt{res}}{}{} \stSeq \stEnd
\end{array}
\)}
\smallskip

\noindent%
Now, consider the following typing context, containing such types:

\smallskip
\centerline{
\(
  \stEnv \;=\; \stEnvMap{%
      \mpChanRole{\mpS}{\roleP}%
    }{%
      \stTi[\roleP]
    }
    \stEnvComp
    \stEnvMap{
      \mpChanRole{\mpS}{\roleQ}
    }{
      \stTi[\roleQ]
    }
    \stEnvComp
    \stEnvMap{
      \mpChanRole{\mpS}{\roleR}
    }{
      \stTi[\roleR]
    }
\)
}%
\smallskip

Such $\stEnv$ is $(\mpS;\setenum{\roleP,\roleR})$-safe. We can
verify it by checking its reductions.
When no crashes occur, we have the following two reductions, where each reductum satisfies \Cref{def:mpst-env-safe}:

\smallskip
\centerline{
\(\small
  \begin{array}{@{}r@{\;\;}c@{\;\;}l@{}}
  \stEnv & \stEnvMoveMaybeCrash[\mpS; \setenum{\roleP,\roleR}] & \stEnvMap{%
      \mpChanRole{\mpS}{\roleP}%
    }{%
      \stExtSum{\roleQ}{}{
        \begin{array}{@{}l@{}}
        \stChoice{\stLabFmt{res}}{} %
        \\
        \stChoice{\stCrashLab}{} \stSeq
        \stOut{\roleR}{\stLabFmt{req}}{} \stSeq
        \stInNB{\roleR}{\stLabFmt{res}}{}{} %
        \end{array}
      }
    }
    \stEnvComp
    \stEnvMap{
      \mpChanRole{\mpS}{\roleQ}
    }{
      \stOut{\roleP}{\stLabFmt{res}}{} %
    }
    \stEnvComp
    \stEnvMap{
      \mpChanRole{\mpS}{\roleR}
    }{
      \stTi[\roleR]
    }
    \\&
    \stEnvMoveMaybeCrash[\mpS; \setenum{\roleP,\roleR}] %
    &
    \stEnvMap{%
      \mpChanRole{\mpS}{\roleP}%
    }{%
      \stEnd
    }
    \stEnvComp
    \stEnvMap{
      \mpChanRole{\mpS}{\roleQ}
    }{
      \stEnd
    }
    \stEnvComp
    \stEnvMap{
      \mpChanRole{\mpS}{\roleR}
    }{
      \stTi[\roleR]
    }
  \end{array}
\)
}%
\smallskip

\noindent%
In the case where $\roleQ$ crashes immediately,
we have:

\smallskip
\centerline{
\(\small
  \begin{array}{@{}r@{\;\;}c@{\;\;}l@{}}
  \stEnv & \stEnvMoveMaybeCrash[\mpS; \setenum{\roleP,\roleR}] & \stEnvMap{%
      \mpChanRole{\mpS}{\roleP}%
    }{%
      \stTi[\roleP]
    }
    \stEnvComp\,
    \stEnvMap{
      \mpChanRole{\mpS}{\roleQ}
    }{
      \stStop
    }
    \stEnvComp\,
    \stEnvMap{
      \mpChanRole{\mpS}{\roleR}
    }{
      \stTi[\roleR]
    }
    \\&
    \stEnvMoveMaybeCrash[\mpS; \setenum{\roleP,\roleR}]
    &
    \stEnvMap{%
      \mpChanRole{\mpS}{\roleP}%
    }{%
      \stExtSum{\roleQ}{}{
        \stChoice{\stLabFmt{res}}{},\, %
        \stChoice{\stCrashLab}{} \stSeq
        \stOut{\roleR}{\stLabFmt{req}}{} \stSeq
        \stInNB{\roleR}{\stLabFmt{res}}{}{} %
      }
    }
    \stEnvComp\,
    \stEnvMap{
      \mpChanRole{\mpS}{\roleQ}
    }{
      \stStop
    }
    \stEnvComp\,
    \stEnvMap{
      \mpChanRole{\mpS}{\roleR}
    }{
      \stTi[\roleR]
    }
    \\
    & \stEnvMoveMaybeCrash[\mpS; \setenum{\roleP,\roleR}] & \stEnvMap{%
      \mpChanRole{\mpS}{\roleP}%
    }{%
      \stOut{\roleR}{\stLabFmt{req}}{} \stSeq
      \stInNB{\roleR}{\stLabFmt{res}}{}{} %
    }
    \stEnvComp\,
    \stEnvMap{
      \mpChanRole{\mpS}{\roleQ}
    }{
      \stStop
    }
    \stEnvComp\,
    \stEnvMap{
      \mpChanRole{\mpS}{\roleR}
    }{
      \stTi[\roleR]
    }
    \\&
    \stEnvMoveMaybeCrash[\mpS; \setenum{\roleP,\roleR}]
    &
    \stEnvMap{%
      \mpChanRole{\mpS}{\roleP}%
    }{%
      \stOut{\roleR}{\stLabFmt{req}}{} \stSeq
      \stInNB{\roleR}{\stLabFmt{res}}{}{} %
    }
    \stEnvComp\,
    \stEnvMap{
      \mpChanRole{\mpS}{\roleQ}
    }{
      \stStop
    }
    \stEnvComp\,
    \stEnvMap{
      \mpChanRole{\mpS}{\roleR}
    }{
      \stInNB{\roleP}{\stLabFmt{req}}{}{} \stSeq
      \stOut{\roleP}{\stLabFmt{res}}{}{} %
    }
    \\
    & \stEnvMoveMaybeCrash[\mpS; \setenum{\roleP,\roleR}] & \stEnvMap{%
      \mpChanRole{\mpS}{\roleP}%
    }{%
      \stInNB{\roleR}{\stLabFmt{res}}{}{} %
    }
    \stEnvComp\,
    \stEnvMap{
      \mpChanRole{\mpS}{\roleQ}
    }{
      \stStop
    }
    \stEnvComp\,
    \stEnvMap{
      \mpChanRole{\mpS}{\roleR}
    }{
      \stOut{\roleP}{\stLabFmt{res}}{} %
    }
    \\&
    \stEnvMoveMaybeCrash[\mpS; \setenum{\roleP,\roleR}]
    &
    \stEnvMap{%
      \mpChanRole{\mpS}{\roleP}%
    }{%
      \stEnd
    }
    \stEnvComp\,
    \stEnvMap{
      \mpChanRole{\mpS}{\roleQ}
    }{
      \stStop
    }
    \stEnvComp\,
    \stEnvMap{
      \mpChanRole{\mpS}{\roleR}
    }{
      \stEnd
    }
  \end{array}
\)
}%
\smallskip

\noindent%
and each reductum satisfies \Cref{def:mpst-env-safe}.
The case where $\roleQ$ crashes after receiving the $\stLabFmt{req}$uest is similar.
There are no other crash reductions to consider, since $\roleP$ and
$\roleR$ are reliable.
\end{example}

%% file: proofs/results.tex
\begin{figure}
{\footnotesize
\begin{tabular}{@{}l@{\hspace{0mm}}l}
\toprule
($\alpha$) &
\begin{minipage}{0.88\textwidth}
\(
\begin{array}{l}
\stEnvMap{\mpChanRole{\mpS}{\roleP}}{
\stOut{\roleQ}{\stLabFmt{req}}{} \stSeq
  \stIn{\roleQ}{\stLabFmt{res}}{}{\stEnd
    \stEnvComp\;
    \stChoice{\stCrashLab}{} \stSeq
    \stOut{\roleR}{\stLabFmt{req}}{} \stSeq
    \stInNB{\roleR}{\stLabFmt{res}}{}{} \stSeq
    \stEnd
  }
}
\\
\stEnvMap{\mpChanRole{\mpS}{\roleQ}}{
\stInNB{\roleP}{\stLabFmt{req}}{}{
  \stOut{\roleP}{\stLabFmt{res}}{} \stSeq \stEnd
}
}
\\
\stEnvMap{\mpChanRole{\mpS}{\roleR}}{
\stInNB{\roleQ}{\stCrashLab}{}{} \stSeq
  \stInNB{\roleP}{\stLabFmt{req}}{}{} \stSeq
  \stOut{\roleP}{\stLabFmt{res}}{}{} \stSeq \stEnd
}
\end{array}
\)
\end{minipage}
\\
\midrule
($\beta$) &
\begin{minipage}{0.88\textwidth}
\(
\begin{array}{l}
\stEnvMap{\mpChanRole{\mpS}{\roleP}}{
  \stRec{\stRecVar}{
    \stOut{\roleQ}{\stLabFmt{add}}{\stTypeInt}
      \stSeq
      \stIntSum{\roleQ}{}{
        \stLabFmt{add}(\stTypeInt) \stSeq
          \stIn{\roleQ}{\stLabFmt{res}}{\stTypeInt}{
            \stRecVar
            \stEnvComp\;
            \stT[\lightning]
          }
        \stEnvComp\;
        \stLabKO \stSeq
          \stIn{\roleQ}{\stT[\stLabKO]}{}{
            \stEnvComp\;
            \stT[\lightning]
          }
      }
  }
}
\\
\stEnvMap{\mpChanRole{\mpS}{\roleQ}}{
  \stRec{\stRecVar}{
    \stIn{\roleP}{\stLabFmt{add}}{\stTypeInt}{
      \stIn{\roleP}{\stLabFmt{add}}{\stTypeInt}{
          \stOut{\roleP}{\stLabFmt{res}}{\stTypeInt} \stSeq \stRecVar
        \stEnvComp\;
        \stLabKO \stSeq
          \stOut{\roleP}{\stT[\stLabKO]}{} %
        \stEnvComp\;
        \stT[\lightning]
        }
      \stEnvComp\;
      \stT[\lightning]
    }
  }
}
\end{array}
\)
\end{minipage}
\\
\midrule
($\gamma$) &
\begin{minipage}{0.88\textwidth}
\(
\begin{array}{l}
  \stEnvMap{\mpChanRole{\mpS}{\roleFmt{b1}}}{
    \stOut{\roleS}{\stLabFmt{r}}{\stTypeString} \stSeq
    \stIn{\roleS}{\stLabFmt{q}}{\stTypeInt}{
      \stOut{\roleFmt{b2}}{\stLabFmt{s}}{\stTypeInt} \stSeq
        \stIn{\roleFmt{b2}}{\stCrashLab}{}{ \stT[1] }
      \stEnvComp\;
      \stCrashLab \stSeq
      \stOut{\roleFmt{b2}}{\stT[\stLabKO]}{} %
    }
  }
  \\
  \hspace{2.5mm}\stT[1] = \stIn{\roleS}{\stLabFmt{rp1}}{}{
    \stIntSum{\roleS}{}{
      \stLabOK \stSeq \stT[2] %
      \stEnvComp\;
      \stT[\stLabKO]
    }
    \stEnvComp\;
    \stLabFmt{rp2} \stSeq \stT[2] %
    \stEnvComp\;
    \stLabFmt{rp3} \stSeq \stIn{\roleS}{\stLabFmt{d}}{\stTypeString}{\stEnd \stEnvComp\; \stT[\lightning]}
    \stEnvComp\;
    \stT[\lightning]
  }
  \,
  \\
  \hspace{2.5mm}\stT[2] = \stOut{\roleS}{\stLabFmt{a}}{\stTypeString} \stSeq \stIn{\roleS}{\stLabFmt{d}}{\stTypeString}{\stEnd \stEnvComp\; \stT[\lightning]}
  \\
  \stEnvMap{\mpChanRole{\mpS}{\roleFmt{b2}}}{
    \stIn{\roleS}{\stLabFmt{q}}{\stTypeInt}{
        \stT[1]
      \stEnvComp\;
      \stT[\stLabKO]
      \stEnvComp\;
      \stCrashLab \stSeq \stT[1]
    }
  }
  \\
  \hspace{2.5mm}\stT[1] = \stIn{\roleFmt{b1}}{\stLabFmt{s}}{\stTypeInt}{
    \stIntSum{\roleS}{}{
      \stLabOK \stSeq
      \stOut{\roleS}{\stLabFmt{a}}{\stTypeString} \stSeq
      \stIn{\roleS}{\stLabFmt{d}}{\stTypeString}{
        \stEnd
        \stEnvComp\;
        \stT[\lightning]
      }
      \stEnvComp\;
      \stT[\stLabKO]
    }
    \stEnvComp\;
    \stT[\stLabKO]
    \stEnvComp\;
    \stCrashLab \stSeq
    \stOut{\roleS}{\stT[\stLabKO]}{} %
  }

  \\
  \stEnvMap{\mpChanRole{\mpS}{\roleS}}{
    \stIn{\roleFmt{b1}}{\stLabFmt{r}}{\stTypeString}{
      \stOut{\roleFmt{b1}}{\stLabFmt{q}}{\stTypeInt} \stSeq
      \stOut{\roleFmt{b2}}{\stLabFmt{q}}{\stTypeInt} \stSeq
      \stIn{\roleFmt{b2}}{\stLabOK}{}{
        \stT[1]
        \stEnvComp\;
        \stT[\stLabKO]
        \stEnvComp\;
        \stCrashLab \stSeq
        \stOut{\roleFmt{b1}}{\stLabFmt{rp1}}{} \stSeq
        \stT[2]
      }
      \stEnvComp\;
      \stCrashLab \stSeq
      \stOut{\roleFmt{b2}}{\stT[\stLabKO]}{} %
    }
  }
  \\
  \hspace{2.5mm}\stT[1] = \stIn{\roleFmt{b2}}{\stLabFmt{a}}{\stTypeString}{
    \stOut{\roleFmt{b2}}{\stLabFmt{d}}{\stTypeString} \stSeq
    \stIn{\roleFmt{b2}}{\stCrashLab}{}{
      \stOut{\roleFmt{b1}}{\stLabFmt{rp3}}{}
      \stSeq
      \stT[4]
    }
    \stEnvComp\;
    \stCrashLab \stSeq
    \stOut{\roleFmt{b1}}{\stLabFmt{rp2}}{}
    \stSeq
    \stT[3]
  }
  \\
  \hspace{2.5mm}\stT[2] = \stIn{\roleFmt{b1}}{\stLabOK}{}{
    \stT[3]
    \stEnvComp\;
    \stT[\stLabKO]
    \stEnvComp\;
    \stT[\lightning]
  }
  \;\;
  \hspace{5mm}\stT[3] = \stIn{\roleFmt{b1}}{\stLabFmt{a}}{\stTypeString}{
    \stT[4]
    \stEnvComp\;
    \stT[\lightning]
  }
  \;\;
  \hspace{5mm}\stT[4] = \stOut{\roleFmt{b1}}{\stLabFmt{d}}{\stTypeString} \stSeq \stEnd
  \end{array}
\)
\end{minipage}
\\
\midrule
($\delta$) &
\begin{minipage}{0.95\textwidth}
\(
\begin{array}{l}
\stEnvMap{\mpChanRole{\mpS}{\roleFmt{n}}}{
  \stIn{\roleFmt{c}}{\stLabFmt{o}}{\stTypeInt}{
    \stRec{\stRecVar}{
      \stOut{\roleFmt{c}}{\stLabFmt{g}}{} \stSeq
      \stIntSum{\roleFmt{c}}{}{
        \stLabFmt{o}(\stTypeInt) \stSeq
        \stIn{\roleFmt{c}}{\stLabFmt{o}}{\stTypeInt}{
          \stRecVar
          \stEnvComp\;
          \stLabOK \stSeq
          \stOut{\roleFmt{c}}{\stT[\stLabOK]}{} %
          \stEnvComp\;
          \stT[\stLabKO]
          \stEnvComp\;
          \stT[\lightning]
        }
        \stEnvComp\;
        \stLabOK \stSeq
        \stExtSum{\roleFmt{c}}{}{
          \stT[\stLabOK]
          \stEnvComp\;
          \stT[\lightning]
        }
        \stEnvComp\;
        \stT[\stLabKO]
      }
    }
    \stEnvComp\;
    \stT[\lightning]
  }
}
\\
\stEnvMap{\mpChanRole{\mpS}{\roleFmt{c}}}{
\stOut{\roleFmt{n}}{\stLabFmt{o}}{\stTypeInt} \stSeq
\stRec{\stRecVar[0]}{
  \stIn{\roleFmt{n}}{\stLabFmt{g}}{}{
    \stIn{\roleFmt{n}}{\stLabFmt{o}}{\stTypeInt}{
      \stT[1]
      \stEnvComp\;
      \stLabOK \stSeq
      \stIn{\roleFmt{n}}{\stCrashLab}{}{
        \stOut{\roleFmt{b}}{\stT[\stLabOK]}{} %
      }
      \stEnvComp\;
      \stT[\stLabKO]
      \stEnvComp\;
      \stCrashLab \stSeq
      \stT[2]
    }
    \stEnvComp\;
    \stCrashLab \stSeq
    \stT[2]
  }
}
}
\\
\hspace{2.5mm}\stT[1] =
\stIntSum{\roleFmt{n}}{}{
  \stLabFmt{o}(\stTypeInt) \stSeq
  \stRecVar[0]
  \stEnvComp\;
  \stLabOK \stSeq
  \stExtSum{\roleFmt{n}}{}{
    \stT[\stLabOK]
    \stEnvComp\;
    \stCrashLab \stSeq
    \stT[2]
  }
  \stEnvComp\;
  \stLabKO \stSeq
  \stIn{\roleFmt{n}}{\stCrashLab}{}{
    \stOut{\roleFmt{b}}{\stT[\stLabKO]}{} %
  }
}
\\
\hspace{2.5mm}\stT[2] = \stOut{\roleFmt{b}}{\stLabFmt{o}}{\stTypeInt} \stSeq
\stRec{\stRecVar[1]}{
  \stIn{\roleFmt{b}}{\stLabFmt{g}}{}{
    \stIn{\roleFmt{b}}{\stLabFmt{o}}{\stTypeInt}{
      \stIntSum{\roleFmt{b}}{}{
        \stLabFmt{o}(\stTypeInt) \stSeq
        \stRecVar[1]
        \stEnvComp\;
        \stLabOK \stSeq
        \stExtSum{\roleFmt{b}}{}{\stT[\stLabOK]}
        \stEnvComp\;
        \stT[\stLabKO]
      }
      \stEnvComp\;
      \stLabOK \stSeq
      \stExtSum{\roleFmt{b}}{}{\stT[\stLabOK]}
      \stEnvComp\;
      \stT[\stLabKO]
    }
  }
}
\\
\stEnvMap{\mpChanRole{\mpS}{\roleFmt{b}}}{
  \stIn{\roleFmt{n}}{\stCrashLab}{}{
    \stIn{\roleFmt{c}}{\stLabFmt{o}}{\stTypeInt}{
      \stRec{\stRecVar}{
        \stOut{\roleFmt{c}}{\stLabFmt{g}}{} \stSeq
        \stIntSum{\roleFmt{c}}{}{
          \stLabFmt{o}(\stTypeInt) \stSeq
          \stT[1]
          \stEnvComp\;
          \stLabOK \stSeq
          \stT[2]
          \stEnvComp\;
          \stT[\stLabKO]
        }
      }
      \stEnvComp\;
      \stT[\stLabOK]
      \stEnvComp\;
      \stT[\stLabKO]
      \stEnvComp\;
      \stT[\lightning]
    }
  }
}
\\
\hspace{2.5mm}\stT[1] =
\stIn{\roleFmt{c}}{\stLabFmt{o}}{\stTypeInt}{
  \stRecVar
  \stEnvComp\;
  \stLabOK \stSeq
  \stOut{\roleFmt{c}}{\stT[\stLabOK]}{} %
  \stEnvComp\;
  \stT[\stLabKO]
}
\quad
\stT[2] =
\stIn{\roleFmt{c}}{\stLabOK}{}{
  \stOut{\roleFmt{c}}{\stT[\stLabOK]}{} %
  \stEnvComp\;
  \stT[\lightning]
}
\end{array}
\)
\end{minipage}
\\
\midrule
($\varepsilon$) &
\begin{minipage}{0.88\textwidth}
\(
\begin{array}{l}
\stEnvMap{\mpChanRole{\mpS}{\roleP}}{
  \stOut{\roleQ}{\stLabFmt{data}}{\stTypeString} \stSeq
  \stOut{\roleR}{\stLabFmt{data}}{\stTypeString} \stSeq
  \stEnd
}
\\
\stEnvMap{\mpChanRole{\mpS}{\roleQ}}{
  \stIn{\roleP}{\stLabFmt{data}}{\stTypeString}{
    \stIn{\roleP}{\stCrashLab}{}{
      \stIn{\roleR}{\stLabFmt{h}}{}{
        \stOut{\roleR}{\stLabFmt{data}}{\stTypeString} \stSeq
        \stEnd
        \stEnvComp\;
        \stT[\lightning]
      }
    }
    \stEnvComp\;
    \stCrashLab \stSeq
    \stIn{\roleR}{\stLabFmt{req}}{}{
      \stOut{\roleR}{\stT[\stLabKO]}{} %
      \stEnvComp\;
      \stT[\lightning]
    }
  }
}
\\
\stEnvMap{\mpChanRole{\mpS}{\roleR}}{
  \stIn{\roleP}{\stLabFmt{data}}{\stTypeString}{
    \stEnd
    \stEnvComp\;
    \stCrashLab \stSeq
    \stOut{\roleQ}{\stLabFmt{req}}{} \stSeq
    \stIn{\roleQ}{\stLabFmt{data}}{\stTypeString}{
      \stEnd
      \stEnvComp\;
      \stT[\stLabKO]
      \stEnvComp\;
      \stT[\lightning]
    }
  }
}
\end{array}
\)
\end{minipage}
\\
\bottomrule
\end{tabular}
}
\caption{Typing contexts for ($\alpha$) \exampleName{DNS}, ($\beta$) \exampleName{Adder}, ($\gamma$) \exampleName{TwoBuyers}, ($\delta$) \exampleName{Negotiate}, and ($\varepsilon$) \exampleName{Broadcast}. Roles $\roleP$ and $\roleQ$ of \exampleName{DNS}, and $\roleFmt{b}$ of \exampleName{Negotiate} are reliable; all other roles are unreliable.
Let $\stT[\lightning] = \stCrashLab \stSeq \stEnd$, $\stT[\stLabFmt{ko}] = \stLabKO \stSeq \stEnd$, and $\stT[\stLabOK] = \stLabOK \stSeq \stEnd$.}
\label{fig:eval-examples-types}
\end{figure}

\section{Tool Evaluation}
\label{sec:appdx:examples}

To verify the properties in~\cref{fig:mc-formulae},
we extend the Multiparty Session Types
toolKit (\theTool) \cite{SY19Artifact}, which uses the mCRL2 model
checker~\cite{TACAS19mCRL2}. Our extended tool is available at:

\smallskip\centerline{
  \url{https://github.com/alcestes/mpstk-crash-stop}
}\smallskip%

We evaluate our approach with 5 examples:
\exampleName{DNS}, from \Cref{sec:overview};
\exampleName{Adder}, \exampleName{Two Buyers}, and
\exampleName{Negotiate}, extended from the session type
literature~\cite{OOPSLA20VerifiedRefinements} with crashes and crash handling behaviour;
and
\exampleName{Broadcast}, inspired by the reliable broadcast algorithms
in~\cite[Ch.\ 3]{DBLP:books/daglib/0025983}.
The full typing contexts for each example are given in
\Cref{fig:eval-examples-types}. 
We model and verify both fully reliable and (partially) unreliable versions of each example. In all examples, we show how the introduction of unreliability leads to an increase of model sizes and verification times.
The increased model size reflects how the addition of crash handling can
complicate even simple protocols, and motivates the use of automatic model
checking. Still, we show that the verification of our examples always completes in less than 100 ms.

\subsection{Description of the Examples in \Cref{fig:eval-examples-types}}

\begin{description}
\item[DNS] is the example described in \cref{sec:overview}. The example demonstrates both backup processes and optional reliability assumptions.

\item[Adder] demonstrates a minimal extension of the fully reliable protocol, in
which $\roleQ$ receives two numbers from $\roleP$, sums them, and communicates
the result to $\roleP$. In our extension, both roles are unreliable and the
protocol ends when a crash is detected. It satisfies safety, deadlock-freedom,
and liveness.

\item[TwoBuyers] is the example described in \Cref{sec:model-checking}. It assumes that both the $\roleFmt{s}$eller and buyers $\roleFmt{b1}$ and $\roleFmt{b2}$ are unreliable. In cases where the split has been agreed upon, and $\roleFmt{b2}$ has crashed, $\roleFmt{b1}$ concludes the sale. It satisfies safety, deadlock-freedom, liveness, and terminating.
This form of \exampleName{TwoBuyers} is not projectable from a global type, since $\roleFmt{b1}$ would need to be informed on conclusion of a sale. \exampleName{TwoBuyers} uses recovery behaviour in order to satisfy deadlock-freedom.
Finally, \exampleName{TwoBuyers} demonstrates the flexibility of crash-handling that our approach permits: $\roleFmt{b2}$ does not alter its behaviour having detected that $\roleFmt{s}$ has crashed (\ie continues as $\stT[1]$), instead leaving $\roleFmt{b1}$ to instigate crash-handling behaviour.

\item[Negotiate] introduces a (reliable) $\roleFmt{b}$ackup negotiator to the
version found in the literature. During normal operation, a $\roleFmt{c}$lient
will send an opening offer to a $\roleFmt{n}$egotiator. Both $\roleFmt{c}$ and
$\roleFmt{n}$ can then choose to repeatedly exchange counter offers until the
other accepts the offer, or rejects it outright, bringing the protocol to an
end.
In our extension, should the
$\roleFmt{c}$ustomer detect that the original $\roleFmt{n}$egotiator crashes,
the $\roleFmt{b}$ackup negotiator activates and continues the negotiation with
$\roleFmt{c}$.
The example satisfies safety, deadlock-freedom, and liveness.
Recovery actions are necessary for $\roleFmt{c}$ in two locations in order to avoid deadlocks: it is otherwise possible for an $\stLabFmt{o}$ffer to be declined or agreed upon, then for $\roleFmt{n}$ to crash without $\roleFmt{c}$ noticing; this results in $\roleFmt{b}$ activating, and expecting a message from the terminated $\roleFmt{c}$.

\item[Broadcast] contains an unreliable broadcaster $\roleP$ attempting to send $\stLabFmt{data}$ to two receivers $\roleQ$ and $\roleR$. In cases where $\roleP$ crashes, $\roleR$ $\stLabFmt{req}$uests the data from $\roleQ$, who
responds with the data it received before $\roleP$ crashed, or with $\stLabKO$ when $\roleP$ crashed immediately.
The example is not projectable from a global type, since $\roleQ$ would otherwise require a message from $\roleR$ even when $\roleP$ had not crashed.
\exampleName{Broadcast} satisfies safety, deadlock-freedom, liveness, and termination.
As in \exampleName{Negotiate}, recovery behaviour is necessary for \exampleName{Broadcast} to satisfy deadlock-freedom.
\end{description}

Notably, \exampleName{Adder}, \exampleName{TwoBuyers} and \exampleName{Broadcast}
have \emph{no} reliability assumptions: any role may crash at any point.
Barring \exampleName{Adder}, our examples
cannot be written using \emph{global types} in the session types literature.
This demonstrates the flexibility of our
generalised MPST system over the classic one.
Moreover, the examples include the use of failover processes
(\exampleName{DNS} and \exampleName{Negotiate})
and complex recovery behaviour
(\exampleName{TwoBuyers},
\exampleName{Negotiate}, and \exampleName{Broadcast}), thus
showcasing the expressivity of our approach.

\begin{table}
\begin{center}
\footnotesize
\begin{tabular}{@{\hspace{0mm}}c@{\hspace{3mm}}C@{\hspace{3mm}}C@{\hspace{3mm}}C@{\hspace{3mm}}C@{\hspace{3mm}}C@{\hspace{3mm}}C@{\hspace{3mm}}C@{\hspace{3mm}}C}
\toprule
& $\rolesS$ & \text{states} & \text{transitions} & \text{safe} & \text{df} & \text{live} & \text{nterm} & \text{term} \\
\midrule
\multirow{2}{5mm}{($\alpha$)} & \{\roleP,\roleR\} & 101 & 427 & 12.28	\pm 1\% &	17.14	\pm 1\% &	11.24	\pm 1\% &	15.47	\pm 0\% &	12.33	\pm 0\% \\
& \roleSet & 10	& 15 & 7.61 \pm 1\% & 8.23 \pm 1\% & 7.46 \pm 1\% & 7.78 \pm 1\% & 7.6 \pm 1\%
\\
\midrule
\multirow{2}{5mm}{($\beta$)} & \rolesSEmpty & 37 & 159 & 12.43	\pm 0\% &	15.74	\pm 0\% &	12.24	\pm 1\% &	14.46	\pm 0\% & 12.06	\pm 1\% \\
& \roleSet & 26 & 56 & 8.92 \pm 2\% & 10.06 \pm 0\% & 8.71 \pm 1\% & 9.42 \pm 0\% & 8.79 \pm 0\%
\\
\midrule
\multirow{2}{5mm}{($\gamma$)} & \rolesSEmpty & 1409 & 10248 & 45.6 \pm 0\% & 88.26 \pm 0\% & 31.33 \pm 0\% & 77.2 \pm 0\% & 45.65 \pm 0\% \\
& \roleSet & 169 & 510 & 11.12 \pm 1\% & 15.94 \pm 0\% & 10.9 \pm 1\% & 12.19 \pm 0\% & 11.06 \pm 1\%
\\
\midrule
\multirow{2}{5mm}{($\delta$)} & \gtFmt{\{\roleFmt{b}\}} & 1089 & 8106 & 34.61 \pm 0\% & 55.07 \pm 0\% & 25.69 \pm 0\%	& 47.46 \pm 0\% & 26.04 \pm 0\% \\
& \roleSet & 50 & 157 & 10.17 \pm 0\% & 12.7 \pm 0\% & 9.93 \pm 0\%	& 11.33 \pm 0\% & 9.72 \pm 0\%
\\
\midrule
\multirow{2}{5mm}{($\varepsilon$)} & \rolesSEmpty & 161 & 925 & 17.99 \pm 1\% & 28.13 \pm 0\% & 14.08 \pm 0\% & 25.72 \pm 1\% & 17.74 \pm 0\% \\
& \roleSet & 13 & 25 & 7.85 \pm 3\% & 8.65 \pm 1\% & 7.7 \pm 0\% & 8.12 \pm 1\% & 7.85 \pm 0\% \\
\bottomrule
\end{tabular}
\end{center}
\caption{Average times (in milliseconds $\pm$ std.\ dev.) for the verification of
\exampleName{DNS} ($\alpha$), \exampleName{Adder} ($\beta$), \exampleName{TwoBuyers} ($\gamma$), \exampleName{Negotiate} ($\delta$), and \exampleName{Broadcast} ($\varepsilon$) in \cref{fig:eval-examples-types}
over safety (safe), deadlock-freedom (df), liveness (live), never-terminating (nterm) and terminating (term).
Each example has two rows of measurements, varying the sets of reliable roles $\rolesS$:
either zero/one/two reliable roles (first row), or all reliable roles (second row).
(Benchmarking specs: Intel Core i7-7700K CPU, 4.20 GHz, 16 GB RAM, mCRL2 202106.0 invoked 30 times with: \texttt{pbes2bool -{}-solve-strategy=2}.)
}
\label{fig:eval-examples-results-full}
\end{table}

\subsection{Experimental Results}

We applied our extended implementation of \theTool to the examples in \cref{fig:eval-examples-types}.
\cref{fig:eval-examples-results-full} gives the full set of verification times, reported in milliseconds with standard deviations, where each time is an average of 30 runs.
\ifTR{
  These results were generated by running \theTool with the \texttt{-{}-benchmark=30} option.
  The number of states reported measures the number of states in the LTS generated by mCRL2 and was procured using the \texttt{ltsinfo} command, via the \texttt{-s} \theTool option. Similarly, the number of transitions generated was retrieved from the \texttt{ltsinfo} command via an extension to \theTool.
}

For each example, we give verification times for both the typing contexts in \cref{fig:eval-examples-types} and a corresponding fully reliable version (\ie where all roles in the protocol are reliable; $\rolesS = \roleSet$). For \exampleName{Adder}, \exampleName{TwoBuyers}, and \exampleName{Negotiate}, we use the standard protocol definitions from the literature. For \exampleName{DNS} and \exampleName{Broadcast}, we omit crash-handling branches. For \exampleName{DNS}, this has the consequence of removing the backup role $\roleR$ entirely.

All examples satisfy safety,
deadlock-freedom, and liveness;
\exampleName{Adder} and \exampleName{Broadcast} satisfy termination; no example
satisfies never-termination.

Unsurprisingly, all examples demonstrate an increase in verification times and the number of states and transitions when comparing unreliable to reliable versions. Even \exampleName{Adder}, which represents minimal crash-handling, demonstrates relevant increases to the number of states and transitions: this is a direct consequence of the unreliable roles, and the resulting generation of crash and crash-detection transitions in the LTS generated by mCRL2. Verification times also increase because the verified properties follow crash and communication actions, thus requiring the exploration of a larger state space compared to the fully reliable versions.

Nevertheless, our verification times do not increase as quickly as the state space grows,
and are always under 100 ms.  This is
because our $\mu$-calculus furmul\ae\xspace only follow communication, crash, and crash detection transitions, and thus, their verification may not need to follow every possible transition into every state.
%
This suggests greater scalability of the approach that would otherwise be suggested by the size of the state space. This also lends greater motivation to the use of model checkers, as it is infeasible to manually determine the properties of a large LTS with complex crash-handling behaviour.

%% file: proofs/subtyping-properties.tex
\section{Subtyping Properties}
\label{sec:proofs:subtyping-properties}

\begin{restatable}{lemma}{lemStenvReductionSubSafe}
    \label{lem:stenv-safe-reduction-sub}%
    Assume that\, $\stEnv$ is $(\mpS;\rolesR)$-safe and
    $\stEnv \stSub \stEnvi \stEnvMoveGenAnnot \stEnvii$ %
    with:
    \[
      \stEnvAnnotGenericSym \in \setcomp{
        \ltsSendRecv{\mpS}{\roleQ}{\roleR}{\stLab},%
        \ltsCrDe{\mpS}{\roleQ}{\roleR},\, \ltsCrash{\mpS}{\roleP}
      }{\,\roleQ,\!\roleR \!\in\! \roleSet,\, \roleP \!\in\! \roleSet \!\setminus\! \rolesR}
    \]
    \;Then, there is\, $\stEnviii$ such that\; %
    $\stEnv \stEnvMoveGenAnnot \stEnviii \stSub \stEnvii$.%
  \end{restatable}
\begin{proof}
    Similar to \cite{POPL19LessIsMore},
    except that we have three more cases to consider for the transition
    $\stEnvi \stEnvMoveGenAnnot \stEnvii$.
    \begin{itemize}
      \item $\stEnvi \stEnvMoveAnnot{\ltsCrash{\mpS}{\roleP}} \stEnvii$ with $\roleP \not\in \rolesR$.\quad
        This means $\stEnvApp{\stEnvi}{\mpChanRole{\mpS}{\roleP}} \neq \stStop$,
        and thus, $\stEnvApp{\stEnv}{\mpChanRole{\mpS}{\roleP}} \neq \stStop$
        (by subtyping);
        moreover, $\stEnvii = \stEnvi\subst{\mpChanRole{\mpS}{\roleP}}{\stStop}$
        (by rule \inferrule{\iruleTCtxCrash} in \Cref{def:mpst-env-reduction}).
        Therefore, we conclude by taking
        $\stEnviii = \stEnv\subst{\mpChanRole{\mpS}{\roleP}}{\stStop}$, which
        implies $\stEnv \stEnvMoveAnnot{\ltsCrash{\mpS}{\roleP}} \stEnviii$
        and $\stEnviii \stSub \stEnvii$, which is the thesis.

      \item $\stEnvi \stEnvMoveAnnot{\ltsCrDe{\mpS}{\roleP}{\roleQ}} \stEnvii$.
        This means $\stEnvApp{\stEnvi}{\mpChanRole{\mpS}{\roleQ}} = \stEnvApp{\stEnvii}{\mpChanRole{\mpS}{\roleQ}} = \stStop$,
        and $\stEnvApp{\stEnvi}{\mpChanRole{\mpS}{\roleP}} = \stExtSum{\roleQ}{j \in J}{\stChoice{\stLab[j]}{\stSi[j]} \stSeq \stTi[j]}$.
        By subtyping, we also have
        $\stEnvApp{\stEnv}{\mpChanRole{\mpS}{\roleQ}} = \stStop$, and
        $\stEnvApp{\stEnv}{\mpChanRole{\mpS}{\roleP}} = \stExtSum{\roleQ}{i \in I}{\stChoice{\stLab[i]}{\stS[i]} \stSeq \stT[i]}$ with $I \subseteq J$
        and $\forall i \in I: \stS[i] \stSub \stSi[i]$ and $\stT[i] \stSub \stTi[i]$.
        Since $\stEnv$ is $(\mpS;\rolesR)$-safe by hypothesis,
        by clause \inferrule{\iruleSafeCrash} of \Cref{def:mpst-env-safe} we know that $\exists k \in I: \stLab[k] = \stCrashLab$
        --- which means that we also have\;
        $\stEnvii = \stEnvi\subst{\mpChanRole{\mpS}{\roleP}}{\stTi[k]}$ \;(since\; $\stEnvi \stEnvMoveAnnot{\ltsCrDe{\mpS}{\roleP}{\roleQ}} \stEnvii$ \;and\; $k \in I \subseteq J$). Therefore, we conclude by taking\;
        $\stEnviii = \stEnv\subst{\mpChanRole{\mpS}{\roleQ}}{\stT[k]}$,
        \;and we obtain\;
        $\stEnv \stEnvMoveAnnot{\ltsCrDe{\mpS}{\roleP}{\roleQ}} \stEnviii$
        \;and\;
        $\stEnviii \stSub \stEnvii$,
        \;which is the thesis.

      \item $\stEnvi \stEnvMoveAnnot{\ltsSendRecv{\mpS}{\roleP}{\roleQ}{\stLab}} \stEnvii$,
        with $\stEnvApp{\stEnvi}{\mpChanRole{\mpS}{\roleQ}} = \stEnvApp{\stEnvii}{\mpChanRole{\mpS}{\roleQ}} = \stStop$.
        This means $\stEnvApp{\stEnvi}{\mpChanRole{\mpS}{\roleP}} = \stIntSum{\roleQ}{j \in J}{\stChoice{\stLab[j]}{\stSi[j]} \stSeq \stTi[j]}$,
        with $\stLab = \stLab[k]$ for some $k \in J$.
        By subtyping, we also have
        $\stEnvApp{\stEnv}{\mpChanRole{\mpS}{\roleQ}} = \stStop$, and
        $\stEnvApp{\stEnv}{\mpChanRole{\mpS}{\roleP}} = \stIntSum{\roleQ}{i \in I}{\stChoice{\stLab[i]}{\stS[i]} \stSeq \stT[i]}$ with $J \subseteq I$
        and $\forall i \in I: \stS[i] \stSub \stSi[i]$ and $\stT[i] \stSub \stTi[i]$.
        Observe that from $\stEnvi \stEnvMoveAnnot{\ltsSendRecv{\mpS}{\roleP}{\roleQ}{\stLab[k]}} \stEnvii$ we have
        $\stEnvii = \stEnvi\subst{\mpChanRole{\mpS}{\roleP}}{\stTi[k]}$;
        also observe that since $k \in J \subseteq I$,
        we can take $\stEnviii$ such that $\stEnv \stEnvMoveAnnot{\ltsSendRecv{\mpS}{\roleP}{\roleQ}{\stLab[k]}} \stEnviii = \stEnv\subst{\mpChanRole{\mpS}{\roleQ}}{\stT[k]}$,
        \;thus also getting\;
        $\stEnviii \stSub \stEnvii$:
        \;this is the thesis.
    \end{itemize}
\end{proof}

\begin{restatable}{proposition}{lemStenvReductionSubSafeInd}%
  \label{lem:stenv-safe-reduction-sub-ind}%
  Assume that\, $\stEnv$ is $(\mpS;\rolesS)$-safe and
  $\stEnv \stSub \stEnvi \stEnvMoveAnnot{\stEnvAnnotGenericSym[1]}\cdots\stEnvMoveAnnot{\stEnvAnnotGenericSym[n]} \stEnvii$, with:
  \[
      \forall i \in 1..n: \stEnvAnnotGenericSym[i] \in \setcomp{
        \ltsSendRecv{\mpS}{\roleQ}{\roleR}{\stLab[i]}, %
        \ltsCrDe{\mpS}{\roleQ}{\roleR},\, \ltsCrash{\mpS}{\roleP}
      }{\,\roleQ,\!\roleR \!\in\! \roleSet,\, \roleP \!\in\! \roleSet \!\setminus\! \rolesS}
  \]
  Then, there is\, $\stEnviii$ such that\; %
  $\stEnv \stEnvMoveAnnot{\stEnvAnnotGenericSym[1]}\cdots\stEnvMoveAnnot{\stEnvAnnotGenericSym[n]} \stEnviii \stSub \stEnvii$.%
\end{restatable}
\begin{proof}
    By induction on the number of transitions $n$ in $\stEnvi \stEnvMoveAnnot{\stEnvAnnotGenericSym[1]}\cdots\stEnvMoveAnnot{\stEnvAnnotGenericSym[n]} \stEnvii$.
    The base case ($n = 0$ transitions) is immediate: we have $\stEnvi = \stEnvii$,
    hence we conclude by taking $\stEnviii = \stEnv$.
    In the inductive case with $n = m+1$ transitions,
    there is $\stEnvii[0]$ such that $\stEnvi \stEnvMoveAnnot{\stEnvAnnotGenericSym_1}\cdots\stEnvMoveAnnot{\stEnvAnnotGenericSym[m]} \stEnvii[0] \stEnvMoveAnnot{\stEnvAnnotGenericSym[n]} \stEnvii$.
    By the induction hypothesis, there is $\stEnviii[0]$ such that
    $\stEnv \stEnvMoveAnnot{\stEnvAnnotGenericSym[1]}\cdots\stEnvMoveAnnot{\stEnvAnnotGenericSym[m]} \stEnviii[0] \stSub \stEnvii[0]$. %
    Hence, by \Cref{lem:stenv-safe-reduction-sub}, there exists $\stEnviii$
    such that $\stEnviii[0] \stEnvMoveAnnot{\stEnvAnnotGenericSym[n]} \stEnviii$
    and $\stEnviii \stSub \stEnvii$. Therefore, we have
    $\stEnv \stEnvMoveAnnot{\stEnvAnnotGenericSym[1]}\cdots\stEnvMoveAnnot{\stEnvAnnotGenericSym[n]} \stEnviii \stSub \stEnvii$,
    which is the thesis.
\end{proof}

\begin{restatable}{lemma}{lemStenvSubSafe}%
  \label{lem:stenv-supertype-safe}%
  If\; $\stEnv$ is $(\mpS;\rolesS)$-safe %
  \;and\; %
  $\stEnv \stSub \stEnvi$, %
  \;then\; $\stEnvi$ is $(\mpS;\rolesS)$-safe.
\end{restatable}
\begin{proof}
    Assume that $\stEnv$ is $(\mpS;\rolesS)$-safe. By contradiction, also assume that $\stEnvi$ is
    \emph{not} $(\mpS;\rolesS)$-safe. This means that there is a series of reductions $\stEnvi \stEnvMoveAnnot{\stEnvAnnotGenericSym[1]}\cdots\stEnvMoveAnnot{\stEnvAnnotGenericSym[n]} \stEnvii$, with:
    \[
        \forall i \in 1..n: \stEnvAnnotGenericSym[i] \in \setcomp{
        \begin{array}{@{}l@{}}
          \ltsSendRecv{\mpS}{\roleQ}{\roleR}{\stLab[i]},\, %
          \ltsCrDe{\mpS}{\roleQ}{\roleR},\, \ltsCrash{\mpS}{\roleP}
        \end{array}}{\,\roleQ,\!\roleR \!\in\! \roleSet,\, \roleP \!\in\! \roleSet \!\setminus\! \rolesS}
    \]
    and with $\stEnvii$ violating clause %
    \inferrule{\iruleSafeComm}
    or \inferrule{\iruleSafeCrash} of \Cref{def:mpst-env-safe}. Now,
    observe that by \Cref{lem:stenv-safe-reduction-sub-ind}, $\stEnv$ can simulate all such reductions of $\stEnvi$, reaching a typing context $\stEnviii \stSub \stEnvii$; by cases on the subtyping, we can easily verify that $\stEnviii$
    violates clause %
    \inferrule{\iruleSafeComm} or \inferrule{\iruleSafeCrash}, similarly to $\stEnvii$.
    But then, we obtain that $\stEnv$ is \emph{not} safe either --- contradiction. %
    Therefore, we conclude that $\stEnvi$ is safe.
\end{proof}

%% file: proofs/type-system-properties.tex
\section{Type System Properties}
\label{sec:proofs:type-system-properties}

\begin{restatable}[Narrowing]{lemma}{lemStenvNarrowingMpstStd}
    \label{lem:narrowing-mpst-std}%
    \label{lem:narrowing}%
    If\; %
    $\stJudge{\mpEnv}{\stEnv}{\mpP}$
    \;and\; %
    $\stEnvi \stSub \stEnv$, %
    \;then\; %
    $\stJudge{\mpEnv}{\stEnvi}{\mpP}$.
  \end{restatable}
  \begin{proof}
    By induction on the derivation of\; $\stJudge{\mpEnv}{\stEnv}{\mpP}$,
    \;we obtain a derivation that concludes\; $\stJudge{\mpEnv}{\stEnvi}{\mpP}$
    \;by inserting (possibly vacuous) instances of rule \inferrule{\iruleMPSub} (\Cref{fig:mpst-rules}).
\end{proof}

\begin{restatable}[Substitution]{lemma}{lemSubstitution}
  \label{lem:substitution}%
  Assume\; %
  $\stJudge{\mpEnv}{%
    \stEnv \stEnvComp%
    \stEnvMap{x}{\stS}%
  }{\mpP}$ %
  \;and\; %
  $\stEnvEntails{\stEnvi}{
    \mpW
  }{\stS}$, %
  \,with\, $\stEnv \stEnvComp \stEnvi$ defined. %
  \;Then,\; %
  $\stJudge{\mpEnv}{%
    \stEnv \stEnvComp \stEnvi%
  }{\mpP\subst{\mpFmt{x}}{
    \mpW
  }}$.%
\end{restatable}
\begin{proof}
  Minor adaptation of \cite[Lemma 5]{Coppo2015GentleIntroMAPST}.
\end{proof}

\begin{restatable}[Subject Congruence]{lemma}{lemSubjectCongruence}
  \label{lem:subject-congruence}%
  Assume\; %
  $\stJudge{\mpEnv}{\stEnv}{\mpP}$ %
  \;and\; %
  $\mpP \equiv \mpPi$. %
  \;Then,\, %
  $\stJudge{\mpEnv}{\stEnv}{\mpPi}$.
\end{restatable}
\begin{proof}
The proof follows as in~\cite{POPL19LessIsMore}. For \inferrule{\iruleCongStopElim}, which does not appear,
\[
\begin{array}{r@{\;\;}c@{\;\;}l}
\mpP &=& \mpRes{\mpS}{(\mpStop{\mpS}{\roleP[1]} \mpPar \cdots \mpPar \mpStop{\mpS}{\roleP[n]})}
\\
\mpPi &=& \mpNil
\end{array}
\]
\noindent
By inversion of \inferrule{\iruleMPRedRes}, we have $s \notin \stEnv$ and
$\stJudge{\mpEnv}{\stEnv\stEnvComp\stEnvi}{\mpStop{\mpS}{\roleP[1]} \mpPar \cdots \mpPar \mpStop{\mpS}{\roleP[n]}}$.
Then, by inversion of \inferrule{\iruleMPPar} and \inferrule{\iruleMPStop}, we have $\stEnvEndP{\stEnv}$ and $\forall i \in 1..n: \stJudge{\mpEnv}{\stEnv\stEnvComp\stEnvMap{\mpChanRole{\mpS}{\roleP[i]}}{\stStop}}{\mpStop{\mpS}{\roleP[i]}}$.
Therefore, by %
\inferrule{\iruleMPNil}, we conclude $\stJudge{\mpEnv}{\stEnv}{\mpP'}$.
\end{proof}

%% file: proofs/subject-reduction.tex
\section{Proofs for Subject Reduction and Type Safety}
\label{sec:proofs:subject-reduction}

\begin{proposition}
  \label{lem:uncrashed-process-no-crash-type}
  If\; $\stJudge{\mpEnv}{\stEnv}{\mpP}$ \;and $\mpP \not\equiv \mpStop{\mpS}{\roleP} \mpPar \mpR$
  (for all $\mpS,\roleP,\mpR$),
  then $\forall \mpC \in \dom{\stEnv}: \stEnvApp{\stEnv}{\mpC} \neq \stStop$.
\end{proposition}
\begin{proof}
  By easy induction on the derivation of $\stJudge{\mpEnv}{\stEnv}{\mpP}$.
\end{proof}

\begin{proposition}
  \label{lem:typing-fc}
  If\; $\stJudge{\mpEnv}{\stEnv}{\mpP}$, \;then $\fc{\mpP} \subseteq \dom{\stEnv}$
  \;and\; $\forall \mpChanRole{\mpS}{\roleP} \in \dom{\stEnv} \setminus \fc{\mpP}: \stEnvApp{\stEnv}{\mpChanRole{\mpS}{\roleP}} \stSub \stEnd$.
\end{proposition}
\begin{proof}
  By easy induction on the derivation of $\stJudge{\mpEnv}{\stEnv}{\mpP}$.
\end{proof}

\begin{proposition}
  \label{lem:fc-not-end-stop}
  Assume $\stJudge{\mpEnv}{\stEnv}{\mpP}$.
  Then, for all $\mpChanRole{\mpS}{\roleP} \in \fc{\mpP}$,
  we have $\stEnvApp{\stEnv}{\mpChanRole{\mpS}{\roleP}} \stNotSub \stEnd$.
\end{proposition}
\begin{proof}
  By induction on the typing derivation of $\stJudge{\mpEnv}{\stEnv}{\mpP}$,
  using the rules in \Cref{fig:mpst-rules}.
  We develop the two most interesting case (the others are similar and easier).

  \medskip
  \noindent
  Base case \inferrule{\iruleMPCall}.  We have:
  \begin{flalign}
    \label{eq:fc-not-end-stop:call:typing}
    \begin{array}{@{}l@{}}
    \begin{array}{@{}l@{}}
      \mpP = \mpCall{\mpX}{\mpD[1],\ldots,\mpD[n]}
      \\
      \stEnv = \stEnv[0] \stEnvComp \stEnv[1] \stEnvComp \ldots \stEnvComp \stEnv[n]
    \end{array}
    \;\text{ such that}\;
    \\\qquad\qquad\qquad
    \inference[\iruleMPCall]{
        \mpEnvEntails{\mpEnv}{X}{
          \stS[1],\ldots,\stS[n]
        }
        &
        \stEnvEndP{\stEnv[0]}
        &
        \forall i \in 1..n
        &
        \stEnvEntails{\stEnv[i]}{\mpD[i]}{\stS[i]}
        &
        \stS[i] \stNotSub \stEnd
    }{
      \stJudge{\mpEnv}{
        \stEnv[0] \stEnvComp
        \stEnv[1] \stEnvComp \ldots \stEnvComp \stEnv[n]
      }{
        \mpCall{\mpX}{\mpD[1],\ldots,\mpD[n]}
      }
    }
    \end{array}
  \end{flalign}
  Now observe:
  \begin{flalign}
    \label{eq:fc-not-end-stop:call:fc-p}
    &\fc{\mpP} \subseteq \setcomp{\mpD[i]}{i \in 1..n}
    \subseteq \bigcup_{i \in 1..n}\dom{\stEnv[i]}
    &\text{(by \eqref{eq:fc-not-end-stop:call:typing} and \Cref{lem:typing-fc})}
    \\
    \label{eq:fc-not-end-stop:call:di}
    &\forall i \in 1..n: \mpD[i] \in \fc{\mpP} \implies \stEnvApp{\stEnv[i]}{\mpD[i]} \stNotSub \stEnd
    &\text{(by \eqref{eq:fc-not-end-stop:call:typing})}
    \\
    \nonumber
    &\forall \mpChanRole{\mpS}{\roleP} \in \fc{\mpP}:
    \stEnvApp{\stEnv}{\mpChanRole{\mpS}{\roleP}} \stNotSub \stEnd
    &\text{(by \eqref{eq:fc-not-end-stop:call:typing}, \eqref{eq:fc-not-end-stop:call:fc-p}, and \eqref{eq:fc-not-end-stop:call:di})}
  \end{flalign}
  which is the thesis.

  \medskip
  \noindent
  Inductive case \inferrule{\iruleMPSel}.  We have:
  \begin{flalign}
    \label{eq:fc-not-end-stop:sel:typing}
    &
      \begin{array}{@{}l@{}}
      \begin{array}{@{}l@{}}
      \mpP = \mpSel{\mpC}{\roleQ}{\mpLab}{\mpD}{\mpPi}
      \\
      \stEnv = \stEnv[0] \stEnvComp \stEnv[1] \stEnvComp \stEnv[2]
    \end{array}
    \;\text{ such that}\;
    \\\qquad\qquad\qquad
    \inference[\iruleMPSel]{
      \stEnvEntails{\stEnv[1]}{\mpC}{
        \stIntSum{\roleQ}{}{\stChoice{\stLab}{\stS} \stSeq \stT}
      }
      &
      \stEnvEntails{\stEnv[2]}{\mpD}{\tyS}
      &
      \stS \stNotSub \stEnd
      &
      \stJudge{\mpEnv}{
        \stEnv[0] \stEnvComp \stEnvMap{\mpC}{\stT}
      }{
        \mpPi
      }
    }{
      \stJudge{\mpEnv}{
        \stEnv[0] \stEnvComp \stEnv[1] \stEnvComp \stEnv[2]
      }{
        \mpSel{\mpC}{\roleQ}{\mpLab}{\mpD}{\mpPi}
      }
    }
    \end{array}
  \end{flalign}
  Now observe:
  \begin{flalign}
    \label{eq:fc-not-end-stop:sel:fc-p}
    &\fc{\mpP} \subseteq \setenum{\mpC, \mpD} \cup \fc{\mpPi}
    \subseteq \dom{\stEnv[0]} \cup \dom{\stEnv[1]} \cup \dom{\stEnv[2]}
    \hspace{-10mm}
    &\text{(by \eqref{eq:fc-not-end-stop:sel:typing} and \Cref{lem:typing-fc})}
    \\
    \label{eq:fc-not-end-stop:sel:c-d}
    &\stEnvApp{\stEnv[1]}{\mpC} \stNotSub \stEnd \;\text{ and }\;
    \left(\mpD \in \fc{\mpP} \implies \stEnvApp{\stEnv[2]}{\mpD} \stNotSub \stEnd\right)
    &\text{(by \eqref{eq:fc-not-end-stop:sel:typing})}
    \\
    \label{eq:fc-not-end-stop:sel:ih}
    &\forall \mpChanRole{\mpS}{\roleP} \in \fc{\mpPi}: \stEnvApp{(\stEnv[0] \stEnvComp \stEnvMap{\mpC}{\stT})}{\mpChanRole{\mpS}{\roleP}} \stNotSub \stEnd
    &\text{(by i.h.)}
    \\
    \label{eq:fc-not-end-stop:sel:stenv-zero}
    &\forall \mpChanRole{\mpS}{\roleP} \in \fc{\mpPi} \setminus \setenum{\mpC}: \stEnvApp{\stEnv[0]}{\mpChanRole{\mpS}{\roleP}} \stNotSub \stEnd
    &\text{(by \eqref{eq:fc-not-end-stop:sel:ih})}
    \\
    \nonumber
    &\forall \mpChanRole{\mpS}{\roleP} \in \fc{\mpP}:
    \stEnvApp{\stEnv}{\mpChanRole{\mpS}{\roleP}} \stNotSub \stEnd
    &\text{(by \eqref{eq:fc-not-end-stop:sel:typing},
      \eqref{eq:fc-not-end-stop:sel:fc-p},
      \eqref{eq:fc-not-end-stop:sel:stenv-zero},
      and \eqref{eq:fc-not-end-stop:sel:c-d})}
  \end{flalign}
  which is the thesis.
\end{proof}

\lemSubjectReduction*
\begin{proof}
  Let us recap the assumptions:
  \begin{align}
    \label{item:subjred:typed-unused}
    &\stJudge{\mpEnv}{\stEnv}{\mpP}
    \\
    \label{item:subjred:stenv-safe}
    &\forall \mpS \in \stEnv: \exists
    \rolesR[\mpS]:\stEnvSafeSessRolesSP{\mpS}{\rolesR[\mpS]}{\stEnv}
    \\
    \label{item:subjred:no-reliable-crash}
    &\mpP \mpMoveMaybeCrashChecked \mpPi
  \end{align}

  The proof proceeds by induction of the derivation of %
  $\mpP \mpMoveMaybeCrashChecked \mpPi$, %
  and when the reduction holds by rule $\inferrule{\iruleMPRedCtx}$, %
  with a further structural induction on the reduction context $\mpCtx$. %
  Most cases hold %
  by inversion of the typing $\stJudge{\mpEnv}{\stEnv}{\mpP}$, %
  and by applying the induction hypothesis. %

  \noindent
  Case \inferrule{\iruleMPRedComm}:
  \begin{flalign}
    \label{eq:subj-red:comm:p-pi}%
    &\begin{array}{rcl}%
      \textstyle%
      \mpP &=&%
      \mpBranch{\mpChanRole{\mpS}{\roleP}}{\roleQ}{i \in I}{%
        \mpLab[i]}{x_i}{\mpP[i]}{}
      \,\mpPar\,%
      \mpSel{\mpChanRole{\mpS}{\roleQ}}{\roleP}{\mpLab[k]}{%
        \mpW
      }{\mpQ}%
      \\[1mm]%
      \mpPi &=&%
      \mpP[k]\subst{\mpFmt{x_k}}{
        \mpW
      }%
      \,\mpPar\,%
      \mpQ%
      \quad%
      (k \in I)%
    \end{array}
    &\text{%
      (by inversion of \inferrule{\iruleMPRedComm})%
    }%
    \\
    \label{eq:subj-red:comm:p-typing}%
    &
    \begin{array}{l}
    \stEnv = \stEnv[\stExtC] \stEnvComp \stEnv[\stIntC]
    \quad\text{s.t.}\quad%
    \\\qquad
    \inference[\iruleMPPar]{%
      \begin{array}{l}
        \stJudge{\mpEnv}{%
          \stEnv[\stExtC]%
        }{%
          \mpBranch{\mpChanRole{\mpS}{\roleP}}{\roleQ}{i \in I}{%
            \mpLab[i]}{x_i}{\mpP[i]}{}
        }%
        \\%
        \stJudge{\mpEnv}{%
          \stEnv[\stIntC]%
        }{%
          \mpSel{\mpChanRole{\mpS}{\roleQ}}{\roleP}{\mpLab[k]}{%
            \mpW
          }{\mpQ}
        }%
      \end{array}
    }{%
      \stJudge{\mpEnv}{%
        \stEnv%
      }{%
        \mpP%
      }%
    }%
    \end{array}
    &\text{%
      (by \eqref{eq:subj-red:comm:p-pi} %
      and inv.~of \inferrule{\iruleMPPar})%
    }%
    \\[1mm]%
    &\label{eq:subj-red:comm:p-branch-typing}
    \begin{array}{@{}l@{}}
    \stEnv[\stExtC] = \stEnv[0] \stEnvComp \stEnv[1]%
    \quad\text{s.t.}\quad%
    \\\qquad
    \inference[\iruleMPBranch]{%
      \begin{array}{l}
        \stEnvEntails{\stEnv[1]}{\mpChanRole{\mpS}{\roleP}}{%
          \stExtSum{\roleQ}{i \in I}{\stChoice{\stLab[i]}{\stS[i]} \stSeq \stT[i]}%
        }%
        \\%
        \forall i \in I%
        \quad%
        \stJudge{\mpEnv}{%
          \stEnv[0] \stEnvComp%
          \stEnvMap{x_i}{\stS[i]} \stEnvComp%
          \stEnvMap{\mpChanRole{\mpS}{\roleP}}{\stT[i]}%
        }{%
          \mpP[i]%
        }%
      \end{array}
    }{%
      \stJudge{\mpEnv}{%
        \stEnv[\stExtC]%
      }{%
        \mpBranch{\mpChanRole{\mpS}{\roleP}}{\roleQ}{i \in I}{\mpLab[i]}{x_i}{\mpP[i]}{}
      }%
    }%
    \hspace{-0mm}%
    \end{array}
    &\text{%
      (by \eqref{eq:subj-red:comm:p-typing} %
      and inv.~of \inferrule{\iruleMPBranch})%
    }%
    \\[2mm]%
    &\label{eq:subj-red:comm:p-sel-typing}%
    \begin{array}{@{}l@{}}
    \stEnv[\stIntC] = \stEnv[2] \stEnvComp \stEnv[3] \stEnvComp \stEnv[4]%
    \;\;\text{s.t.}\;\;%
    \\\quad
    \inference[\iruleMPSel]{%
      \begin{array}{l}
        \stEnvEntails{\stEnv[4]}{\mpChanRole{\mpS}{\roleQ}}{
          \stIntSum{\roleP}{}{\stChoice{\stLab[k]}{\stSi[k]} \stSeq \stTi[k]}%
        }%
        \\%
        \stEnvEntails{\stEnv[3]}{%
          \mpW
        }{\stSi[k]}%
        \quad
        \stSi[k] \stNotSub \stEnd
        \quad%
        \stJudge{\mpEnv}{%
          \stEnv[2] \stEnvComp \stEnvMap{\mpChanRole{\mpS}{\roleQ}}{\stTi[k]}%
        }{%
          \mpQ%
        }%
      \end{array}
    }{%
      \stJudge{\mpEnv}{%
        \stEnv[\stIntC]%
      }{%
        \mpSel{\mpChanRole{\mpS}{\roleQ}}{\roleP}{\mpLab[k]}{
          \mpW
        }{\mpQ}%
      }%
    }%
    \end{array}
    \hspace{-15mm}%
    &\text{%
      (by \eqref{eq:subj-red:comm:p-typing} %
      and inv.~of \inferrule{\iruleMPSel})%
    }%
  \end{flalign}
  Now, notice that:%
  \begin{flalign}
    &\label{eq:subj-red:comm:stenv-composition}%
    \stEnv = \stEnv[0] \stEnvComp \stEnv[1] \stEnvComp \stEnv[2] \stEnvComp%
    \stEnv[3] \stEnvComp \stEnv[4]%
    &\text{%
      (by \eqref{eq:subj-red:comm:p-typing}, %
      \eqref{eq:subj-red:comm:p-branch-typing}, %
      and \eqref{eq:subj-red:comm:p-sel-typing})%
    }%
    \\%
    &\label{eq:subj-red:comm:stenv-i-composition}
    \stEnv[1] = \stEnvMap{\mpChanRole{\mpS}{\roleP}}{\stT}%
    \;\text{with}\;%
    \stT \stSub %
    \stExtSum{\roleQ}{i \in I}{\stChoice{\stLab[i]}{\stS[i]} \stSeq \stT[i]}%
    &\text{%
      (by \eqref{eq:subj-red:comm:p-branch-typing} %
      and \Cref{fig:mpst-rules}, rule \inferrule{\iruleMPSub})%
    }%
    \\%
    &\label{eq:subj-red:comm:stenv-iv-composition}
    \stEnv[4] = \stEnvMap{\mpChanRole{\mpS}{\roleQ}}{\stTi}%
    \;\text{with}\;
    \stTi \stSub%
    \stIntSum{\roleQ}{}{\stChoice{\stLab[k]}{\stSi[k]} \stSeq \stTi[k]}%
    &\text{%
      (by \eqref{eq:subj-red:comm:p-sel-typing} %
      and \Cref{fig:mpst-rules}, rule \inferrule{\iruleMPSub})%
    }%
    \\%
    &\label{eq:subj-red:comm:stenvii-suptype}%
    \begin{array}{l}
    \stEnv \stSub \stEnvii =%
    \stEnv[0] \stEnvComp%
    \stEnvi[1]
    \stEnvComp \stEnv[2] \stEnvComp%
    \stEnv[3] \stEnvComp%
    \stEnvi[4]
    \text{ \;where}
    \\
    \stEnvi[1] = \stEnvMap{\mpChanRole{\mpS}{\roleP}}{%
      \stExtSum{\roleQ}{i \in I}{\stChoice{\stLab[i]}{\stS[i]} \stSeq \stT[i]}%
    }%
    \\
    \stEnvi[4] =
    \stEnvMap{\mpChanRole{\mpS}{\roleQ}}{%
      \stIntSum{\roleQ}{}{\stChoice{\stLab[k]}{\stSi[k]} \stSeq \stTi[k]}%
    }%
    \hspace{-16mm}%
    \end{array}
    &\text{%
      (by \eqref{eq:subj-red:comm:stenv-composition}, %
      \eqref{eq:subj-red:comm:stenv-i-composition}, %
      \eqref{eq:subj-red:comm:stenv-iv-composition},
      and \Cref{def:mpst-env-subtype})%
    }%
    \\%
    &\label{eq:subj-red:comm:stenvii-safe}%
    \forall \mpS \in \stEnv: \stEnvSafeSessRolesSP{\mpS}{\rolesR[\mpS]}{\stEnvii}
    &\text{%
      (by %
      \ref{item:subjred:stenv-safe},
      \eqref{eq:subj-red:comm:stenvii-suptype} %
      and \Cref{lem:stenv-supertype-safe})%
    }%
    \\%
    &\label{eq:subj-red:comm:stenv-i-k-carried-sub}%
    k \in I%
    \quad\text{and}\quad%
    \stSi[k] \stSub \stS[k]%
    &%
    \hspace{-50mm}%
    \text{%
      (by \eqref{eq:subj-red:comm:stenvii-suptype}, %
      \eqref{eq:subj-red:comm:stenvii-safe} %
      and \Cref{def:mpst-env-safe}, %
      clause \inferrule{\iruleSafeComm})%
    }%
    \\%
    &\label{eq:subj-red:comm:stenvii-move-stenviii}%
    \stEnvii \stEnvMove \stEnviii%
    = \stEnv[0] \stEnvComp%
    \stEnvMap{\mpChanRole{\mpS}{\roleP}}{\stT[k]} \stEnvComp%
    \stEnv[2] \stEnvComp \stEnv[3] \stEnvComp%
    \stEnvMap{\mpChanRole{\mpS}{\roleQ}}{\stTi[k]}%
    &\text{%
      (by \eqref{eq:subj-red:comm:stenvii-suptype}, %
      \eqref{eq:subj-red:comm:stenv-i-k-carried-sub} %
      and \Cref{def:mpst-env-reduction})%
    }%
    \\%
    &\label{eq:subj-red:comm:stenviii-safe}%
    \forall \mpS \in \stEnv: \stEnvSafeSessRolesSP{\mpS}{\rolesR[\mpS]}{\stEnviii}%
    &\text{%
      (by \eqref{eq:subj-red:comm:stenvii-safe}, %
      \eqref{eq:subj-red:comm:stenvii-move-stenviii} %
      and \Cref{def:mpst-env-safe}, clause \inferrule{\iruleSafeMove})%
    }%
  \end{flalign}
  We can now use $\stEnviii$ to type $\mpPi$:
  \begin{flalign}
    &\label{eq:subj-red:pi-branch-cont-typing}%
    \stJudge{\mpEnv}{%
      \stEnv[0] \stEnvComp%
      \stEnvMap{x_k}{\stS[k]} \stEnvComp%
      \stEnvMap{\mpChanRole{\mpS}{\roleP}}{\stT[k]}%
    }{%
      \mpP[k]%
    }%
    &\text{%
      (by \eqref{eq:subj-red:comm:stenv-i-k-carried-sub}, %
      \eqref{eq:subj-red:comm:p-sel-typing} %
      and \eqref{eq:subj-red:comm:p-branch-typing})%
    }%
    \\[1mm]%
    &\label{eq:subj-red:comm-payload-entails}%
    \stEnvEntails{\stEnv[3]}{%
      \mpW
    }{\stS[k]}%
    &
      \begin{array}{@{}l@{}}
      \text{%
      (by \eqref{eq:subj-red:comm:p-sel-typing} %
      (for $\stEnvEntails{\stEnv[3]}{%
        \mpW
      }{\stSi[k]}$),}
      \eqref{eq:subj-red:comm:stenv-i-k-carried-sub}, %
      \\\text{
      transitivity of $\stSub$, %
      and \inferrule{\iruleMPSub})%
    }%
    \end{array}
    \\%
    &\label{eq:subj-red:pi-branch-cont-typing-env-subst-def}%
    \stEnv[0] \stEnvComp \stEnv[3] \stEnvComp%
    \stEnvMap{\mpChanRole{\mpS}{\roleP}}{\stT[k]}%
    \text{\; defined}%
    &\text{%
      (by \eqref{eq:subj-red:comm:p-sel-typing},
      \eqref{eq:subj-red:comm:p-branch-typing}, %
      and \eqref{eq:subj-red:comm:p-typing})%
    }%
    \\%
    &\label{eq:subj-red:comm:pi-branch-cont-typing-subst}%
    \stJudge{\mpEnv}{%
      \stEnv[0] \stEnvComp%
      \stEnv[3] \stEnvComp%
      \stEnvMap{\mpChanRole{\mpS}{\roleP}}{\stT[k]}%
    }{%
      \mpP[k]\subst{\mpFmt{x_k}}{
        \mpW
      }%
    }%
    &\text{%
      (by \eqref{eq:subj-red:pi-branch-cont-typing}, %
      \eqref{eq:subj-red:comm-payload-entails}, %
      \eqref{eq:subj-red:pi-branch-cont-typing-env-subst-def}, %
      and \Cref{lem:substitution})%
    }%
    \\[1mm]%
    &\label{eq:subj-red:comm:pi-branch-cont-typing-subst-stenviii}%
    \inference[\iruleMPPar]{%
      \begin{array}{l}
      \stJudge{\mpEnv}{%
        \stEnv[0] \stEnvComp%
        \stEnv[3] \stEnvComp%
        \stEnvMap{\mpChanRole{\mpS}{\roleP}}{\stT[k]}%
      }{%
        \mpP[k]\subst{\mpFmt{x_k}}{
          \mpW
        }%
      }%
      \\
      \stJudge{\mpEnv}{%
        \stEnv[2] \stEnvComp \stEnvMap{\mpChanRole{\mpS}{\roleQ}}{\stTi[k]}%
      }{%
        \mpQ%
      }%
      \end{array}
    }{%
      \stJudge{\mpEnv}{%
        \stEnviii%
      }{%
        \mpPi%
      }%
    }%
    &%
     \text{%
      (by \eqref{eq:subj-red:comm:pi-branch-cont-typing-subst}, %
      \eqref{eq:subj-red:comm:p-sel-typing}, %
      \eqref{eq:subj-red:comm:stenvii-move-stenviii}, %
      \eqref{eq:subj-red:comm:stenviii-safe} %
      and \eqref{eq:subj-red:comm:p-pi})%
    }%
  \end{flalign}
  We conclude this case by showing that there exists
  some $\stEnvi$ that satisfies the statement:
  \begin{flalign}
    &\label{eq:subj-red:comm:exists-stenvi}%
    \exists \stEnvi: %
    \stEnv \stEnvMove \stEnvi \stSub \stEnviii%
    &\text{%
      (by \eqref{eq:subj-red:comm:stenvii-suptype}, %
      \eqref{eq:subj-red:comm:stenvii-move-stenviii}, %
      and \Cref{lem:stenv-safe-reduction-sub})%
    }%
    \\%
    &\label{eq:subj-red:comm:stenvi-safe}%
    \forall \mpS \in \stEnvi: \stEnvSafeSessRolesSP{\mpS}{\rolesR[\mpS]}{\stEnvi}
    &\text{%
      (by \eqref{eq:subj-red:comm:exists-stenvi} %
      and \Cref{def:mpst-env-safe}, clause \inferrule{\iruleSafeMove})%
    }%
    \\%
    &\nonumber%
    \stJudge{\mpEnv}{%
      \stEnvi%
    }{%
      \mpPi%
    }%
    &\text{%
      (by \eqref{eq:subj-red:comm:pi-branch-cont-typing-subst-stenviii}, %
      \eqref{eq:subj-red:comm:exists-stenvi}, %
      and \Cref{lem:narrowing})%
    }%
  \end{flalign}

\noindent
Case \inferrule{\iruleMPRedCommE}:
\begin{flalign}
\label{eq:subj-red:comm-ii:p-pi}
&
\begin{array}{rcl}
\mpP &\equiv& \mpStop{\mpS}{\roleP}
           \,\mpPar\,
           \mpSel{\mpChanRole{\mpS}{\roleQ}}{\roleP}{\mpLab}{
             \mpChanRole{\mpSi}{\roleR}
           }{\mpQ}
\\
\mpPi &\equiv& \mpStop{\mpS}{\roleP}
            \,\mpPar\,
            \mpStop{\mpSi}{\roleR}
            \,\mpPar\,
            \mpQ
\end{array}
&\text{
  (by inversion of \inferrule{\iruleMPRedCommE})%
}
\\
\label{eq:subj-red:comm2:p-typing}%
&
\stEnv = \stEnv[\stStopSym],\stEnv[\stIntC]
  \;\;\text{s.t.}\;\;
  \inference[\iruleMPPar]{%
    \begin{array}{l}
      \stJudge{\mpEnv}{
        \stEnv[\stStopSym]%
      }{%
        \mpStop{\mpS}{\roleP}
      }%
      \\%
      \stJudge{\mpEnv}{%
        \stEnv[\stIntC]%
      }{%
        \mpSel{\mpChanRole{\mpS}{\roleQ}}{\roleP}{\mpLab}{%
          \mpChanRole{\mpSi}{\roleR}%
        }{\mpQ}
      }%
    \end{array}
  }{%
    \stJudge{\mpEnv}{%
      \stEnv%
    }{%
      \mpP%
    }%
  }%
&\text{
  (by \eqref{eq:subj-red:comm-ii:p-pi} and inv.\ of
  \inferrule{\iruleMPRedPar})%
}
\\
\label{eq:subj-red:comm2:p-crash-typing}
&\stEnv[\stStopSym] =
    \stEnv[0]\stEnvComp\stEnvMap{\mpChanRole{\mpS}{\roleP}}{\stStop}
  \;\;\text{s.t.}\;\;
  \inference[\iruleMPStop]{
    \stEnvEndP{\stEnv[0]}
  }{
    \stJudge{\mpEnv}{
      \stEnv[0]\stEnvComp\stEnvMap{\mpChanRole{\mpS}{\roleP}}{\stStop}
    }{
      \mpStop{\mpS}{\roleP}
    }
  }
&\text{
  (by \eqref{eq:subj-red:comm2:p-typing} and inv.\ of \inferrule{\iruleMPStop})
}
\\[2mm]
&\label{eq:subj-red:comm2:p-sel-typing}%
\begin{array}{l}
\stEnv[\stIntC] = \stEnv[1] \stEnvComp \stEnv[2] \stEnvComp \stEnv[3]%
\;\;\text{s.t.}\;\;%
\\\qquad
\inference[\iruleMPSel]{%
  \begin{array}{l}
    \stEnvEntails{\stEnv[3]}{\mpChanRole{\mpS}{\roleQ}}{
      \stIntSum{\roleP}{}{\stChoice{\stLab}{\stS} \stSeq \stT}%
    }%
    \\%
    \stEnvEntails{\stEnv[2]}{\mpChanRole{\mpSi}{\roleR}}{\stS}%
    \\%
    \stS \stNotSub \stEnd
    \\
    \stJudge{\mpEnv}{%
      \stEnv[1] \stEnvComp \stEnvMap{\mpChanRole{\mpS}{\roleQ}}{\stT}%
    }{%
      \mpQ%
    }%
  \end{array}
}{%
  \stJudge{\mpEnv}{%
    \stEnv[\stIntC]%
  }{%
    \mpSel{\mpChanRole{\mpS}{\roleQ}}{\roleP}{\mpLab[k]}{\mpChanRole{\mpSi}
{\roleR}}{\mpQ}%
  }%
}%
\hspace{-5mm}%
\end{array}
&\text{%
  (by \eqref{eq:subj-red:comm2:p-typing} %
  and inv.~of \inferrule{\iruleMPSel})%
}%
\end{flalign}

  Now, notice that:%
  \begin{flalign}
    &\label{eq:subj-red:comm2:stenv-composition}%
    \stEnv = \stEnv[0]
      \stEnvComp \stEnvMap{\mpChanRole{\mpS}{\roleP}}{\stStop}
      \stEnvComp \stEnv[1]
      \stEnvComp \stEnv[2]
      \stEnvComp \stEnv[3]
    &\text{
      (by \eqref{eq:subj-red:comm2:p-typing},
      \eqref{eq:subj-red:comm2:p-crash-typing},
      and \eqref{eq:subj-red:comm2:p-sel-typing})%
    }%
    \\%
    &\label{eq:subj-red:comm2:stenv-iv-composition}
    \stEnv[3] = \stEnvMap{\mpChanRole{\mpS}{\roleQ}}{\stTi}%
    \quad\text{with}\quad%
    \stTi \stSub%
    \stIntSum{\roleP}{}{\stChoice{\stLab}{\stS} \stSeq \stT}%
    &\text{%
      (by \eqref{eq:subj-red:comm:p-sel-typing} %
      and \Cref{fig:mpst-rules}, rule \inferrule{\iruleMPSub})%
    }%
    \\%
    &\label{eq:subj-red:comm2:stenvii-suptype}%
    \begin{array}{@{}l@{}}
      \stEnv \stSub \stEnvii \text{ \;where}\\
      \qquad\stEnvii = \stEnv[0] \stEnvComp%
      \stEnvMap{\mpChanRole{\mpS}{\roleP}}{%
        \stStop
      }%
      \stEnvComp \stEnv[1] \stEnvComp%
      \stEnv[2] \stEnvComp%
      \stEnvMap{\mpChanRole{\mpS}{\roleQ}}{%
        \stIntSum{\roleQ}{}{\stChoice{\stLab}{\stS} \stSeq \stT}%
      }%
    \end{array}
    \hspace{-10mm}
    &
    \text{%
          (by \eqref{eq:subj-red:comm2:stenv-composition}, %
          \eqref{eq:subj-red:comm2:stenv-iv-composition},
          and \Cref{def:mpst-env-subtype})%
        }%
    \\%
    &\label{eq:subj-red:comm2:stenvii-safe}%
    \forall \mpS \in \stEnv: \stEnvSafeSessRolesSP{\mpS}{\rolesR[\mpS]}{\stEnv}
    &\text{%
      (by \ref{item:subjred:stenv-safe}, %
      \eqref{eq:subj-red:comm2:stenvii-suptype} %
      and \Cref{lem:stenv-supertype-safe})%
    }%
    \\%
    &\label{eq:subj-red:comm2:stenvii-sr-not-stop}%
    \stEnvApp{\stEnv[2]}{\mpChanRole{\mpSi}{\roleR}} \stNotSub \stStop
    &
      \begin{array}{@{}r@{}}
      \text{%
      (by \eqref{eq:subj-red:comm2:p-sel-typing}
      (where $\stS$ cannot }
      \\\text{be $\stStop$)
      and \Cref{def:subtyping})%
    }
      \end{array}
    \\
    &\label{eq:subj-red:comm2:stenvii-sr-stop}%
    \stEnv[2] \stEnvMoveAnnot{\ltsCrash{\mpSi}{\roleR}} \stEnvMap{\mpChanRole{\mpSi}{\roleR}}{\stStop}
    &
      \begin{array}{@{}r@{}}
      \text{%
      (by \eqref{eq:subj-red:comm2:p-sel-typing}
      (for $\stEnvApp{\stEnv[2]}{\mpChanRole{\mpSi}{\roleR}} \stNotSub \stEnd$),}
      \\
      \text{\eqref{eq:subj-red:comm2:stenvii-sr-not-stop},
      and \Cref{def:mpst-env-reduction} rule \inferrule{\iruleTCtxCrash})%
    }
      \end{array}
    \\
    &\label{eq:subj-red:comm2:stenvii-move-stenviii}%
    \begin{array}{@{}l@{}}
      \stEnvii \stEnvMoveAnnot{\ltsSendRecv{\mpS}{\roleQ}{\roleP}{\stLab}} \stEnvMoveAnnot{\ltsCrash{\mpSi}{\roleR}} \stEnviii%
      \text{ \;where}
      \\
      \qquad\stEnviii%
      = \stEnv[0] \stEnvComp
      \stEnvMap{\mpChanRole{\mpS}{\roleP}}{\stStop} \stEnvComp
      \stEnv[1] \stEnvComp
      \stEnvMap{\mpChanRole{\mpSi}{\roleR}}{\stStop} \stEnvComp%
      \stEnvMap{\mpChanRole{\mpS}{\roleQ}}{\stT}%
    \end{array}
    &\text{%
      (by \eqref{eq:subj-red:comm2:stenvii-suptype}, %
      \eqref{eq:subj-red:comm2:stenvii-sr-stop},
      \Cref{def:mpst-env-reduction} rule \inferrule{\iruleTCtxSendToCrashed})%
    }%
    \\
    &\label{subj-red:comm2:si-r-unreliable}
    \roleR \not\in \rolesR[\mpSi]
    &\text{
      (by \eqref{eq:subj-red:comm-ii:p-pi},
      \ref{item:subjred:no-reliable-crash}
      and \Cref{def:assumption-abiding-reduction})%
    }
    \\%
    &\label{eq:subj-red:comm2:stenviii-safe}%
    \forall \mpS \in \stEnv: \stEnvSafeSessRolesSP{\mpS}{\rolesR[\mpS]}{\stEnviii}%
    &
      \begin{array}{@{}r@{}}
      \text{%
      (by \eqref{eq:subj-red:comm2:stenvii-safe}, %
      \eqref{eq:subj-red:comm2:stenvii-move-stenviii}, }
      \\
      \text{
      \eqref{subj-red:comm2:si-r-unreliable}
      and \Cref{def:mpst-env-safe}, clause \inferrule{\iruleSafeMove})%
    }%
      \end{array}
  \end{flalign}

  We can now use $\stEnviii$ to type $\mpPi$:
  \begin{flalign}
    &\label{eq:subj-red:comm2:pi-branch-cont-typing-subst-stenviii}%
  \begin{array}{@{}l@{}}
  \inference[\iruleMPPar]{
    \stJudge{\mpEnv}{
      \stEnv[0]
      \stEnvComp \stEnvMap{\mpChanRole{\mpS}{\roleP}}{\stStop}
    }{
      \mpStop{\mpS}{\roleP}
    }
    &
    \inference[\iruleMPPar]{
      \stJudge{\mpEnv}{
        \stEnvMap{\mpChanRole{\mpSi}{\roleR}}{\stStop}
      }{
        \mpStop{\mpSi}{\roleR}
      }
      \\[2mm]
      \stJudge{\mpEnv}{
        \stEnv[1]
        \stEnvComp \stEnvMap{\mpChanRole{\mpS}{\roleQ}}{\stT}
      }{
        \mpQ
     }
    }{
      \stJudge{\mpEnv}{
        \stEnvMap{\mpChanRole{\mpSi}{\roleR}}{\stStop}
        \stEnvComp
        \stEnv[1]
        \stEnvComp \stEnvMap{\mpChanRole{\mpS}{\roleQ}}{\stT}
      }{
        \mpStop{\mpSi}{\roleR} \mpPar \mpQ
      }
    }
  }{
    \stJudge{\mpEnv}{
        \stEnviii
      }{
        \mpPi
      }
  }
  \\
    \qquad\qquad\qquad\qquad\qquad\qquad\qquad\qquad\qquad
    \quad
    \text{%
      (by
      \eqref{eq:subj-red:comm2:p-sel-typing}, %
      \eqref{eq:subj-red:comm2:stenvii-move-stenviii}, %
      \eqref{eq:subj-red:comm2:stenviii-safe},
      \eqref{eq:subj-red:comm2:p-crash-typing},
      and \eqref{eq:subj-red:comm-ii:p-pi})%
    }%
  \end{array}
  \end{flalign}

  We conclude this case by showing that there exists
  some $\stEnvi$ that satisfies the statement:
  \begin{flalign}
    &\label{eq:subj-red:comm2:exists-stenvi}%
    \exists \stEnvi: %
    \stEnv \stEnvMoveMaybeCrashStar \stEnvi \stSub \stEnviii%
    &\text{%
      (by \eqref{eq:subj-red:comm2:stenvii-suptype}, %
      \eqref{eq:subj-red:comm2:stenvii-move-stenviii}, %
      \eqref{subj-red:comm2:si-r-unreliable}
      and \Cref{lem:stenv-safe-reduction-sub-ind})%
    }%
    \\%
    &\label{eq:subj-red:comm2:stenvi-safe}%
    \forall \mpS \in \stEnvi: \stEnvSafeSessRolesSP{\mpS}{\rolesR[\mpS]}{\stEnvi}%
    &\text{%
      (by \ref{item:subjred:stenv-safe},
      \eqref{eq:subj-red:comm2:exists-stenvi} %
      and \Cref{def:mpst-env-safe}, clause \inferrule{\iruleSafeMove})%
    }%
    \\%
    &\nonumber%
    \stJudge{\mpEnv}{%
      \stEnvi%
    }{%
      \mpPi%
    }%
    &\text{%
      (by \eqref{eq:subj-red:comm2:pi-branch-cont-typing-subst-stenviii}, %
      \eqref{eq:subj-red:comm:exists-stenvi}, %
      and \Cref{lem:narrowing})%
    }%
  \end{flalign}

\noindent
Case \inferrule{\iruleMPRedCommEGround}:
\begin{flalign}
  \nonumber%
  &
  \begin{array}{rcl}
  \mpP &\equiv& \mpStop{\mpS}{\roleP}
             \,\mpPar\,
             \mpSel{\mpChanRole{\mpS}{\roleQ}}{\roleP}{\mpLab}{
               \mpV
             }{\mpQ}
  \\
  \mpPi &\equiv& \mpStop{\mpS}{\roleP}
              \,\mpPar\,
              \mpQ
  \end{array}
  &\text{
    (by inversion of \inferrule{\iruleMPRedCommEGround})
  }
\end{flalign}
The proof is similar to case \inferrule{\iruleMPRedCommE} above, but simpler:
since a basic value $v$ is being sent to a crashed endpoint $\mpChanRole{\mpS}{\roleP}$,
we have that $\mpPi$ does not contain a new crashed session endpoint $\mpChanRole{\mpSi}{\roleR}$,
and the typing context $\stEnv[2]$ (which types the message payload $\mpV$) is empty
(by rule \inferrule{\iruleMPGround} in \Cref{fig:mpst-rules}).
Consequently, we can adapt the proof by omitting the crashed endpoint $\mpChanRole{\mpSi}{\roleR}$,
skipping step \eqref{eq:subj-red:comm2:stenvii-sr-stop},
and adjusting step \eqref{eq:subj-red:comm2:stenvii-move-stenviii} to have
$\stEnvii \stEnvMoveAnnot{\ltsSendRecv{\mpS}{\roleQ}{\roleP}{\stLab}} \stEnviii$.

\medskip

\noindent
Case \inferrule{\iruleMPRedCommD}:
\begin{flalign}
\label{eq:subj-red:comm3:p-pi}
&
\begin{array}{rcl}
\mpP &=& \mpBranch{\mpChanRole{\mpS}{\roleP}}{\roleQ}{i \in I}{\mpLab[i]}{x_i}{\mpP[i]}{\mpPii}
        \,\mpPar\,
        \mpStop{\mpS}{\roleQ}
\\
\mpPi &=& \mpPii
          \,\mpPar\,
          \mpStop{\mpS}{\roleQ}
\end{array}
&\text{
  (by inversion of \inferrule{\iruleMPRedCommD})%
}
\\
\label{eq:subj-red:comm3:p-typing}%
&
\stEnv = \stEnv[\stExtC],\stEnv[\stStopSym]
  \;\;\text{s.t.}\;\;
  \inference[\iruleMPPar]{%
    \begin{array}{l}
      \stJudge{\mpEnv}{%
        \stEnv[\stExtC]%
      }{%
        \mpBranch{\mpChanRole{\mpS}{\roleP}}{\roleQ}{i \in I}{\mpLab[i]}{x_i}{\mpP[i]}{\mpPii}
      }%
      \\%
      \stJudge{\mpEnv}{
        \stEnv[\stStopSym]%
      }{%
        \mpStop{\mpS}{\roleQ}
      }%
    \end{array}
  }{%
    \stJudge{\mpEnv}{%
      \stEnv%
    }{%
      \mpP%
    }%
  }%
\hspace{-5mm}
&\text{
  (by \eqref{eq:subj-red:comm3:p-pi} and inv.\ of \inferrule{\iruleMPRedPar})%
}
\\[2mm]
&\label{eq:subj-red:comm3:p-branch-typing}%
\stEnv[\stExtC] = \stEnv[0] \stEnvComp \stEnv[1] %
\;\;\text{s.t.}\;\;
\inference[\iruleMPBranch]{%
  \begin{array}{l}
    \stEnvEntails{\stEnv[1]}{\mpChanRole{\mpS}{\roleP}}{%
      \stExtSumErr{\roleQ}{i \in I}{\stChoice{\stLab[i]}{\stS[i]} \stSeq \stT[i]}{\stT}
    }%
    \\%
    \forall i \in I%
    \quad%
    \stJudge{\mpEnv}{%
      \stEnv[0] \stEnvComp%
      \stEnvMap{x_i}{\stS[i]} \stEnvComp%
      \stEnvMap{\mpChanRole{\mpS}{\roleP}}{\stT[i]}%
    }{%
      \mpP[i]%
    }%
    \\
    \stJudge{\mpEnv}{
      \stEnv[0] \stEnvComp \stEnvMap{\mpChanRole{\mpS}{\roleP}
    }{
      \stT
    }}{\mpPii}
  \end{array}
}{%
  \stJudge{\mpEnv}{%
    \stEnv[\stExtC]%
  }{%
    \mpBranch{\mpChanRole{\mpS}{\roleP}}{\roleQ}{i \in I}{\mpLab[i]}{x_i}{\mpP[i]}{\mpPii}
  }%
}%
\hspace{-5mm}%
&\text{%
  (by \eqref{eq:subj-red:comm3:p-typing} %
  and inv.~of \inferrule{\iruleMPBranch})%
}%
\\
\label{eq:subj-red:comm3:p-crash-typing}
&\stEnv[\stStopSym] =
    \stEnv[2]\stEnvComp\stEnvMap{\mpChanRole{\mpS}{\roleQ}}{\stStop}
  \;\;\text{s.t.}\;\;
  \inference[\iruleMPStop]{
    \stEnvEndP{\stEnv[2]}
  }{
    \stJudge{\mpEnv}{
      \stEnv[2]\stEnvComp\stEnvMap{\mpChanRole{\mpS}{\roleQ}}{\stStop}
    }{
      \mpStop{\mpS}{\roleQ}
    }
  }
&\text{
  (by \eqref{eq:subj-red:comm3:p-typing} and inv.\ of
  \inferrule{\iruleMPStop})%
}
\end{flalign}
Now, notice that:
\begin{flalign}
  &\label{eq:subj-red:comm3:stenv-composition}%
  \stEnv = \stEnv[0]
    \stEnvComp \stEnv[1]
    \stEnvComp \stEnv[2]
    \stEnvComp \stEnvMap{\mpChanRole{\mpS}{\roleQ}}{\stStop}
  &\text{%
    (by \eqref{eq:subj-red:comm3:p-typing},
    \eqref{eq:subj-red:comm3:p-branch-typing},
    and \eqref{eq:subj-red:comm3:p-crash-typing})%
  }%
  \\
  &\label{eq:subj-red:comm3:stenv-i-composition}
  \stEnv[1] = \stEnvMap{\mpChanRole{\mpS}{\roleP}}{\stTi}%
  \quad\text{with}\quad%
  \stTi \stSub %
  \stExtSumErr{\roleQ}{i \in I}{\stChoice{\stLab[i]}{\stS[i]} \stSeq \stT[i]}{\stT}
  &\text{%
    (by \eqref{eq:subj-red:comm3:p-branch-typing} %
    and \Cref{fig:mpst-rules}, rule \inferrule{\iruleMPSub})%
  }%
  \\%
  &\label{eq:subj-red:comm3:stenvii-suptype}%
  \begin{array}{@{}l@{}}
    \stEnv \stSub \stEnvii \text{ \;where} \\
    \qquad \stEnvii = \stEnv[0] \stEnvComp%
    \stEnvMap{\mpChanRole{\mpS}{\roleP}}{%
      \stExtSumErr{\roleQ}{i \in I}{\stChoice{\stLab[i]}{\stS[i]} \stSeq \stT[i]}{\stT}
    }%
    \stEnvComp \stEnv[2] \stEnvComp
    \stEnvMap{\mpChanRole{\mpS}{\roleQ}}{%
      \stStop
    }%
  \end{array}
  &\text{%
    (by \eqref{eq:subj-red:comm3:stenv-composition}, %
    \eqref{eq:subj-red:comm3:stenv-i-composition}, %
    and \Cref{def:mpst-env-subtype})%
  }%
  \\%
  &\label{eq:subj-red:comm3:stenvii-move-stenviii}%
  \stEnvii \stEnvMoveAnnot{\ltsCrDe{\mpS}{\roleP}{\roleQ}} \stEnviii%
  = \stEnv[0] \stEnvComp%
  \stEnvMap{\mpChanRole{\mpS}{\roleP}}{\stT} \stEnvComp%
  \stEnv[2] \stEnvComp
  \stEnvMap{\mpChanRole{\mpS}{\roleQ}}{\stStop}%
  &\text{%
    (by \eqref{eq:subj-red:comm3:stenvii-suptype} %
    and \Cref{def:mpst-env-reduction})%
  }%
\end{flalign}

  We can now use $\stEnviii$ to type $\mpPi$:
  \begin{flalign}
    &\label{eq:subj-red:comm3:pi-branch-cont-typing-subst-stenviii}%
    \inference[\iruleMPPar]{%
      \stJudge{\mpEnv}{%
        \stEnv[0] \stEnvComp%
        \stEnvMap{\mpChanRole{\mpS}{\roleP}}{\stT}%
      }{%
        \mpPii%
      }%
      \quad%
      \stJudge{\mpEnv}{%
        \stEnv[2] \stEnvComp \stEnvMap{\mpChanRole{\mpS}{\roleQ}}{\stStop}%
      }{%
        \mpStop{\mpS}{\roleQ}
      }%
    }{%
      \stJudge{\mpEnv}{%
        \stEnviii%
      }{%
        \mpPi%
      }%
    }%
    &\text{%
      (by \eqref{eq:subj-red:comm3:p-branch-typing}, %
      \eqref{eq:subj-red:comm3:p-crash-typing}, %
      \eqref{eq:subj-red:comm3:stenvii-move-stenviii}, %
      and \eqref{eq:subj-red:comm3:p-pi})%
    }%
  \end{flalign}

  We conclude this case by showing that there exists
  some $\stEnvi$ that satisfies the statement:
  \begin{flalign}
    &\label{eq:subj-red:comm3:exists-stenvi}%
    \exists \stEnvi: %
    \stEnv \stEnvMoveMaybeCrashStar \stEnvi \stSub \stEnviii%
    &\text{%
      (by \eqref{eq:subj-red:comm3:stenvii-suptype}, %
      \eqref{eq:subj-red:comm3:stenvii-move-stenviii}, %
      and \Cref{lem:stenv-safe-reduction-sub})%
    }%
    \\%
    &\label{eq:subj-red:comm3:stenvi-safe}%
    \forall \mpS \in \stEnvi: \stEnvSafeSessRolesSP{\mpS}{\rolesR[\mpS]}{\stEnvi}
    &\text{%
      (by \eqref{eq:subj-red:comm3:exists-stenvi} %
      and \Cref{def:mpst-env-safe}, clause \inferrule{\iruleSafeMove})%
    }%
    \\%
    &\nonumber%
    \stJudge{\mpEnv}{%
      \stEnvi%
    }{%
      \mpPi%
    }%
    &\text{%
      (by \eqref{eq:subj-red:comm3:pi-branch-cont-typing-subst-stenviii}, %
      \eqref{eq:subj-red:comm:exists-stenvi}, %
      and \Cref{lem:narrowing})%
    }%
  \end{flalign}

\noindent
Case \inferrule{\iruleMPCrashS}:
\begin{flalign}
  \label{eq:subj-red:crash1:p-pi}%
  &\begin{array}{r@{\;}c@{\;}l}%
    \textstyle%
    \mpP &=&%
    \mpSel{\mpChanRole{\mpS}{\roleP}}{\roleQ}{\mpLab}{
              \mpW
            }{\mpQ}
    \\[1mm]%
    \mpPi &=&%
    \mpBigPar{j \in J}{\mpStop{\mpS[j]}{\roleP[j]}}
    \quad\text{where }%
    \setenum{\mpChanRole{\mpS[j]}{\roleP[j]}}_{j \in J} =
    \fc{\mpP} %
  \end{array}
  &\text{%
    (by inversion of \inferrule{\iruleMPCrashS})%
  }%
  \\
  &\label{eq:subj-red:crash1:p-stenv}
  \stEnv = \stEnv[0]\stEnvComp\stEnv[1]\stEnvComp\stEnv[2]
  \;\;\text{s.t.}\;\;
  \inference[\iruleMPSel]{%
    \begin{array}{l}
      \stEnvEntails{\stEnv[2]}{\mpChanRole{\mpS}{\roleP}}{
        \stIntSum{\roleQ}{}{\stChoice{\stLab}{\stS} \stSeq \stT}
      }
      \\
      \stEnvEntails{\stEnv[1]}{%
        \mpW
      }{\stS}
      \\
      \stJudge{\mpEnv}{
        \stEnv[0] \stEnvComp \stEnvMap{\mpChanRole{\mpS}{\roleP}}{\stT}
      }{
        \mpQ
      }
    \end{array}
  }{
    \stJudge{\mpEnv}{
      \stEnv
    }{
      \mpSel{\mpChanRole{\mpS}{\roleP}}{\roleQ}{\mpLab}{
        \mpW
      }{\mpQ}
    }
  }
  \hspace{-10mm}
  &
  \text{
    (by \eqref{eq:subj-red:crash1:p-pi}, inv.\ \inferrule{\iruleMPCrashS})%
  }
  \\
  \label{eq:subj-red:crash1:p-gamma-fc-not-end}
  &\forall j \in J: \mpChanRole{\mpS[j]}{\roleP[j]} \stNotSub \stEnd
  &\text{%
    (by \eqref{eq:subj-red:crash1:p-pi}, \eqref{eq:subj-red:crash1:p-stenv}, and \Cref{lem:fc-not-end-stop})%
  }
\end{flalign}

Now, notice that:
\begin{flalign}
  &\label{eq:subj-red:crash1:stenv-i}
  \begin{array}{@{}l@{}}
    \stEnv \stEnvMoveMaybeCrashStar \stEnvi
    \;\text{ s.t.}\\
    \quad\forall j \in J: \stEnvApp{\stEnvi}{\mpChanRole{\mpS[j]}{\roleP[j]}} = \stStop
  \end{array}
  \hspace{-0mm}
  &\text{%
    (by \eqref{eq:subj-red:crash1:p-gamma-fc-not-end},
    rule~\inferrule{\iruleTCtxCrash} in \Cref{def:mpst-env-reduction})%
  }
  \\
  &\label{eq:subj-red:crash1:pi-typing}
  \stJudge{\mpEnv}{%
    \stEnvi%
  }{%
    \mpPi%
  }%
  &\text{%
    (by \eqref{eq:subj-red:crash1:p-pi}, %
    \eqref{eq:subj-red:crash1:stenv-i}, %
    and \inferrule{\iruleMPPar} and \inferrule{\iruleMPStop})%
  }%
  \\%
  &\label{eq:subj-red:crash1:stenv-pj-not-reliable}%
  \forall j \in J: \roleP[j] \not\in \rolesR[{\mpS[j]}]
  &\text{%
    (by \eqref{eq:subj-red:crash1:p-pi}
    and \ref{item:subjred:no-reliable-crash})
  }
  \\
  &\label{eq:subj-red:crash1:stenvi-safe}%
  \forall \mpS \in \stEnvi: \stEnvSafeSessRolesSP{\mpS}{\rolesR[\mpS]}{\stEnvi}
  &\text{%
    (by \ref{item:subjred:stenv-safe}, %
    \eqref{eq:subj-red:crash1:stenv-i}, %
    \eqref{eq:subj-red:crash1:stenv-pj-not-reliable}
    and \Cref{def:mpst-env-safe}, clause \inferrule{\iruleSafeMove})%
  }%
\end{flalign}
Hence, we obtain the thesis by \eqref{eq:subj-red:crash1:stenv-i}, \eqref{eq:subj-red:crash1:stenvi-safe} and \eqref{eq:subj-red:crash1:pi-typing}.

\noindent
Case \inferrule{\iruleMPCrashR}: similar to case \inferrule{\iruleMPCrashS} above,
except that we proceed by inversion of \inferrule{\iruleMPCrashR}.

\noindent
Cases \inferrule{\iruleMPRedCtx} and \inferrule{\iruleMPRedCtxCrash}.
The proofs for these two cases are similar.
By inversion of the rule and \Cref{def:mpst-proc-context}, we have to prove the statement in the following sub-cases:
\begin{enumerate}[\inferrule{\iruleMPRedCtx} (1)]
  \item\label{item:subj-red:ctx:par}
    $\mpP = \mpQ \mpPar \mpR$ \;\;and\;\; $\mpPi = \mpQi \mpPar \mpR$ \;\;and\;\; $\mpQ \mpMoveMaybeCrash \mpQi$
  \item\label{item:subj-red:ctx:res}
    $\mpP = \mpRes{\mpSi}{\mpQ}$ \;\;and\;\; $\mpPi = \mpRes{\mpSi}{\mpQi}$ \;\;and\;\; $\mpQ \mpMoveMaybeCrash \mpQi$
  \item\label{item:subj-red:ctx:def}
    $\mpP = \mpDefAbbrev{\mpDefD}{\mpQ}$ \;\;and\;\; $\mpPi = \mpDefAbbrev{\mpDefD}{\mpQi}$ \;\;and\;\; $\mpQ \mpMoveMaybeCrash \mpQi$
\end{enumerate}
Cases \ref{item:subj-red:ctx:par} and \ref{item:subj-red:ctx:def} are easily proved using the induction hypothesis.  Therefore, here we focus on case \ref{item:subj-red:ctx:res}.
\begin{flalign}
  \label{eq:subj-red:ctx:res:p-typing}%
  &
  \exists \stEnv[\mpSi], \rolesRi \;\;\text{s.t.}\;\;
  \inference[\iruleMPResProp]{%
    \begin{array}{@{}l@{}}
      \stEnv[\mpSi] = \setenum{
        \stEnvMap{\mpChanRole{\mpSi}{\roleP}}{\stT[\roleP]}
      }_{\roleP \in I}
      \\
      \stEnvSafeSessRolesSP{\mpSi}{\rolesRi}{\stEnv[\mpSi]}
      \\
      \mpSi \!\not\in\! \stEnv
      \qquad%
      \stJudge{\mpEnv}{
        \stEnv \stEnvComp \stEnv[\mpSi]
      }{
        \mpQ
      }
    \end{array}
  }{%
    \stJudge{\mpEnv}{%
      \stEnv%
    }{%
      \mpP %
    }%
  }%
  &\text{%
    (by \ref{item:subj-red:ctx:res} and inv.\ of \inferrule{\iruleMPResProp})
  }
  \\
  \label{eq:subj-red:ctx:res:stenvi-stenvsi}%
  &\exists \stEnvi, \stEnvi[\mpSi] \;\;\text{s.t.}\;
  \left\{\begin{array}{@{}l@{}}
    \stEnvi[\mpSi] = \setenum{
      \stEnvMap{\mpChanRole{\mpSi}{\roleP}}{\stTi[\roleP]}
    }_{\roleP \in I}
    \\ %
    \mpSi \!\not\in\! \stEnvi
    \\ %
    \stEnv \stEnvMoveMaybeCrashStar \stEnvi
    \\ %
    \stEnv[\mpSi] \stEnvMoveMaybeCrashStar \stEnvi[\mpSi]
    \\[2mm]
    \forall \mpS \in \stEnvi: \stEnvSafeSessRolesSP{\mpS}{\rolesR[\mpS]}{\stEnvi}
    \\ %
    \stJudge{\mpEnv}{\stEnvi \stEnvComp \stEnvi[\mpSi]}{\mpQi}
  \end{array}\right\}
  &\text{%
    (by \eqref{eq:subj-red:ctx:res:p-typing}
    and i.h.)
  }
  \\
  \label{eq:subj-red:ctx:res:qi-nocrash}%
  &
  \forall \roleP \in \rolesRi:
  \not\exists \mpR: \mpQi \equiv \mpR \mpPar \mpStop{\mpS}{\roleP}
  &\text{%
    (by \ref{item:subj-red:ctx:res},
    \ref{item:subjred:no-reliable-crash}
    and \Cref{def:assumption-abiding-reduction})
  }
  \\
  \label{eq:subj-red:ctx:res:stenvii-nostop}%
  &
  \forall \roleP \in \rolesRi:
  \stEnvApp{\stEnv[\mpSi]}{\mpChanRole{\mpSi}{\roleP}} \neq \stStop
  &\text{%
    (by \ref{eq:subj-red:ctx:res:qi-nocrash}
    and \Cref{lem:uncrashed-process-no-crash-type})
  }
  \\
  \label{eq:subj-red:ctx:res:stenvi-stenvisi}%
  &
  \stEnvSafeSessRolesSP{\mpSi}{\rolesRi}{\stEnvi[\mpSi]}
  &\hspace{-20mm}%
  \text{%
    (by \eqref{eq:subj-red:ctx:res:p-typing},
     \eqref{eq:subj-red:ctx:res:stenvi-stenvsi}
     \eqref{eq:subj-red:ctx:res:stenvii-nostop}
     and \Cref{def:mpst-env-safe}, clause \inferrule{\iruleSafeMove})
  }
  \\
  \label{eq:subj-red:ctx:res:pi-typing}%
  &
  \inference[\iruleMPResProp]{%
    \begin{array}{@{}l@{}}
      \stEnv[\mpSi] = \setenum{
        \stEnvMap{\mpChanRole{\mpSi}{\roleP}}{\stTi[\roleP]}
      }_{\roleP \in I}
      \\
      \stEnvSafeSessRolesSP{\mpSi}{\rolesRi}{\stEnvi[\mpSi]}
      \\
      \mpSi \!\not\in\! \stEnvi
      \quad%
      \stJudge{\mpEnv}{
        \stEnvi \stEnvComp \stEnvi[\mpSi]
      }{
        \mpQi
      }
    \end{array}
  }{%
    \stJudge{\mpEnv}{%
      \stEnvi%
    }{%
      \mpPi %
    }%
  }%
  &\text{%
    (by \eqref{eq:subj-red:ctx:res:stenvi-stenvsi},
    \eqref{eq:subj-red:ctx:res:stenvi-stenvisi}
    and \ref{item:subj-red:ctx:res})
  }
\end{flalign}
Hence, we obtain the thesis by \eqref{eq:subj-red:ctx:res:stenvi-stenvsi} and \eqref{eq:subj-red:ctx:res:pi-typing}.
\end{proof}

\lemTypeSafety*
\begin{proof}
  From the hypothesis $\mpP \mpMoveMaybeCrashCheckedStar \mpPi$,
  we know that $\mpP = \mpP[0] \mpMoveMaybeCrashChecked \mpP[1] \mpMoveMaybeCrashChecked \cdots
  \mpMoveMaybeCrashChecked \mpP[n] = \mpPi$ (for some $n$).
  The proof proceeds by induction on $n$.  The base case $n=0$ is immediate: we have $\mpP = \mpPi$,
  hence $\mpPi$ is well-typed --- and since the term $\mpErr$ is not typeable, $\mpPi$ cannot contain such a term.
  In the inductive case $n = m+1$, we know (by the induction hypothesis) that $\mpP[m]$ is well-typed,
  and we apply \Cref{lem:subject-reduction} to conclude that $\mpP[m+1] = \mpPi$ is also well-typed and has no $\mpErr$ subterms.
\end{proof}

%% file: proofs/session-fidelity.tex
\section{Proofs for Session Fidelity and Process Properties}
\label{sec:proofs:session-fidelity}
\label{sec:proofs:proc-properties}

\lemSessionFidelity*
\begin{proof}
  The proof structure is similar to Thm.\@ 5.4 in \cite{POPL19LessIsMore}: by
  induction on the derivation of the reduction of $\stEnv$,
  we infer the contents of $\stEnv$ and then the shape of $\mpP$ and its
  sub-processes $\mpP[\roleP]$, showing that they can mimic the reduction of
  $\stEnv$.
  The main differences \wrt \cite{POPL19LessIsMore} are that
  \begin{enumerate}[(1)]
    \item we now account for crashed session endpoints with type $\stStop$; %
      and
    \item the proof %
      covers more cases, as it now includes crash detection reductions, and outputs to crashed
      processes.
  \end{enumerate}

    Compared to the proof of Thm.\@ 5.4 in \cite{POPL19LessIsMore},
    we have the following additional cases to consider
    when a crash is detected, or a selection targets a crashed process.
    \begin{itemize}
    \item case $\stEnv \stEnvMoveAnnot{\ltsCrDe{\mpS}{\roleP}{\roleQ}} \stEnvi$.\quad
      In this case, the process $\mpP[\roleP]$  playing role $\roleP$ in session $\mpS$ is a branching on $\mpChanRole{\mpS}{\roleP}$ from $\roleQ$ (possibly within a process definition) including crash detection; therefore, $\mpP[\roleP]$ can correspondingly detect that the channel endpoint $\mpChanRole{\mpS}{\roleQ}$ is crashed, by rule \inferrule{\iruleMPRedCommD} in \Cref{fig:mpst-pi-semantics} (possibly after a finite number of transitions under rule \inferrule{\iruleMPRedCall}).
      The resulting continuation process $\mpPi$ is typed by $\stEnvi$; %
    \item case $\stEnv \stEnvMoveAnnot{\ltsSendRecv{\mpS}{\roleP}{\roleQ}{\stLab}} \stEnvi$ and $\stEnvApp{\stEnv}{\mpChanRole{\mpS}{\roleQ}} = \stStop$.\quad
      In this case, the process $\mpP[\roleP]$ playing role $\roleP$ in session $\mpS$ is a selection on $\mpChanRole{\mpS}{\roleP}$ towards $\roleQ$ (possibly within a process definition); therefore, $\mpP$ could correspondingly reduce to $\mpPi$ by sending either a basic value $\mpV$ or a channel endpoint $\mpChanRole{\mpSi}{\rolePi}$ (possibly after a finite number of transitions under rule \inferrule{\iruleMPRedCall}) to the crashed channel endpoint $\mpChanRole{\mpS}{\roleQ}$.  We have two possible cases for the communication reduction leading from $\mpP$ to $\mpPi$:
      \begin{itemize}
        \item rule \inferrule{\iruleMPRedCommE} in \Cref{fig:mpst-pi-semantics}, with $\mpChanRole{\mpSi}{\rolePi}$ crashed in $\mpPi$.  This case is impossible: in fact, by the side condition of the typing rule \inferrule{\iruleMPSel} (\Cref{fig:mpst-rules}), we must have $\stEnvApp{\stEnv[\roleP]}{\mpChanRole{\mpSi}{\rolePi}} \stNotSub \stEnd$ --- and this would contradict the assumption that $\mpP[\roleP]$ only plays role $\roleP$ in session $\mpS$, by $\stEnv[\roleP]$;
        \item rule \inferrule{\iruleMPRedCommEGround} in \Cref{fig:mpst-pi-semantics}.  In this case, we have $\stEnv[\roleP] \stEnvMoveOutAnnot{\roleP}{\roleQ}{\stChoice{\stLab}{\tyGround}} \stEnvi[\roleP]$ (for some $\stLab$ and basic type $\tyGround$), and the continuation process $\mpPi$ is typed by the resulting $\stEnvi$.
      \qedhere
      \end{itemize}
    \end{itemize}
\end{proof}

\Cref{lem:single-session-persistent} below says that if a process $\mpP$ satisfies the assumptions of session fidelity (\Cref{lem:session-fidelity})
then all its reductums will satisfy such assumptions, too. This means that if $\mpP$ enjoys session fidelity, then all its reductums enjoy session fidelity, too.

\begin{restatable}{proposition}{lemSingleSessionPersistent}%
    \label{lem:single-session-persistent}%
    Assume\, $\stJudge{\mpEnvEmpty\!}{\!\stEnv}{\!\mpP}$, %
    where %
    $\stEnv$ is ($\mpS; \rolesR$)-safe, %
    \,$\mpP \equiv \mpBigPar{\roleP \in I}{\mpP[\roleP]}$, %
    \,and\, $\stEnv = \bigcup_{\roleP \in I}\stEnv[\roleP]$ %
    such that, for each $\mpP[\roleP]$, %
    we have\, $\stJudge{\mpEnvEmpty\!}{\stEnv[\roleP]}{\!\mpP[\roleP]}$.
    \,Further, assume that each $\mpP[\roleP]$
    is either\, $\mpNil$ (up to $\equiv$), %
    or only plays $\roleP$ in $\mpS$, by $\stEnv[\roleP]$. %
    Then,\, $\mpP \mpMoveMaybeCrash[\mpS;\rolesR] \mpPi$
    \,implies\, $\exists \stEnvi$ %
    such that\, %
    $\stEnv \!\stEnvMoveMaybeCrashStar\! \stEnvi$ %
    \,and\, %
    $\stJudge{\mpEnvEmpty\!}{\!\stEnvi}{\mpPi}$, %
    \;with\; %
    $\stEnvi$ ($\mpS; \rolesR$)-safe, %
    \,$\mpPi \equiv \mpBigPar{\roleP \in I}{\mpPi[\roleP]}$, %
    \,and\, $\stEnvi = \bigcup_{\roleP \in I}\stEnvi[\roleP]$ %
    such that, for each $\mpPi[\roleP]$, ;
    we have\, $\stJudge{\mpEnvEmpty\!}{\stEnvi[\roleP]}{\!\mpPi[\roleP]}$;
    \,furthermore, each $\mpPi[\roleP]$
    is $\mpNil$ (up to $\equiv$),
    or only plays $\roleP$ in $\mpS$, by $\stEnvi[\roleP]$.%
\end{restatable}
\begin{proof}
  Straightforward from the proof of \Cref{lem:subject-reduction}, which accounts for all possible transitions from $\mpP$ to $\mpPi$, and in all cases yields the desired properties for its typing context $\stEnvi$.
\end{proof}

\lemProcessPropertiesVerif*
\begin{proof}
  \textbf{Deadlock-freedom}\quad
  Consider any $\mpPi$ such that $\mpP \!\mpMoveMaybeCrashCheckedStar\! \mpNotMoveP{\mpPi}$ \,with\, $\mpP = \mpP[0] \!\mpMoveMaybeCrashChecked\! \mpP[1] \!\mpMoveMaybeCrashChecked\! \cdots \!\mpMoveMaybeCrashChecked\! \mpP[n] = \mpNotMoveP{\mpPi}$ (for some $n$)
  with each reduction $\mpP[i] \!\mpMoveMaybeCrashChecked\! \mpP[i+1]$ ($i \!\in\! 0..n\!-\!1$)
  satisfying~%
  \Cref{def:assumption-abiding-reduction}.
  By \Cref{lem:single-session-persistent}, we know that each $\mpP[i]$ is well-typed
  and its typing context $\stEnv[i]$ is such that $\stEnv \stEnvMoveMaybeCrashStar[\mpS; \rolesR] \stEnv[i]$;
  moreover, $\mpP[i]$ satisfies the single-session requirements of \Cref{lem:session-fidelity}.
  Now observe that, since the process $\mpP[n] = \mpNotMoveP{\mpPi}$ cannot reduce further (except by crashing),
  by the contrapositive of \Cref{lem:session-fidelity}
  we obtain $\stEnvNotMoveP{\stEnv[n]}$;
  and since $\stEnv$ is ($\mpS;\rolesR$)-deadlock-free by hypothesis, by
  \Cref{def:typing-ctx-properties} (item \textbf{Deadlock-freedom}%
  )
  we have
  $\forall \mpChanRole{\mpS}{\roleP} \!\in\! \stEnv[n]$: $\stEnvApp{\stEnv[n]}{\mpChanRole{\mpS}{\roleP}} \!\stSub\! \stEnd$ or $\stEnvApp{\stEnv[n]}{\mpChanRole{\mpS}{\roleP}} \!=\! \stStop$ or $\stEnvApp{\stEnv[n]}{\mpChanRole{\mpS}{\roleP}} \!\stSub\! \stInNB{\roleQ}{\stCrashLab}{}{} \stSeq \stTi$.
  Therefore, by inversion of typing,
  we have $\mpPi \equiv \mpNil \mpPar \mpBigPar{i \in I}{\mpStop{\mpS}{\roleP[i]}} \mpPar \mpBigPar{j \in J}{(\mpDefAbbrev{\mpDefD[j,1]}{\ldots\mpDefAbbrev{\mpDefD[j,n_j]}{\mpBranchSingle{\mpChanRole{\mpS}{\roleP[j]}}{\roleQ[j]}{\mpLabCrash}{}{\mpQi[j]}}})}$
  --- which (by \Cref{def:proc-properties}, item~\textbf{Deadlock-freedom}) is the thesis.

  \textbf{Terminating}\quad
  We know that
  $\exists j$ finite such that, $\forall n \ge j$, $\stEnv = \stEnv[0] \stEnvMoveMaybeCrash[\mpS;\rolesR] \stEnv[1] \stEnvMoveMaybeCrash[\mpS;\rolesR] \cdots \stEnvMoveMaybeCrash[\mpS;\rolesR] \stEnv[n]$ implies $\stEnvNotMoveP{\stEnv[n]}$;
  moreover,  since $\stEnv$ is $(\mpS;\rolesR)$-deadlock-free (by \Cref{def:typing-ctx-properties}, item~\ref{item:typing-ctx-properties:term}),
  whenever $\stEnv \stEnvMoveMaybeCrash[\mpS;\rolesR] \stEnv[1] \stEnvMoveMaybeCrash[\mpS;\rolesR] \cdots \stEnvMoveMaybeCrash[\mpS;\rolesR] \stEnvNotMoveP{\stEnv[n]}$ (for any $n$), then $\forall \mpChanRole{\mpS}{\roleP} \!\in\! \stEnv[n]$: $\stEnvApp{\stEnv[n]}{\mpChanRole{\mpS}{\roleP}} \!\stSub\! \stEnd$ or $\stEnvApp{\stEnv[n]}{\mpChanRole{\mpS}{\roleP}} \!=\! \stStop$ or $\stEnvApp{\stEnv[n]}{\mpChanRole{\mpS}{\roleP}} \!\stSub\! \stInNB{\roleQ}{\stCrashLab}{}{} \stSeq \stTi$.
  Considering all the possible sequences of reductions of $\mpP$, we have the following cases:
  \begin{enumerate}[(1)]
  \item $\mpP \!=\! \mpP[0] \!\mpMoveMaybeCrashChecked\! \mpP[1] \!\mpMoveMaybeCrashChecked\!%
    \cdots \!\mpMoveMaybeCrashChecked\! \mpP[m]$ and $\mpNotMoveP{\mpP[m]}$ (for some $m$).
    Since $\stEnv$ is $(\mpS;\rolesR)$-deadlock-free (by \Cref{def:typing-ctx-properties}, item~\ref{item:typing-ctx-properties:term}),
    we obtain that, by item~\textbf{Deadlock-freedom} above, $\mpP[m] \equiv \mpNil \mpPar \mpBigPar{i \in I}{\mpStop{\mpS}{\roleP[i]}} \mpPar \mpBigPar{j \in J}{(\mpDefAbbrev{\mpDefD[j,1]}{\ldots\mpDefAbbrev{\mpDefD[j,n_j]}{\mpBranchSingle{\mpChanRole{\mpS}{\roleP[j]}}{\roleQ[j]}{\mpLabCrash}{}{\mpQi[j]}}})}$;
  \item there is an infinite sequence of reductions $\mpP \!=\! \mpP[0] \!\mpMoveMaybeCrashChecked\! \mpP[1] \!\mpMoveMaybeCrashChecked\! \mpP[2] \!\mpMoveMaybeCrashChecked\! \cdots$ such that $\forall i \ge 0: \mpMoveP{\mpP[i]}$.
    This case is impossible.  In fact, if we admit it, by the proof of \Cref{lem:subject-reduction}
    we have two possibilities (both leading to a contradiction):
    \begin{itemize}
      \item  there is an infinite sequence of typing context reductions $\stEnv = \stEnv[0] \stEnvMoveMaybeCrash[\mpS;\rolesR] \stEnv[1] \stEnvMoveMaybeCrash[\mpS;\rolesR] \stEnv[2] \stEnvMoveMaybeCrash[\mpS;\rolesR] \cdots$ to type each $\mpP[i]$ with a suitable $\stEnv[j]$ (with $j \leq i$); moreover, for each such $\stEnv[j]$, we have $\stEnvMoveP{\stEnv[j]}$ (otherwise, $\stEnv$ would not be $(\mpS;\rolesR)$-deadlock-free, hence by \Cref{def:typing-ctx-properties}, item~\ref{item:typing-ctx-properties:term}, it would also not be $(\mpS;\rolesR)$-terminating). But then, we contradict the hypothesis that $\exists j$ finite such that, $\forall n \ge j$, $\stEnv = \stEnv[0] \stEnvMoveMaybeCrash[\mpS;\rolesR] \stEnv[1] \stEnvMoveMaybeCrash[\mpS;\rolesR] \cdots \stEnvMoveMaybeCrash[\mpS;\rolesR] \stEnv[n]$ implies $\stEnvNotMoveP{\stEnv[n]}$;
      \item there are infinitely many processes reductums that can be typed by a same $\stEnv[i]$.  By the proof of \Cref{lem:subject-reduction}, this can only happen in the following ways (all leading to a contradiction):
      \begin{itemize}
        \item firing infinitely many reduction within some restricted session --- which would contradict the hypothesis that each parallel sub-process of $\mpP$ only plays one role in session $\mpS$ (\Cref{def:unique-role-proc});
        \item performing infinitely many process calls by firing rule rule \inferrule{\iruleMPRedCall} (in \Cref{fig:mpst-pi-semantics}) infinitely many times, without other message transmissions or error detection reductions (which would cause the typing context to reduce). However, this would contradict the hypothesis that $\mpP$ has guarded definitions (\Cref{lem:guarded-definitions});
        \item having a recursive protocol in $\stEnv$ such that $\stEnv[i] = \stEnv[i+1]$. This would lead to the same contradiction addressed in the first case above.
      \end{itemize}
    \end{itemize}
  \end{enumerate}
  Summing up, all possible sequences of reductions of $\mpP$ are finite, and they are all deadlock-free --- which is the thesis.

  \textbf{Never-Terminating}\quad
    By hypothesis and \Cref{def:typing-ctx-properties} (item~\ref{item:typing-ctx-properties:nterm}), we know that
    $\stEnv \stEnvMoveMaybeCrashStar[\mpS;\rolesR] \stEnvi$ implies $\stEnvMoveP{\stEnvi}$.
    By contradiction, assume that $\mpP$ is \emph{not} never-terminating, \ie
    $\exists \mpPi$ such that $\mpP \!=\! \mpP[0] \!\mpMoveMaybeCrashChecked\! \mpP[1] \!\mpMoveMaybeCrashChecked\! \cdots \!\mpMoveMaybeCrashChecked\! \mpP[n] = \mpNotMoveP{\mpPi}$.
    By \Cref{lem:single-session-persistent}, we know that each $\mpP[i]$ is well-typed
    and its typing context $\stEnv[i]$ is such that $\stEnv \stEnvMoveMaybeCrashStar[\mpS; \rolesR] \stEnv[i]$;
    moreover, $\mpP[i]$ satisfies the single-session requirements of \Cref{lem:session-fidelity}.
    Now observe that, since the process $\mpP[n] = \mpNotMoveP{\mpPi}$ cannot reduce further (except by crashing),
    by the contrapositive of \Cref{lem:session-fidelity}
    we obtain $\stEnvNotMoveP{\stEnv[n]}$
    --- but this contradicts the hypothesis that $\stEnv$ is $(\mpS;\rolesR)$-never-terminating.  Therefore, we conclude that $\mpP$ is never-terminating.

  \textbf{Live}\quad
    By contradiction, assume that $\mpP$ is \emph{not} live.
    Since (by hypothesis) each parallel component of $\mpP$ only plays one role $\roleP$ in session $\mpS$,
    this means that there are $\mpPi, \mpCtx, \mpQ$ such that
    $\mpP = \mpP[0] \!\mpMoveMaybeCrashChecked\! \mpP[1] \!\mpMoveMaybeCrashChecked\! \cdots \!\mpMoveMaybeCrashChecked\! \mpP[n] = \mpPi \!\equiv\! \mpCtxApp{\mpCtx}{\mpQ}$ where either:
    \begin{itemize}
    \item%
      $\mpQ = \mpSel{\mpChanRole{\mpS}{\roleP}}{\roleQ}{\mpLab}{\mpW}{\mpQi}$ %
      (for some $\mpLab, \mpW, \mpQi$), %
      \;and\; %
      $\not\exists \mpCtxi$: %
      $\mpPi \!\mpMoveStar\! \mpCtxApp{\mpCtxi}{\mpQi}$.\quad%
      By \Cref{lem:single-session-persistent}, we know that each $\mpP[i]$ is well-typed
      and its typing context $\stEnv[i]$ is such that $\stEnv \stEnvMoveMaybeCrashStar[\mpS; \rolesR] \stEnv[i]$;
      moreover, each $\mpP[i]$ satisfies the single-session requirements of \Cref{lem:session-fidelity}.
      Therefore, $\mpPi$ satisfies the single-session requirements of \Cref{lem:session-fidelity},
      and is typed by some $\stEnvi$ such that $\stEnv \stEnvMoveMaybeCrashStar[\mpS; \rolesR] \stEnv[i]$
      --- hence, by inversion of typing, $\mpQ$ is typed by some $\stEnvi[\roleP]$ (part of $\stEnvi$)
      where $\stEnvApp{\stEnvi[\roleP]}{\mpChanRole{\mpS}{\roleP}}$ is a (possibly recursive) internal choice
      towards $\roleQ$, including a choice $\stChoice{\stLab}{\stS}$ (where $\stS$ types the message payload $\mpW$). Therefore, we have $\stEnvMoveAnnotP{\stEnvi}{\stEnvOutAnnot{\roleP}{\roleQ}{\stChoice{\stLab}{\stS}}}$. Now, recall that (for the sake of the proof by contradiction) we are assuming that
      no sequence of reductions of $\mpPi$ can fire the top-level selection of $\mpQ$;
      this means that no parallel component of $\mpPi$ ever exposes an external choice by role $\roleQ$
      including message label $\mpLab$; correspondingly, there is at least
      one fair and non-crashing path beginning with $\stEnvi$
      (yielded by \Cref{lem:subject-reduction}) that never fires a transmission label $\ltsSendRecv{\mpS}{\roleP}{\roleQ}{\stLabi}$ (for any $\stLabi$).
      But then, such a fair path starting from $\stEnvi$ is not live, hence
      (by \Cref{def:typing-ctx-properties}, item~\ref{item:typing-ctx-properties:live}) we obtain that $\stEnv$ is \emph{not} live
      --- contradiction;

    \item%
      $\mpQ = \mpBranch{\mpChanRole{\mpS}{\roleP}}{\roleQ}{i \in I}{\mpLab[i]}{x_i}{\mpQi[i]}{}$ %
      (for some $I$, $\mpLab[i], \mpFmt{x_i}, \mpQi[i]$ such that either $|I| \neq 1$ or $\forall i \!\in\! I: \mpLab[i] \!\neq\! \mpLabCrash$), %
      \;and\; %
      $\not\exists \mpCtxi, k \!\in\! I, \mpW$:\, %
      $\mpPi \mpMoveStar \mpCtxApp{\mpCtxi}{\mpQi[k]\subst{x_k}{\mpW}}$. %
      The proof is similar to the previous case, and reaches a similar contradiction.
    \end{itemize}
    Summing up, we have shown that if we assume $\mpP$ not live, we reach a contradiction.
    Therefore, we conclude that $\mpP$ is live.
\end{proof}